\documentclass[11pt]{article}
\usepackage{amssymb}
\usepackage{amsmath}
\usepackage{amsmath}
\usepackage{amsthm}
\usepackage{amsfonts}
\usepackage{graphicx}
\usepackage{enumerate}
\usepackage{bbm} 
\usepackage{dsfont}
\usepackage[authoryear,longnamesfirst]{natbib}
\usepackage{multirow}
\usepackage{multicol}
\usepackage{tikz}
\numberwithin{equation}{section}
\newtheorem{theorem}{Theorem}
\newtheorem{algorithm}{Algorithm}
\newtheorem{assumption}{Assumption}

\newtheorem{corollary}{Corollary}

\newtheorem{example}{Example}

\newtheorem{lemma}{Lemma}
\theoremstyle{definition}

\DeclareMathOperator{\E}{\text{E}}

\DeclareMathOperator{\var}{\text{Var}}

\newcommand{\supp}{\rm{supp}}
\renewcommand{\hat}{\widehat}
\renewcommand{\tilde}{\widetilde}

\newcommand{\sumIJ}{\sum_{(i,j)\in \overline{I_k^2}}}
\newcommand{\sumln}{\sum\limits_{l \in [n]}}

\newcommand{\sumk}{\sum\limits_{k \in [K]}}
\newcommand{\tE}{\text{E}}
\newcommand{\Ep}{{\rm E}_{P}}
\newcommand{\Epn}{{\rm E}_{P_N}}
\newcommand{\sumij}{\sum\limits_{(i,j) \in \overline{[N]^2}} }
\newcommand{\sumijn}{\sum\limits_{(i,j) \in \overline{[n]^2}} }
\newcommand{\sumijN}{\sum\limits_{(i,j) \in \overline{[N]^2}} }

\newcommand{\Enk}{\mathbb{E}_{I_k}}
\newcommand{\Ikc}{\overline{I_k^{c,2}}}

\newcommand{\Real}{\mathbbm R}

\newcommand{\Gn}{\mathbb{G}_{n}}
\newcommand{\Gnk}{\mathbbm{G}_{N,K}}
\newcommand{\N}{\mathbb{N}}

\newcommand{\Hn}{\text H_n}
\newcommand{\Hnk}{\text H_{N}}

\newcommand{\T}{\mathcal T}
\renewcommand{\E}{\mathcal E}
\newcommand{\Op}{O_P}
\newcommand{\Opn}{O_{P_N}}
\newcommand{\op}{o_P}
\newcommand{\opn}{o_{P_N}}

\newcommand{\V}{{\rm V}}

\newcommand{\Cov}{{\rm Cov}}

\newcommand{\1}{\mathbbm 1}
\newcommand{\calS}{\mathcal{S}}
\newcommand{\calG}{\mathcal{G}}
\newcommand{\EN}{\mathbb{E}_N}
\newcommand{\calH}{\mathcal{H}}
\newcommand{\calF}{\mathcal{F}}
\newcommand{\NN}{\overline{[N]^2}}
\begin{document}

\title{Dyadic Double/Debiased Machine Learning for \\Analyzing Determinants of Free Trade Agreements\thanks{\setlength{\baselineskip}{4.0mm}\footnotesize\setlength{\baselineskip}{4.4mm}  We benefited from very comments by participants in AMES 2021, ESAM 2021, ESEM 2021, IAAE 2021, NASMES 2021, and New York Camp Econometrics XV. H. Chiang's research is supported by the Office of the Vice Chancellor for Research and Graduate Education at the University of Wisconsin-Madison with funding from the Wisconsin Alumni Research Foundation. \bigskip}}
\author{
Harold D. Chiang\thanks{\setlength{\baselineskip}{4.0mm}Harold D. Chiang: hdchiang@wisc.edu. Department of Economics, University of Wisconsin-Madison, William H. Sewell Social Science Building 1180 Observatory Drive 
	Madison, WI 53706-1393, USA\bigskip} 
\qquad 
Yukun Ma\thanks{\setlength{\baselineskip}{4.0mm}Yukun Ma: yukun.ma@vanderbilt.edu. Department of Economics, Vanderbilt University, VU Station B \#351819, 2301 Vanderbilt Place, Nashville, TN 37235-1819, USA\bigskip} 
\qquad 
Joel B. Rodrigue\thanks{\setlength{\baselineskip}{4.0mm}Joel B. Rodrigue: joel.b.rodrigue@vanderbilt.edu. Department of Economics, Vanderbilt University, VU Station B \#351819, 2301 Vanderbilt Place, Nashville, TN 37235-1819, USA\bigskip}
\qquad 
Yuya Sasaki\thanks{\setlength{\baselineskip}{4.0mm}Yuya Sasaki: yuya.sasaki@vanderbilt.edu. Department of Economics, Vanderbilt University, VU Station B \#351819, 2301 Vanderbilt Place, Nashville, TN 37235-1819, USA\bigskip}
}

\date{}

\maketitle

\begin{abstract}\setlength{\baselineskip}{5.7mm}
This paper presents novel methods and theories for estimation and inference about parameters in econometric models using machine learning for nuisance parameters estimation when data are dyadic.
We propose a dyadic cross fitting method to remove over-fitting biases under arbitrary dyadic dependence.
Together with the use of Neyman orthogonal scores, this novel cross fitting method enables root-$n$ consistent estimation and inference robustly against dyadic dependence.
We illustrate an application of our general framework to high-dimensional network link formation models.
With this method applied to empirical data of international economic networks, we reexamine determinants of free trade agreements (FTA) viewed as links formed in the dyad composed of world economies. We document that standard methods may lead to misleading conclusions for numerous classic determinants of FTA formation due to biased point estimates or standard errors which are too small.
\end{abstract}

\newpage 
\section{Introduction}

Empirical research investigating network formation and outcomes spans the literature covering labor economics, industrial organization, macroeconomics, development, and international trade, among others. 
In a large fraction of these studies, the process through which economic entities are determined -- e.g., free/preferential trade agreements, friendships, and financial relationships -- concern dyadic/network link formation models. 
While recent methodological advances have increased the popularity of empirical network formation models (see \citet*{graham2019dyadic} and \citet{graham2019econometric} for surveys), the large majority of existing studies focus on parsimonious (low-dimensional) model specifications.
The increasing availability of `big data' should, in principle, allow researchers to consider increasingly rich specifications and uncover new insights into the nature of network formation. 
Rather, researchers are hamstrung by two competing limitations. 
On one hand, there is no existing econometric methodology that allows for both dyadic robust inference and high-dimensional link formation specifications.  {  In the presence of significant data heterogeneity in the policy variable of interest, one needs to control a large number of covariates.
However, simply estimating a rich specification across a wide set of covariates is not yet possible.  }
Moreover, we document that common off-the-shelf approaches, such as the conventional double/debiased machine learning for i.i.d. sampling, are likely to result in both biased estimates and misleading standard errors under dyadic dependence. 
On the other hand, the most common empirical approach, reducing the dimensionality of the network formation model by assumption, risks model mis-specification, is a likely source of estimation bias and/or misleading inference, and reduces the set of questions researchers can investigate.
In light of the increasing availability of big data today, we advance state-of-the-art approaches to estimating network models by developing novel methods for root-$n$ consistent estimation and inference for high-dimensional dyadic regressions and high-dimensional dyadic/network link formation models. 

Our work builds on the literature which suggests using (near) Neyman orthogonal scores to accommodate possibly slow convergence rates of general machine learners \citep*[e.g.,][]{BCK15,BCCW18,CCDDHNR18,CEINR18} in high-dimensional regression estimation (and more generally high-dimensional Z-estimation). In particular, our approach extends the branch of research which employs cross-fitting, in conjunction with orthogonal scores, to mitigate over-fitting biases. This combined method is referred to as the double/debiased machine learning \citep*[DML,][]{CCDDHNR18}. Although standard cross fitting requires independent sampling \citep{CCDDHNR18}, \citet*{CKMS2019} develop a multiway cross fitting algorithm to extend the DML to multiway cluster dependent data with linear scores (their results, however, does not cover estimators defined by nonlinear scores, such as Logit). This recent contribution has paved the way for machine learning to be applied to cross-sectional dependent data, but, unfortunately, does not provide a result which allows for the estimation of network models in the presence of dyadic dependence. Dyadic dependence is different from (multiway) cluster dependence,\footnote{Specifically, while multiway clustered data are represented by separately exchangeable arrays, dyadic data are represented by jointly exchangeable arrays -- see Section 7 in \citet*{Kallenberg2006}. The mathematical characterizations of these two types of exchangeable arrays are different, and the separate exchangeability does not imply the joint exchangeability.} and hence there is no guarantee that these existing cross fitting algorithms will work for dyadic regressions and dyadic/network link formation models.
In this light, we propose a novel cross fitting algorithm that is effective under dyadic dependence. Our dyadic machine learning approach and asymptotic theories guarantee that estimation and inference based on this method are robust against arbitrary dyadic dependence.

In this paper, we demonstrate the importance of addressing dyadic dependence and model specification by reexamining the determinants of free trade networks, a classic network setting. 
Accounting for both dyadic dependence and high-dimensional controls, we reconfirm the two important theoretical implications suggested by the international trade literature: namely, (A) a greater distance between economies makes an FTA less beneficial \citep[e.g.,][]{krugman1991move,frankel1993continental,frankel1995trading,frankel1996regional,baier2004economic} and thus makes an FTA less likely to be formed; and (B) larger sizes of economies make an FTA more beneficial \citep{krugman1998comment,baier2004economic} and thus make an FTA more likely to be formed.  In contrast to the workhorse empirical model, we find that after accounting for dyadic dependence and high-dimensional controls suggest that differences in country-pair factor endowments are not an important determinant of free trade agreements. We demonstrate that traditional approaches lead to fragile estimates: researchers may be misled into concluding that factor endowments differences either encourage or discourage free trade agreements depending on which common estimation approach is employed. Further, estimates produced by our method turned out to differ from those of the simple logistic regression or the conventional double/debiased machine learning not accounting for dyadic sampling for all key determinants of free trade agreements. Indeed, our analysis suggests that the landmark empirical work in this area overstates the importance of bilateral country size to FTA network formation. Allowing for arbitrary dyadic dependence, the proposed method implies less statistical significance.
That being said, our results still support the aforementioned theoretical predictions about the determinants of free trade networks despite the inflated standard errors for robustly accounting for the dyadic dependence as well as high-dimensionality.

{\bf 1.1. Relation to the Econometrics Literature:}
The history of statistics on dyadic data dates back at least to 1970s, when
\citet*{holland1976local} derive moments of dyadic sums and
\citet*{silverman1976limit} derives asymptotic normality for the sum of jointly exchangeable dissociated arrays.
\citet*{eagleson1978limit} further demonstrated the asymptotic normal mixture for the sum of a weakly exchangeable array.

Explicit accounts for dyadic dependence in econometrics arose in the context of gravity models, and the use of fixed effects was recommended to control for such dependence \citep*[see][]{matyas1997proper,matyas1998gravity}.
\citet*{cameron2005estimation} point out the importance of controlling for dyadic clustering, and develop a FGLS estimation method accounting for the dyadic cluster dependence.
\citet*{fafchamps2007formation} propose dyadic cluster robust variance estimators for the OLS and logit. 
\citet*{cameron2014robust} generalize the dyadic cluster robust variance estimator for GMM and M-estimation frameworks as well as others cases.
\citet*{aronow2015cluster} also discuss consistent estimation of the dyadic cluster robust variance estimators.
Along with most, if not all, other papers that develop inference theories for dyadic data, we take advantage of the forms of the dyadic cluster robust variance formulas developed in these pioneering papers.

The more recent econometrics literature includes
\citet*{tabord2019inference} who studies the asymptotic behavior of the t-statistic based on the dyadic cluster robust variance estimator of \citet*{fafchamps2007formation};
\citet*{davezies2019empirical} who study the asymptotic behavior of empirical processes and their bootstrap counterparts for dyadic data;
\citet*{graham2019kernel} who propose nonparametric density estimation for dyadic data, show that the convergence rate of the stochastic part of the estimator is the square root of the number of nodes, and derive the asymptotic distribution of the estimator (see also \citealp{graham2021minimax} for extension to local regression estimation); and \citet{chiang2020inference} who develop methods of inference for high-dimensional parameters.
\citet*{graham2019dyadic,graham2020network,graham2020sparse} provide reviews on the asymptotic distribution and variance estimation in dyadic data in the context of parametric models, along with other closely related topics.

Also related to, but different from, the literature on dyadic data is the set of studies which investigate multiway clustered data 
\citep*[e.g.,][]{cameron2011robust,thompson2011simple,cameron2015practitioner,menzel2018bootstrap,davezies2018asymptotic,mackinnon2019wild,CKMS2019}.  The robust variance formulas exposited in this literature are related to those relevant for dyadic cases, but are sufficiently different that they cannot be applied in settings with dyadic dependence. In particular, the structure of cluster dependent data (also known as separately exchangeable arrays) is related to, but mathematically different from, the structure of dyadic data (also known as jointly exchangeable arrays) -- see \citet*{Kallenberg2006}. Consequently, methods and theories developed in this literature will not apply to the dyadic setting.

Turning to the machine learning literature, as we have already emphasized, we use (near) Neyman orthogonal scores to accommodate possibly slow convergence rates of various machine learners following \citet*{BCH14review,BCK15,farrell2015robust,BCCW18,CCDDHNR18,CEINR18,farrell2021deep}.
In developing the dyadic link formation models, we also extend the framework of \citet*{BelloniChernozhukovWei2016} and \citet*{BCCW18} to deal with dyadic dependent data and obtain convergence rates for high-dimensional lasso logit under dyadic dependence, which we in turn use as sufficient conditions to apply our general dyadic machine learning theory in order to establish asymptotic theories for high-dimensional dyadic link formation models.

 The proposed dyadic cross-fitting algorithm sheds light on handling dyadic networks in cross-validation literature. 
There are several network-oriented cross-fitting and cross-validation procedures in econometrics and statistics. 
In the study of stochastic block models, \cite{chen2018network} propose network cross-validation algorithms for detecting the number of communities. 
\cite{li2020network} propose cross-validation algorithms for general tool for model selection and parameter tuning that rely on using a subset of node pairs and apply a low-rank matrix completion algorithm to obtain a predicted adjacency matrix. It is thus fundamentally different from ours.
\cite{viviano2019policy} proposes a network cross-fitting algorithm for estimating treatment allocation rules
under network interference. It is different from our algorithm as its unit of observations is a node while ours is a dyad. In addition, our algorithm allows for fully connected networks that are important for international trade applications.

{\bf 1.2. Relation to the Trade Literature:}
Understanding the formation of free trade networks across countries dates back to at least \citet*{Viner1950} and its history is fraught with debate. A rich theoretical literature emerged over time elucidating complex economics and political trade-offs associated with free trade agreement formation. \citet*{Levy1997} and \citet*{Krishna1998} pioneered a set of research aimed at understanding the nature of preferential trade in a multilateral trade network.\footnote{\citet*{GrossmanHelpman1995} analyze the political-economy determinants of FTAs but do not consider the implications of FTAs for the multilateral trade system.} Subsequent research investigates network formation and stability \citep*{FK05,GJ06,FurusawaKonishi2007} and characterizes the differential impact that FTAs have on countries excluded and included from FTAs \citep*{KennanRiezman1990,BagwellStaiger1999,BondRiezmanSyropoulos2004,AghionAntrasHelpman2007}.


Country asymmetries, either economic \citep*{FurusawaKonishi2007,SaggiYildiz2010,SaggiWoodlandYildiz2013,Lake2017} or political \citep*{Ornelas2005,StoyanovYildiz2015}, are understood as central to the willingness of any pair of countries to form a stable free trade agreement and their optimal responses policy change within and outside of their region. Likewise, dynamic considerations, such as the likelihood of future FTAs, are expected to encourage or deter firms from joining FTAs in the present \citep*{McLaren2002,MissiosSaggiYildiz2016,Lake2017,LakeRoy2017,LakeNkenYildiz2018,Lake2019}.  In other words, the determinants of FTAs systematically vary across the trading network with country characteristics and the evolution of the trading system.

By comparison there is a dearth of empirical results.  \citeauthor{baier2004economic}'s (\citeyear{baier2004economic}) seminal article takes up \citeauthor{krugman1993regionalism}'s (\citeyear{krugman1993regionalism}) call to empirically investigate the determinants of free trade regions.\footnote{As \citet*{{krugman1993regionalism}} writes: ``To make any headway, one must either get into detailed empirical work, or make strategic simplifications and stylizations that one hopes do not lead one too far astray. Obviously detailed empirical work is the right direction...''}  They identify a parsimonious set of key economic determinants for the formation of free trade agreements: trade costs, the market size of the free trade zone, and the similarity of trading partners in terms of economic development and/or factor-endowments.

Subsequent empirical research, such as that by \citet*{ChenJoshi2010} and \citet*{BaierBergstrandMariutto2014}, \citet*{SaggiStoyanovYildiz2018} among others, build upon the original \citet{baier2004economic} specification to capture novel features of FTA formation eminating from the theoretical literature: endogenous responses in trade policy, excluded country characteristics, common external trade partners and/or changes in the nature of FTAs over time.  Econometrically, in each case, researchers enrich the benchmark specification with additional co-variates intended to capture mechanisms which were beyond the scope of the benchmark \citet{baier2004economic} specification.  {Our application reexamines the Baier and Bergstrand (2004) empirical structure as it is the most highly cited FTA network specification and the basis for all subsequent empirical work in this literature.}  

Our work addresses this literature in two important dimensions.  First, we demonstrate that ignoring dyadic dependence can lead to biased estimates of model parameters and standard errors.  Second, we show that incorporating a wider set of co-variates understates true standard errors among classic determinants and can potentially lead to misleading conclusions. Specifically, our estimates also suggest that the importance market size in FTA formation may be biased and overstated when employing conventional methods, including conventional DML, and a rich specification.  Further, we document that standard estimation approaches, such as logit or conventional DML, may lead researchers to conclude that larger differences in relative factor endowment may encourage or discourage FTA formation depending on their preferred estimation approach. Accounting for dyadic dependence, in contrast, increases standard errors sufficiently that we can no longer conclude that relative factor endowments are an important determinant of FTA formation. In this sense, we further confirm that the use of econometric methods which do not allow for high-dimensional controls and/or dyadic dependence can lead to misleading conclusions, either in terms of point estimates or standard errors in the context of FTA formation models.

These empirical features are common to a wide set of network formation models across fields (see \citet*{graham2019dyadic} and \citet{graham2019econometric} for surveys), but are particularly striking in our context.  Although trade agreements are increasingly common, they continue to be relatively rare events. Despite the rise of trade agreements and the availability of rich network data, incorporating a wide host of determinants, allowing for country asymmetry, and flexibly controlling for dyadic dependence creates a particularly demanding empirical setting beyond the reach of standard tools.  As we document below, our dyadic DML approach leverages benefits to dyadic robust inference and high-dimensional link formation to identify and robustly quantify the key determinants of FTA formation. In this sense, our work yields insights which bridges a long series of theoretical research in international trade and allows for a rich of set of new empirical investigations into the nature of international trade agreement networks.

\section{The Dyadic Double/Debiased Machine Learning}\label{sec:setup}

In this section, we introduce the model and present our proposed method of dyadic machine learning.
A theoretical guarantee that this method works will be presented in Section \ref{sec:theory}.
We start by fixing notations to be used to describe  dyadic data.  
Let $\overline{\mathbb{N}^{+2}} = \{(i,j) \in \mathbb{N}^{+} \times \mathbb{N}^{+}: i \neq j\}$ denote the set of two-tuples of $\mathbb{N}^+$ without repetition, where $A^+:=A\cap (0,+\infty)$ for any $A\subset \mathbb{R}$.
A dyadic observation is written as $W_{ij}$ for $(i,j)\in \overline{\mathbb{N}^{+2}}$, and we assume that this random vector $W_{ij}$ is Borel measurable throughout.
Assume that a dyadic sample contains $N$ nodes with no self link.
Let $\{\mathcal P_N\}_N$ be a sequence of sets of probability laws of $\{W_{ij}\}_{ij}$,
let $P=P_{N}\in \mathcal P_N$ denote the law with the sample size $N$, and 
let $\Ep$ denote the expectation with respect to $P$.
We write the dyadic sample expectation operator by $\mathbb{E}_N[\cdot]=\frac{1}{N(N-1)}\sumijN[\cdot]$, where $[r]=\{1,...,r\}$ for any $r \in \mathbb{N}$. 
For any finite set $I$ with $I\subset [N]$, we let $|I|$ denote the cardinality of $I$, and
let $I^c$ denote the complement of $I$, namely, $I^c=[N]\setminus I$. 
We write the dyadic subsample expectation operator by $\mathbb{E}_{I}[\cdot]:=\frac{1}{|I|(|I|-1)}\sum_{(i,j)\in \overline{I^2}}[\cdot]$.

\subsection{The Model}\label{sec:data_structure}

The economic model is assumed to satisfy the moment restriction
\begin{align}
\Ep[\psi(W_{ij};\theta_0,\eta_0)]=0
\label{eq:existance_condition}
\end{align} 
for some score function $\psi$ that depends on a low-dimensional parameter vector $\theta \in \Theta \subset \Real^{d_\theta}$ and a nuisance parameter $\eta \in T$ for a convex set $T$.
The nuisance parameter $\eta$ may be finite-, high-, or infinite-dimensional.
The true values of $\theta$ and $\eta$ are denoted by $\theta_0 \in \Theta$ and $\eta_0 \in T$, respectively.

Let $\tilde T=\{\eta - \eta_0 : \eta \in T\}$, and define the Gateaux derivative map $D_r: \tilde T \rightarrow \Real^{d_\theta}$ by
$
D_r[\eta-\eta_0]:=\partial_r \Big\{
\Ep[\psi(W_{ij};\theta_0,\eta_0+r(\eta-\eta_0))]\Big\}
$
for all $r\in[0,1)$.
Let its limit denoted by
$
\partial_\eta\Ep\psi(W_{ij};\theta_0,\eta_0)[\eta - \eta_0]:=D_0[\eta-\eta_0].
$
We say that the Neyman orthogonality condition holds at $(\theta_0,\eta_0)$ with respect to a nuisance realization set $\mathcal T_N \subset T$ if the score $\psi$ satisfies (\ref{eq:existance_condition}), the pathwise derivative $D_r[\eta-\eta_0]$ exists for all $r\in[0,1)$ and $\eta\in \mathcal T_N$, and the orthogonality equation
\begin{align}
\partial_\eta\Ep\psi(W_{ij};\theta_0,\eta_0)[\eta - \eta_0]=0
\label{eq:Neyman_orthogonal_condition}
\end{align}
holds for all $\eta\in \mathcal T_N$. 
Furthermore, we also say that the $\lambda_N$ Neyman near-orthogonality condition holds at $(\theta_0,\eta_0)$ with respect to a nuisance realization set $\mathcal T_N\subset T$ if the score $\psi$ satisfies (\ref{eq:existance_condition}), the pathwise derivative $D_r[\eta-\eta_0]$ exists for all $r\in[0,1)$ and $\eta\in \mathcal T_N$, and the orthogonality equation
\begin{align}
\sup_{\eta \in \mathcal T_N}\Big\| \partial_\eta \Ep\psi(W_{ij};\theta_0,\eta_0)[\eta-\eta_0] \Big\|\le \lambda_N
\label{eq:Neyman_near_orthogonal_condition}
\end{align}
holds for all $\eta\in \mathcal T_N$
for some positive sequence $\{\lambda_N\}_N$ such that $\lambda_N=o(N^{-1/2})$.
We refer readers to \citet{CCDDHNR18} for detailed discussions of the Neyman (near) orthogonality, including its intuitions and a general procedure to construct scores $\psi$ satisfying it.

Throughout, we will consider structural models satisfying the moment restriction (\ref{eq:existance_condition}) and either form of the Neyman orthogonality conditions, (\ref{eq:Neyman_orthogonal_condition}) or (\ref{eq:Neyman_near_orthogonal_condition}).
As a concrete motivating example of the above abstract formulation, Example \ref{sec:example_logit} below presents the logit dyadic link formation models, which we consider for our main application in this paper.
In addition, we also present a couple of simpler examples with the linear regression models and the linear IV regression models in Appendix \ref{sec:two_examples}.

\begin{example}[Logit Dyadic Link Formation Models]\label{sec:example_logit}
Let $W_{ij}=(Y_{ij},D_{ij},X_{ij}')'$ where $Y_{ij}$ is binary indicator for a link formed between nodes $i$ and $j$, $D_{ij}$ is an explanatory variable of interest, and $X_{ij}$ is a high-dimensional random vector of controls.
Consider the logistic dyadic link formation model
\begin{align*}
\Ep[Y_{ij}|D_{ij},X_{ij}]=\Lambda(D_{ij}\theta_0 + X_{ij}'\beta_0) \text{ for }(i,j)\in \overline{[N]^2},
\end{align*}
where $\theta$ is a parameter of interest, $\beta$ is a nuisance parameter, and $\Lambda(t)=\exp(t)/(1+\exp (t))$ for all $t \in \Real$.
A nonlinear Neyman orthogonal score $\psi$ in this model is given by
$$
\psi(W_{ij};\theta,\eta)=\{Y_{ij}-\Lambda(D_{ij}\theta+X_{ij}'\beta)\}(D_{ij}-X_{ij}' \gamma),
$$
where $\eta = (\beta',\gamma')'$ with $\gamma$ denoting the coefficients of a weighted projection of $D_{ij}$ on $X_{ij}$.
See Section \ref{sec:application_dyadic_link_formation_models} for more details about this example as well as our general theory applied to this example.
\qed
\end{example}

\subsection{The Dyadic Cross Fitting}\label{sec:multiway_dml}
For the class of models introduced in Section \ref{sec:data_structure}, we now propose a novel dyadic cross fitting procedure for estimation of $\theta_0$.
With a fixed positive integer $K$, randomly partition the set $[N] = \{1,...,N\}$ of indices of $N$ nodes into $K$ parts $\{I_1,...,I_K\}$. 
For each $k \in [K] = \{1,...,K\}$, obtain an estimate
$$\hat \eta_{k}=\hat \eta\left((W_{ij})_{(i,j)\in \overline{([N]\setminus I_k )^2}}\right)$$ 
of the nuisance parameter $\eta$ by a machine learning method (e.g., lasso, post-lasso, elastic nets, ridge, deep neural networks, and boosted trees) using only the subsample of those observations with dyadic indices $(i,j)$ in $\overline{([N]\setminus I_k )^2}$.
In turn, we define the dyadic machine learning estimator $\tilde \theta$ for $\theta_0$ as the solution to 
\begin{align}
\frac{1}{K}\sumk \Enk[\psi(W;\tilde \theta,\hat \eta_{k})] =0,\label{eq:MDML}
\end{align}
where we recall that $\Enk [f(W)] = \frac{1}{|I_k|(|I_k|-1)}\sum_{(i,j)\in \overline{I_k^2}} f(W_{ij})$ denotes the subsample empirical expectation using only the those observations with dyadic indices $(i,j)$ in $\overline{I_k ^2}$.
If an achievement of the exact 0 is not possible as in \eqref{eq:MDML}, then one may define the estimator $\tilde{\theta}$ as an approximate $\epsilon_N$-solution:
\begin{align}
\label{eq:epsilon_solution}
\left\| \frac{1}{K} \sumk \Enk[\psi(W ; \tilde{\theta}, \hat{\eta}_{k})]\right\|
\leq
 \inf _{\theta \in \Theta}\left\| \frac{1}{K} \sumk \Enk\left[\psi\left(W ; \theta, \hat{\eta}_{ k}\right)\right] \right\|+\epsilon_{N},
\end{align}
where $\epsilon_N=o(\delta_N N^{-1/2})$ and restrictions on the sequence $\{\delta_N\}_N$ will be formally discussed in Section \ref{sec:theory}.
Under the assumptions to be introduced in Section \ref{sec:theory}, this dyadic machine learning estimator $\tilde\theta$ enjoys the root-$N$ asymptotic normality
$
\sqrt{N}\sigma^{-1}(\tilde \theta - \theta_0) \leadsto N(0,I_{d_\theta}),
$
where a concrete expression for the asymptotic variance $\sigma^2$ will be presented in Theorem \ref{theorem_DDML_non_linear} ahead.

Note that, for each $k\in [K]$, the nuisance parameter estimator $\hat\eta_{k}$ is computed using the subsample of those observations with dyadic indices $(i,j) \in\overline{ ([N]\setminus I_k )^2} $, and in turn the average score $\Enk[\psi(W; \cdot,\hat\eta_{k})]$ is computed using the subsample of those observations with dyadic indices $(i,j) \in\overline{ I_k ^2}$.
This two-step computation is repeated $K$ times for every  $k \in [K]$.
We call this cross fitting procedure the $K$-fold dyadic cross fitting.
It differs from and complements the cross fitting procedure of \citet{CCDDHNR18} for i.i.d. data and the cross fitting procedure of \citet{CKMS2019} for multiway clustered data -- neither of these existing cross fitting methods will work under dyadic dependence.
If the dyadic sample $\{W_{ij}\}_{ij}$ would reduce to the monadic sample $\{W_i\}_i$, then the $K$-fold dyadic cross fitting would accordingly boil down to the $K$-fold cross fitting procedure of \citet{CCDDHNR18}.

\begin{figure}
\caption{An illustration of the $2$-fold dyadic cross fitting.}\label{fig:cross_fitting}
\tikzstyle{my help lines}=[gray,very thick,dashed]
\begin{multicols}{2}
\qquad\\
\begin{tikzpicture}
\filldraw[fill=gray!33] (2,7) rectangle (1,8);
\filldraw[fill=gray!33] (3,7) rectangle (2,8);
\filldraw[fill=gray!33] (4,7) rectangle (3,8);
\filldraw[fill=gray!33] (1,6) rectangle (0,7);
\filldraw[fill=gray!33] (3,6) rectangle (2,7);
\filldraw[fill=gray!33] (4,6) rectangle (3,7) node[below] {Score}  node[above] { \ \ \ \ \ $I_1$};
\filldraw[fill=gray!33] (1,5) rectangle (0,6);
\filldraw[fill=gray!33] (2,5) rectangle (1,6);
\filldraw[fill=gray!33] (4,5) rectangle (3,6);
\filldraw[fill=gray!33] (1,4) rectangle (0,5);
\filldraw[fill=gray!33] (2,4) rectangle (1,5) node[above] {Score} node[below] {$I_1$ \ \ \ \ \ };
\filldraw[fill=gray!33] (3,4) rectangle (2,5);
\filldraw[fill=gray!33] (6,3) rectangle (5,4);
\filldraw[fill=gray!33] (7,3) rectangle (6,4);
\filldraw[fill=gray!33] (8,3) rectangle (7,4);
\filldraw[fill=gray!33] (5,2) rectangle (4,3);
\filldraw[fill=gray!33] (7,2) rectangle (6,3);
\filldraw[fill=gray!33] (8,2) rectangle (7,3) node[below] {Nuisance} node[above] { \ \ \ \ \ $I_1^c$};
\filldraw[fill=gray!33] (5,1) rectangle (4,2);
\filldraw[fill=gray!33] (6,1) rectangle (5,2);
\filldraw[fill=gray!33] (8,1) rectangle (7,2);
\filldraw[fill=gray!33] (5,0) rectangle (4,1);
\filldraw[fill=gray!33] (6,0) rectangle (5,1) node[above] {Nuisance} node[below] {$I_1^c$ \ \ \ \ \ };
\filldraw[fill=gray!33] (7,0) rectangle (6,1);
\draw[style=my help lines] (8,0) grid (0,8);
\end{tikzpicture}
\qquad\\

\begin{tikzpicture}
\filldraw[fill=gray!33] (2,7) rectangle  (1,8);
\filldraw[fill=gray!33] (3,7) rectangle (2,8);
\filldraw[fill=gray!33] (4,7) rectangle (3,8);
\filldraw[fill=gray!33] (1,6) rectangle (0,7);
\filldraw[fill=gray!33] (3,6) rectangle (2,7);
\filldraw[fill=gray!33] (4,6) rectangle (3,7) node[below] {Nuisance} node[above] { \ \ \ \ \ $I_2^c$};
\filldraw[fill=gray!33] (1,5) rectangle (0,6);
\filldraw[fill=gray!33] (2,5) rectangle (1,6);
\filldraw[fill=gray!33] (4,5) rectangle (3,6);
\filldraw[fill=gray!33] (1,4) rectangle (0,5);
\filldraw[fill=gray!33] (2,4) rectangle (1,5) node[above] {Nuisance} node[below] {$I_2^c$ \ \ \ \ \ };
\filldraw[fill=gray!33] (3,4) rectangle (2,5);
\filldraw[fill=gray!33] (6,3) rectangle (5,4);
\filldraw[fill=gray!33] (7,3) rectangle (6,4);
\filldraw[fill=gray!33] (8,3) rectangle (7,4);
\filldraw[fill=gray!33] (5,2) rectangle (4,3);
\filldraw[fill=gray!33] (7,2) rectangle (6,3);
\filldraw[fill=gray!33] (8,2) rectangle (7,3) node[below] {Score} node[above] { \ \ \ \ \ $I_2$};
\filldraw[fill=gray!33] (5,1) rectangle (4,2);
\filldraw[fill=gray!33] (6,1) rectangle (5,2);
\filldraw[fill=gray!33] (8,1) rectangle (7,2);
\filldraw[fill=gray!33] (5,0) rectangle (4,1);
\filldraw[fill=gray!33] (6,0) rectangle (5,1) node[above] {Score} node[below] {$I_2$ \ \ \ \ \ };
\filldraw[fill=gray!33] (7,0) rectangle (6,1);
\draw[style=my help lines] (8,0) grid (0,8);
\end{tikzpicture}
\end{multicols}
\end{figure}
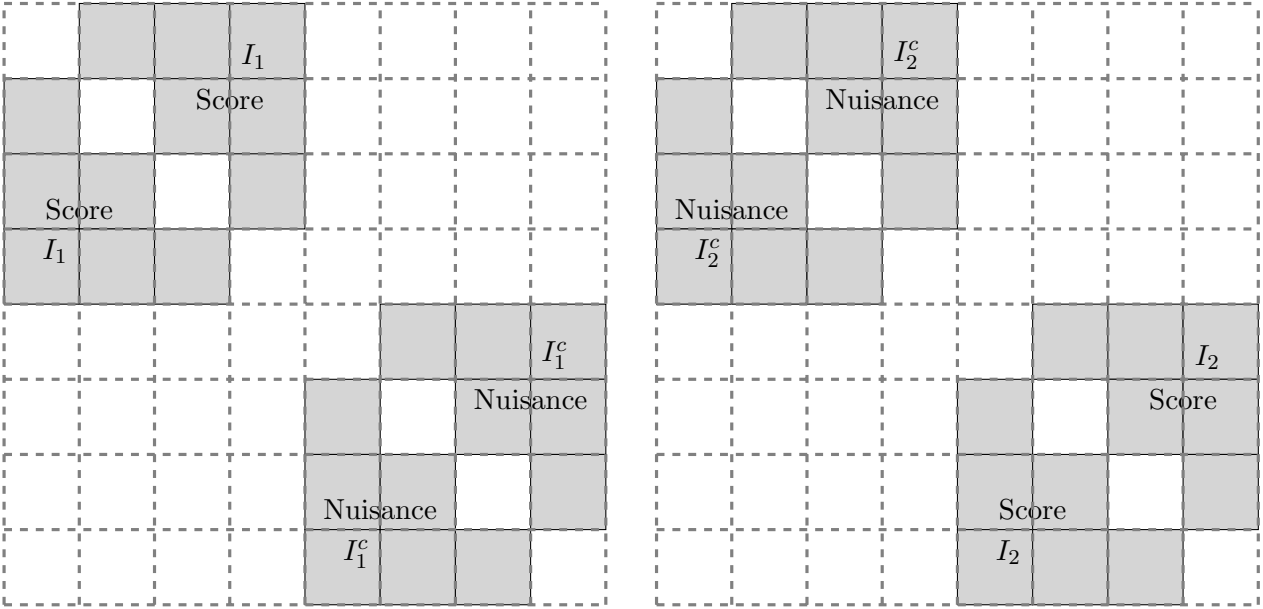

Figure \ref{fig:cross_fitting} illustrates the dyadic cross fitting for the case of $K=2$.
We let $N=8$ for simplicity of illustration, although actual sample sizes should be much larger.
Suppose that a random partition of $[8] = \{1,...,8\}$ entails two folds with one consisting of $I_1 = \{1,...,4\}$ and the other consisting of $I_2 = \{5,...,8\}$.
In the left panel, we compute $\widehat\eta_1$ with a machine learner applied to the subsample $\overline{(I_1^c)^2}$ marked by ``Nuisance'' in the gray shades at the bottom right quarter of the grid, and then evaluate the subsample mean $\mathbb{E}_{I_1}[\psi(W;\tilde \theta,\hat \eta_{1})]$ of the score using the subsample $\overline{I_1^2}$ marked by ``Score'' in the gray shades at the top left quarter of the grid.
In the right panel, in turn, we compute $\widehat\eta_2$ with a machine learner applied to the subsample $\overline{(I_2^c)^2}$ marked by ``Nuisance'' in the gray shades at the top left quarter of the grid, and then evaluate the subsample mean $\mathbb{E}_{I_2}[\psi(W;\tilde \theta,\hat \eta_{2})]$ of the score using the subsample $\overline{I_2^2}$ marked by ``Score'' in the gray shades at the bottom right quarter of the grid.
Although the figure may apparently suggest as if we were discarding the information in the off-diagonal half of the dyad, our dyadic machine learning method uses the full information of $N$ nodes for estimation of $\theta$.
Specifically, this means that Equation \eqref{eq:MDML} evaluates the score with all of the $N$ nodes -- see  Equation \eqref{eq:main_result_nonlinear} ahead for a precise mathematical expression in support of this argument.
The white off-diagonal blocks reflect that the $(K-1)/K$ fraction of data are used to estimate the nuisance parameter $\eta_k$ for each $k \in [K]$, but this will not affect the asymptotic behavior of the estimator $\widetilde\theta$ of interest due to the Neyman orthogonality, as is the case with the conventional DML of \citet{CCDDHNR18}.
As a matter of fact, the next section shows that our estimator $\widetilde\theta$ achieves the root-$N$ asymptotic normality with the asymptotic variance being the same as the one in the case where an oracle would let us know the true nuisance parameter $\eta$.

\section{Theory of the Dyadic Double/Debiased Machine Learning}\label{sec:theory}


In this section, we present a formal theory that guarantees that the dyadic machine learning estimator introduced in Section \ref{sec:multiway_dml} works.
For convenience, we fix additional notations.
Let $\{a_N\}_{N\geq 1}$, $\{v_N\}_{N\geq 1}$, $\{D_N\}_{N\geq 1}$, $\{B_{1N}\}_{N\geq 1}$, $\{B_{2N}\}_{N\geq 1}$ be some sequences of positive constants, possibly growing to infinity, where $a_N\geq N\vee D_N$ and $v_N\geq 1$, $B_{1N}\geq 1$, $B_{2N}\geq 1$ for all $N\geq 1$. 
Let $\{\delta_N\}_{N\ge 1}$, $\{\Delta_N\}_{N\ge 1}$, and $\{\tau_N\}_{N\geq 1}$ be sequences of positive constants that converge to zero such that $\delta_N \ge N^{-1/2}$.
We use $a\lesssim b$ to mean $a\leq cb$ for some $c>0$ that does not depend on $n$, use $a\gtrsim b$ to mean $a\geq cb$ for some $c>0$ that does not depend on $n$, and use the notations $a\vee b=\max\{a,b\}$ and $a\wedge b=\min\{a,b\}$.
For any vector $\delta$, we define the $l_1$-norm by $\|\delta\|_1$, $l_2$-norm by $\|\delta\|$, $l_{\infty}$-norm by $\|\delta\|_{\infty}$, and $l_0$-seminorm (the number of non-zero components of $\delta$) by $\|\delta\|_0$. For any function $f\in L^q(P)$, we define $\|f\|_{P,q}=(\int |f(w)|^q dP(w))^{1/q}$.
We use $\|x_{ij}'\delta\|_{2,N}$ to denote the prediction norm of $\delta$, namely, $\|x_{ij}'\delta\|_{2,N}=\sqrt{\mathbb{E}_N[(x_{ij}'\delta)^2]}$.

The following assumption formally states the dyadic sampling under our consideration.

\begin{assumption}[Sampling]\label{a:sampling}
Suppose that $N \to \infty $ and the following conditions hold.
\begin{enumerate}[(i)]
\item $(W_{ij})_{(i,j)\in\overline{\N^{+2}}}$ is an infinite sequence of jointly exchangeable $p$-dimensional random vectors. 
That is, for any permutation $\pi$  of $\mathbbm N$, we have
$
(W_{ij})_{(i,j)\in \overline {\N^{+2}}}\overset{d}{=} (W_{\pi(i)\pi(j)})_{(i,j)\in \overline {\N^{+2}}}.
$
\item $(W_{ij})_{(i,j)\in \overline {\N^{+2}}}$ is dissociated. 
That is, for any disjoint subsets $A,B$ of $\mathbb{N}^+$, with $\min(|A|,|B|)\geq 2$,
$
(W_{ij})_{(i,j)\in \overline{A^2}}
$
is independent of 
$
(W_{ij})_{(i,j)\in \overline{B^2}}.
$
\item For each $N$, an econometrician observes $(W_{ij})_{(i,j)\in \overline {[N]^2}}$.
\end{enumerate}
\end{assumption}
\noindent
Part (i) requires the identical distribution of $W_{ij}$ while we do allow for dyadic dependence.
Part (ii) requires that two groups of links in the dyad are independent whenever they do not share a common node.
On the other hand, we do allow for arbitrary statistical dependence between any pair of links whenever they share a common node.
Assumption \ref{a:sampling} replaces the i.i.d. sampling assumption of \citet{CCDDHNR18}, and  is a dyadic counterpart of Assumption 1 in \citet{CKMS2019} for multiway clustered data.
Assumption \ref{a:sampling} (i) and (ii) together imply the Aldous-Hoover-Kallenberg representation (e.g. \citealp[Corollary 7.35]{Kallenberg2006}), which states that there exists an unknown (to the researcher) Borel measurable function $\tau_n$ such that
\begin{align}
W_{ij}\overset{d}{=}\tau_n (U_i,U_j,U_{\{i,j\}}),\label{eq:AKH}
\end{align}
where $\{U_i, U_{\{i,j\}}: i,j\in [N], i\ne j\}$ are some i.i.d. latent shocks that can be taken to be $\text{Unif}[0,1]$ without loss of generality -- see \cite{aldous1981representations}.
This is a mathematically different structural representation (Aldous-Hoover-Kallenberg representation; see e.g. \citealp[Theorem 7.22 and Corollary 7.35]{Kallenberg2006}) from that in either of these preceding papers and thus none of the existing DML methods apply here -- see \citet[Appendix I]{chiang2020inference} for a detailed comparison and discussion of the differences.
Finally, we remark that monadic variables (i.e., coordinates of $W_{ij}$ that do not vary with $i$ or $j$) are not ruled out from our sampling framework of Assumption \ref{a:sampling}. 

We further state the following two conditions.

\begin{assumption}[Nonlinear Moment Condition Problem with Approximate Neyman Orthogonality]\label{a:nonlinear_moment_condition}
The following conditions hold for all $P\in \mathcal P_N$ for all $N \ge 4$.
\begin{enumerate}[(i)]
\item $\Theta$ is independent of $N$ and contains a ball of radius $c_1N^{-1/2}\log N=o(1)$ centred at $\theta_0$.
\item The map $(\theta,\eta)\mapsto \Ep[\psi(W_{12};\theta,\eta)]$ is twice continuously Gateaux differentiable on $\Theta\times \mathcal{T}_N$.
\item For all $\theta \in \Theta$, the identification relation 
\begin{align*}
2||\Ep[\psi(W_{12};\theta,\eta_0)]||\geq ||J_0(\theta-\theta_0)||\wedge c_0
\end{align*}
is satisfied with the Jacobian matrix
\begin{align*}
J_0:=\partial_{\theta'}\{\Ep[\psi(W;\theta,\eta_0)]\}|_{\theta=\theta_0}
\end{align*}
having singular values between $c_0$ and $c_1$.
\item The score $\psi$ satisfies the Neyman orthogonality or, more generally, the Neyman near-orthogonality with $\lambda_N=\delta_N N^{-1/2}$ for the set $\mathcal{T}_N \subset T$.
\end{enumerate}
\end{assumption}

\begin{assumption}[Score Regularity and Requirement on the Quality of Estimation of Nuisance Parameters]\label{a:nonlinear_score_regularity_nuisance_parameters}
Let $K$ be a fixed integer. The following conditions hold for all $P\in \mathcal P_N$ for all $N \ge 4$.
\begin{enumerate}[(i)]
\item 
The function class $\calF_{1}=\{\psi(\cdot;\theta,\eta):\theta \in \Theta,\eta\in \mathcal{T}_N\}$ is pointwise measurable\footnote{See \citet[pp. 110]{van1996weak} for its definition.} and its uniform entropy numbers\footnote{See \citet[Ch 2.6]{van1996weak} for its definition.} satisfy
\begin{align}
\label{unifrom_entropy_assumption}
\sup_{Q }\log N(\calF_{1},\|\cdot\|_{Q,2},\varepsilon\|F_{1}\|_{Q,2})\leq v_N\log (a_N/\varepsilon),\text{ for all } \:0<\varepsilon\le 1,
\end{align}
where $F_1$ is a measurable envelope for $\mathcal{F}_1$ that satisfies $\|F_{1}\|_{P,q}\leq D_N$.
\item For all $f\in\mathcal{F}_1$, we have $c_0\leq \|f\|_{P,2}\leq C_0$.
\item  Given random subsets $I\subset [N]$  such that $|I|=\lfloor N/K\rfloor$, the nuisance parameter estimator $\hat \eta=\hat\eta((W_{ij})_{(i,j)\in \overline{([N]\setminus I)^2}}) $ belongs to the realization set $\mathcal T_N$ with probability at least $1-\Delta_N$, where $\mathcal T_N$ contains $\eta_0$.

\item  The following conditions hold for all $r\in [0,1)$, $\theta\in\Theta$, and $\eta\in\mathcal{T}_N$,
\begin{enumerate}
\item
$\Ep\big[\|\psi(W_{12};\theta,\eta)-\psi(W_{12};\theta_0,\eta_0)\|^2\big]\lesssim (\|\theta-\theta_0\|\vee\|\eta-\eta_0\|)^2,$
\item
$\big\|\partial_r\Ep\big[\psi\big(W_{12};\theta,\eta_0+r(\eta-\eta_0)\big)\big]\big\|\leq B_{1N}\|\eta-\eta_0\|,$
\item
$\big\|\partial_r^2\Ep\big[\psi\big(W_{12};\theta_0+r(\theta-\theta_0),\eta+r(\eta-\eta_0)\big)\big]\big\|\leq B_{2N}\Big(\|\theta-\theta_{0}\|^2\vee\|\eta-\eta_0\|^2\Big).
$
\end{enumerate}
\item For all $\eta\in\mathcal{T}_N$, $\|\eta-\eta_0\|\leq \tau_N$.
\item All eigenvalues of the matrix
\begingroup
\allowdisplaybreaks
\begin{align*}
\Gamma:=\Ep[\psi(W_{12};\theta_0,\eta_0)\psi(W_{13};\theta_0,\eta_0)']+\Ep[\psi(W_{12};\theta_0,\eta_0)\psi(W_{31};\theta_0,\eta_0)']\\
+\Ep[\psi(W_{21};\theta_0,\eta_0)\psi(W_{13};\theta_0,\eta_0)']+\Ep[\psi(W_{21};\theta_0,\eta_0)\psi(W_{31};\theta_0,\eta_0)']
\end{align*}
\endgroup
are bounded from below by $c_0$. 
\item  $a_N$ and $v_N$ satisfy
\begin{enumerate}
\item
$N^{-1/2}(v_N\log a_N)^{1/2}\leq C_0\tau_N,$
\item
$N^{-1/2+1/q}v_ND_N\log a_N\leq C_0\delta_N,$
\item
$N^{1/2}B_{1N}^2B_{2N}\tau_N^2\leq C_0\delta_N.$
\end{enumerate}
\end{enumerate}
\end{assumption}

\noindent
The last two assumptions are counterparts of Assumptions 2.1 and 2.2 in \citet{BCCW18}, except that the data in this paper are allowed to exhibit dyadic dependence.
{  One observation is that although there are multiple $N$-dependent sequences of constants appearing in these assumptions that might be, at first glimpse, interacting in complicated ways, their rates are pinned down by $\delta_N$ in Assumption \ref{a:nonlinear_score_regularity_nuisance_parameters} (vii), analogously to the high level conditions in \cite{BCCW18} and \cite{CCDDHNR18}. }
Although these conditions are high-level statements, we provide explicit lower-level primitive conditions that imply these assumptions in the context of the logit dyadic link formation model (Example \ref{sec:example_logit}) in Section \ref{sec:lower_level}. 

The following theorem establishes the root-$N$ asymptotic normality of the dyadic DML estimator $\hat\theta$ along with its consistent asymptotic variance estimation.

\begin{theorem}(Dyadic DML for Nonlinear Scores)
\label{theorem_DDML_non_linear}
Suppose that Assumptions \ref{a:sampling}, \ref{a:nonlinear_moment_condition}
and \ref{a:nonlinear_score_regularity_nuisance_parameters} hold.
Then,
\begin{align}
\sqrt{N}\sigma^{-1}(\tilde{\theta}-\theta_0)
&=\frac{\sqrt{N}}{K}\sumk\Enk\bar{\psi}(W)+O_{P_N}(\rho_N)
\notag\\
&=\frac{1}{\sqrt{N}}\sum_{i=1}^N \left\{\Ep [\bar\psi(W_{i0})\mid U_i] +\Ep [\bar\psi(W_{0i})\mid U_i]\right\}+\Opn(\rho_N) \leadsto N(0,I_{d_\theta})
\label{eq:main_result_nonlinear}
\end{align}
holds uniformly over $P\in\mathcal{P}_N$, where the order of remainder follows $
\rho_N \lesssim \delta_N,
$
the influence function takes the form $\bar{\psi}(\cdot):=-\sigma^{-1}J_0^{-1}\psi(\cdot;\theta_0,\eta_0)$, and the variance is given by
\begin{align*}
\sigma^2:=J_0^{-1}\Gamma(J_0^{-1})'.
\end{align*}
Moreover, $\sigma^2$ can be replaced by $\hat{\sigma}^2=\hat{J}^{-1}\hat{\Gamma}(\hat{J}^{-1})'$, where 
\begin{align*}
\hat{J}:=\frac{1}{K}\sumk\Enk[\partial_\theta \psi(W;\theta,\hat{\eta}_k)]|_{\theta=\tilde{\theta}}
\end{align*}
and
\begingroup
\allowdisplaybreaks
\begin{align}
\label{estimator_gamma}
\hat \Gamma:&=\frac{1}{K}\sumk\frac{|I_k|-1}{(|I_k|(|I_k|-1))^2}
\Big[
\sum_{i \in I_k}\sum_{\substack{{j,j' \in I_k}\\{j,j'\neq i}}}\psi(W_{ij};\tilde \theta,\hat\eta_{k}) \psi(W_{ij'};\tilde \theta,\hat\eta_{k})' \\
\nonumber
&+\sum_{j \in I_k}\sum_{\substack{{i,i' \in I_k}\\{i,i'\neq j}}}\psi(W_{ij};\tilde \theta,\hat\eta_{k}) \psi(W_{i'j};\tilde \theta,\hat\eta_{k})' 
+\sum_{i \in I_k}\sum_{\substack{{j,j' \in I_k}\\{j,j'\neq i}}}\psi(W_{ij};\tilde \theta,\hat\eta_{k}) \psi(W_{j'i};\tilde \theta,\hat\eta_{k})'
\\ 
\nonumber
&+\sum_{j \in I_k}\sum_{\substack{{i,i' \in I_k}\\{i,i'\neq j}}}\psi(W_{ij};\tilde \theta,\hat\eta_{k}) \psi(W_{ji'};\tilde \theta,\hat\eta_{k})'
 \Big].
\end{align}
\endgroup
\end{theorem}

The asymptotic distribution is different from both that of the prototypical DML\citep{CCDDHNR18} for i.i.d. data and that of the multi-way cluster-robust DML \citep{CKMS2019}.
In addition to the differences in dependence structures, a yet major departure from Theorem 1 in  \citet{CKMS2019} is that this result allows a large class of nonlinear score, whereas only linear scores are permitted in Chiang et al. To handle the complications resulted from the nonlinearity, it is necessary to handle empirical processes that consist of dyadic random elements. Coping with these technical issues requires delicate works and hence Theorem \ref{theorem_DDML_non_linear} is not a straightforward extension of results in \citet{CKMS2019}.
\section{Application to Dyadic Link Formation Models}\label{sec:application_dyadic_link_formation_models}
\subsection{The High-Dimensional Logit Dyadic Link Formation Models}
\label{model}
In this section, we demonstrate an application of the general theory of dyadic machine learning from Section \ref{sec:theory} to high-dimensional dyadic/network link formation models.
Consider a class of logit dyadic link formation models where the binary outcome $Y_{ij}$, which indicates a link formed between nodes $i$ and $j$, is chosen based on an explanatory variable $D_{ij}$ and $p$-dimensional control variables $X_{ij}$.
Specifically, we write the model
\begin{align*}
\Ep[Y_{ij}|D_{ij},X_{ij}]=\Lambda(D_{ij}\theta_0 + X_{ij}'\beta_0) \text{ for }(i,j)\in \overline{[N]^2},
\end{align*}
where $\Lambda(t)=\exp(t)/(1+\exp (t))$ for all $t \in \Real$. Suppose that we are interested in the parameter $\theta$, while $x \mapsto x'\beta_0$ is treated as a nuisance function $\eta$.

Under i.i.d. sampling environments, \cite{BelloniChernozhukovWei2016} propose a method of inference for $\theta_0$ in high-dimensional logit models.
We demonstrate that a modification of their suggested procedure with a flavor of our proposed dyadic cross fitting enables root-$N$ consistent estimation and asymptotically valid inference for $\theta_0$ under dyadic sampling environments by virtue of Theorem \ref{theorem_DDML_non_linear}.
In Section \ref{sec:lower_level}, we will supplement this procedure with a formal discussion of lower-level primitive conditions for the high-level statements in Assumptions \ref{a:nonlinear_moment_condition} and \ref{a:nonlinear_score_regularity_nuisance_parameters} that are invoked in Theorem \ref{theorem_DDML_non_linear}. 

As in \cite{BelloniChernozhukovWei2016}, our goal is to construct a generated random variable $Z_{ij}=Z(D_{ij},X_{ij})$ such that the following three conditions are satisfied:
\begingroup
\allowdisplaybreaks
\begin{align}
 &\ \Ep\left[\left\{Y_{ij}-\Lambda\left(D_{ij} \theta_{0}+X_{ij}'\beta_0\right)\right\} Z_{ij}\right] 
 =0,
 \label{IV_estimation_1} \\
 &\left.\frac{\partial}{\partial \theta} \Ep\left[\left\{Y_{ij}-\Lambda\left(D_{ij} \theta+X_{ij}' \beta_0\right)\right\} Z_{ij}\right]\right|_{\theta=\theta_{0}} 
  \neq 0,
	\qquad\text{and}
 \label{IV_estimation_2} \\
 &\left.\frac{\partial}{\partial \beta} \Ep\left[\left\{Y_{ij}-\Lambda\left(D_{ij} \theta_{0}+X_{ij}'\beta\right)\right\} Z_{ij}\right]\right|_{\beta=\beta_0} 
 =0.
 \label{orthogonality_equation_logit}
\end{align}
\endgroup
Equations (\ref{IV_estimation_1}) and (\ref{IV_estimation_2}) are used for consistent estimation of $\theta_0$, while Equation (\ref{orthogonality_equation_logit}) obeys the Neyman orthogonality condition -- see Section \ref{sec:data_structure}.

To this end, following \cite{BelloniChernozhukovWei2016}, consider the weighted regression of $D_{ij}$ on $X_{ij}$:
\begin{align}
\label{decomposition}
f_{ij} D_{ij}=f_{ij} X_{ij}' \gamma_{0}+V_{ij}, \quad \text { with } \quad \mathrm{E}_P\left[f_{ij} V_{ij} X_{ij}\right]=0,
\end{align}
where 
$f_{ij}:=w_{ij} / \sigma_{ij}$, $\sigma_{ij}^{2}:=\V\left(Y_{ij} | D_{ij}, X_{ij}\right)$,
$w_{ij}:=\Lambda^{(1)}\left(D_{ij} \theta_{0}+X_{ij}^{\prime} \beta_{0}\right)$, and $\Lambda^{(1)}(t)=\frac{\partial}{\partial t}\Lambda(t)$.
Under the logit link $\Lambda$, these objects are specifically given by $f_{ij}^2=w_{ij}$ and
$\sigma_{ij}^{2}=w_{ij}=\Lambda(D_{ij} \theta_{0}+X_{ij}^{\prime} \beta_{0})\{1-\Lambda(D_{ij} \theta_{0}+X_{ij}^{\prime} \beta_{0})\}$.
We let $Z_{ij}=D_{ij}-X_{ij}'\gamma_0$.

From the log-likelihood function of the logit model $
L(\theta, \beta)=\mathbb{E}_{N}\left[L(W_{ij};\theta, \beta)\right],
$
where $L(W_{ij};\theta, \beta)=\log \{1+\exp(D_{ij} \theta+X_{ij}' \beta)\}-Y_{ij}(D_{ij} \theta+X_{ij}' \beta)$,
one can derive the following  Neyman orthogonal score
\begin{align}
\label{score_logistic}
\psi(W_{ij};\theta,\eta)=\{Y_{ij}-\Lambda(D_{ij}\theta+X_{ij}'\beta)\}(D_{ij}-X_{ij}' \gamma),
\end{align}
where $\eta=(\beta',\gamma')'$ denotes the nuisance parameters. 

With this procedure, applying Theorem \ref{theorem_DDML_non_linear} yields the asymptotic normality result
\begin{align}\label{eq:asymptotic_normality_logit}
\sqrt{N}\sigma^{-1}(\tilde{\theta}-\theta_0)\leadsto N(0,1),
\end{align}
where the components of the variance
$
\sigma^2:=J_0^{-1}\Gamma(J_0^{-1})'
$
take the forms of
\begingroup
\allowdisplaybreaks
\begin{align}
\label{true_J_logistic}
J_0&=\Ep[-D_{ij}\Lambda(D_{ij}\theta_0+X_{ij}'\beta_0)\{1-\Lambda(D_{ij}\theta_0+X_{ij}'\beta_0)\}(D_{ij}-X_{ij}'\gamma_0)]
\qquad\text{and}
\\
\label{true_Gamma_logistic}
\nonumber
\Gamma&=\Ep[\{Y_{12}-\Lambda(D_{12}\theta_0+X_{12}'\beta_0)\}(D_{12}-X_{12}' \gamma_0)\{Y_{13}-\Lambda(D_{13}\theta_0+X_{13}'\beta_0)\}(D_{13}-X_{13}' \gamma_0)]
\\
&+\nonumber
\Ep[\{Y_{12}-\Lambda(D_{12}\theta_0+X_{12}'\beta_0)\}(D_{12}-X_{12}' \gamma_0)\{Y_{31}-\Lambda(D_{31}\theta_0+X_{31}'\beta_0)\}(D_{31}-X_{31}' \gamma_0)]
\\
&+\nonumber
\Ep[\{Y_{21}-\Lambda(D_{21}\theta_0+X_{21}'\beta_0)\}(D_{21}-X_{21}' \gamma_0)\{Y_{13}-\Lambda(D_{13}\theta_0+X_{13}'\beta_0)\}(D_{13}-X_{13}' \gamma_0)]
\\
&+
\Ep[\{Y_{21}-\Lambda(D_{21}\theta_0+X_{21}'\beta_0)\}(D_{21}-X_{21}' \gamma_0)\{Y_{31}-\Lambda(D_{31}\theta_0+X_{31}'\beta_0)\}(D_{31}-X_{31}' \gamma_0)].
\end{align}
\endgroup
Section \ref{sec:lower_level} will present formal theories to guarantee that this application of Theorem \ref{theorem_DDML_non_linear} is valid, based on lower-level sufficient conditions tailored to the current specific model.

Summarizing the above procedure, we provide a step-by-step algorithm below for dyadic machine learning estimation and inference about $\theta_0$.
For any set $I$, we let $I^c$ denote its complement, and let $|I|$ denote the cardinality of $I$. 
We remind readers of the notation for the dyadic subsample expectation operator $\mathbb{E}_{I}[\cdot]:=\frac{1}{|I|(|I|-1)}\sum_{(i,j)\in \overline{I^2}}[\cdot]$.

\begin{algorithm}[$K$-fold Dyadic DML for High-Dimensional Logit Dyadic Link Formation Models]${}$
\label{algorithm_logit} 
\begin{enumerate}
\item Randomly partition $[N]$ into $K$ parts $\{I_1,...,I_K\}$.
\item For each $k\in[K]$: obtain an post-lasso logistic estimate  $(\tilde{\theta}_k,\tilde{\beta}_k)$ of the nuisance parameter by  using only the subsample of those observations with dyadic indices $(i,j)$ in $\overline{([N]\setminus I_k)^2}$,
\begingroup
\allowdisplaybreaks
\begin{align*}
(\hat{\theta}_k, \hat{\beta}_k) &\in \arg \min _{\theta, \beta} \mathbb{E}_{I_k^c}\left[L(W_{ij};\theta, \beta)\right]+\frac{\lambda_{1}}{|I_k^c|}\|(\theta, \beta)\|_{1}, \\ (\widetilde{\theta}_k, \widetilde{\beta}_k) &\in \arg \min _{\theta, \beta} \mathbb{E}_{I_k^c}\left[L(W_{ij};\theta, \beta)\right]: \operatorname{support}(\theta,\beta) \subseteq \operatorname{support}(\hat{\theta}_k,\hat{\beta}_k).
\end{align*}
\endgroup
\item For each $k\in[K]$: calculate the weight $\hat{f}_{ij,k}^2=\Lambda(D_{ij} \widetilde{\theta}_k+X_{ij}' \widetilde{\beta}_k)
\{1-\Lambda(D_{ij} \widetilde{\theta}_k+X_{ij}' \widetilde{\beta}_k)\}$ where $(i,j)\in \overline{([N]\setminus I_k)^2}$.
\item For each $k\in [K]$: obtain an post-lasso OLS estimate  $\tilde{\gamma}_k$ of the nuisance parameter by using only the subsample of those observations with dyadic indices $(i,j)$ in $\overline{([N]\setminus I_k)^2}$,
\begingroup
\allowdisplaybreaks
\begin{align*}
\hat{\gamma}_k &\in \arg \min _{\gamma} \mathbb{E}_{I_k^c}\left[\hat{f}_{ij,k}^{2}\left(D_{ij}-X_{ij}'\gamma\right)^{2}\right]+\frac{\lambda_{2}}{|I_k^c|}\| \gamma\|_{1}, \\
 \widetilde{\gamma}_k &\in \arg \min _{\gamma} \mathbb{E}_{I_k^c}\left[\hat{f}_{ij,k}^{2}\left(D_{ij}-X_{ij}' \gamma\right)^{2}\right]: \quad \operatorname{support}(\gamma) \subseteq \operatorname{support}(\widehat{\gamma}_k).
\end{align*}
\endgroup
\item Solve the equation $\frac{1}{K}\sum\limits_{k\in[K]} \mathbb{E}_{I_k} [\psi(W;\theta,\tilde{\eta}_{k})]=0$ for $\theta$ to obtain the dyadic machine learning estimate $\check{\theta}$, where $\psi(W;\theta,\tilde{\eta}_k)$ takes the form in Equation(\ref{score_logistic}) with $\tilde{\eta}_k=(\tilde{\beta}'_k,\tilde{\gamma}'_k)'$ and $(i,j)\in \overline{I_k^2}$. 
\item Let the dyadic Lasso DML asymtotic variance estimator be given by $\hat{\sigma}^2=\hat{J}^{-1}\hat{\Gamma}(\hat{J}^{-1})'$ where 
\begingroup
\allowdisplaybreaks
\begin{align*}
\hat{J}&=-\frac{1}{K}\sum_{k\in[K]}\mathbb{E}_{I_k}\big[\Lambda(D_{ij}\check{\theta}+X_{ij}'\tilde{\beta}_k)\{1-\Lambda(D_{ij}\check{\theta}+X_{ij}'\tilde{\beta}_k)\}(D_{ij}-X_{ij}'\tilde{\gamma}_k)D_{ij}\big]
,
\\
\hat{\Gamma}&=\frac{1}{K}\sum_{k\in[K]}\frac{|I_k|-1}{(|I_k|(|I_k|-1))^2}
\Big[
\sum_{i \in I_k}\sum_{\substack{{j,j' \in I_k}\\{j,j'\neq i}}}\psi(W_{ij};\check \theta,\tilde\eta_{k}) \psi(W_{ij'};\check \theta,\tilde\eta_{k})' \\
&+\sum_{j \in I_k}\sum_{\substack{{i,i' \in I_k}\\{i,i'\neq j}}}\psi(W_{ij};\check{\theta},\tilde\eta_{k}) \psi(W_{i'j};\check \theta,\tilde\eta_{k})' 
+\sum_{i \in I_k}\sum_{\substack{{j,j' \in I_k}\\{j,j'\neq i}}}\psi(W_{ij};\check \theta,\tilde\eta_{k}) \psi(W_{j'i};\check {\theta},\tilde\eta_{k})' \\
&+\sum_{j \in I_k}\sum_{\substack{{i,i' \in I_k}\\{i,i'\neq j}}}\psi(W_{ij};\check {\theta},\tilde\eta_{k}) \psi(W_{ji'};\check \theta,\tilde\eta_{k})'
 \Big].
\end{align*}
\endgroup
\item Report the estimate $\check{\theta}$, its standard error $\sqrt{\hat{\sigma}^2/N}$, and/or the $(1-a)$ confidence interval
\begin{align*}
\text{CI}_a:=[\check\theta\pm \Phi^{-1}(1-a/2)\sqrt{\hat \sigma^2 /N}].
\end{align*} 
\end{enumerate}
\end{algorithm}

\subsection{Lower-Level Sufficient Conditions}\label{sec:lower_level}
In this section, we provide lower-level sufficient conditions that guarantee the root-$N$ asymptotic normality \eqref{eq:asymptotic_normality_logit} in the application to the high-dimensional logit dyadic link formation models.
For convenience of stating such conditions, we first introduce additional notations.
We continue to use $Z_{ij}=D_{ij}-X_{ij}'\gamma_0$ from the previous subsection. 
Let $a_N=p\vee N$. 
Let $q$, $c_1$, $C_1$ be strictly positive and finite constants with $q>4$,
 $M_N$ be a positive sequence of constants such that 
$M_N\geq \{\Ep[(D_{12}\vee \|X_{12}\|_{\infty})^{2q}]\}^{1/2q}$.
For any $T\subset [p+1]$, $\delta=(\delta_1,...,\delta_{p+1})'\in \Real^{p+1}$, denote $\delta_T=(\delta_{T,1},...,\delta_{T,p+1})'$ with $\delta_{T,j}=\delta_j$ if $j\in T$ and $\delta_{T,j}=0$ if $j\not \in T$.  Define the minimum and maximum sparse eigenvalues by 
\begin{align*}
\phi_{\min}(m)=
\inf _{\|\delta\|_0 \le m} \frac{\|\sqrt{w_{ij}
} (D_{ij},X_{ij}')'\delta\|_{2,N}}{\|\delta_{T}\|_1}, \quad \phi_{\max}(m)=
\sup _{\|\delta\|_0 \le m} \frac{\|\sqrt{w_{ij}}(D_{ij},X_{ij}')'\delta\|_{2,N}}{\|\delta_{T}\|_1}.
\end{align*}
The following assumptions constitute the lower level sufficient conditions.
Let $s_N\ge 1$ be a sequence of integers.
\begin{assumption}(Sparse eigenvalue conditions).
\label{a:RE}
The sparse eigenvalue conditions hold with probability $1-o(1)$, namely, for some $\ell_N\to \infty$ slow enough, we have
\begin{align*}
1\lesssim  \phi_{\min}(\ell_N s_N)\le
 \phi_{\max}(\ell_N s_N)\lesssim 1.
\end{align*}
\end{assumption}

\begin{assumption}(Sparsity).
	\label{a:sparsity}
	$\|\beta_0\|_{0}+\|\gamma_0\|_{0} \leqslant s_{N}$.
\end{assumption}

\begin{assumption}(Parameters).
\label{a:parameter}
\begin{enumerate}[(i)]
\item $\|\beta_0\|+\|\gamma_0\|\leq C_1$.
\item $\sup_{\theta\in \Theta}|\theta	|\leq C_1$.
\end{enumerate}
\end{assumption}
\begin{assumption}(Covariates).
\label{a:covariates}
The following inequalities hold for some $q>4$:
\begin{enumerate}[(i)]
\item $\max_{j=1,2}\inf _{\|\xi\|=1} \mathrm{E}_{P}\big[\big(\Ep [f_{12}\left(D_{12}, X_{12}'\right) \mid U_{j}]\xi\big)^{2}\big] \geqslant c_{1}$.
\item $\max_{j=1,2}\min_{l\in[p]} \big(\mathrm{E}_{P}\big[\big(\Ep [f_{12}^{2} Z_{12} X_{12,l}\mid U_j]\big)^{2}\big] \wedge \mathrm{E}_{P}\big[\big(\Ep [f_{12}^{2} D_{12} X_{12,l}\mid U_j]\big)^{2}\big] \big) \geqslant c_{1}$.
\item $\sup _{\|\xi\|=1} \mathrm{E}_{P}[\{(D_{12}, X_{12}') \xi\}^{4}] \leqslant C_{1}$.

\item $N^{-1/2+2/q}M_{N}^2s_N\log^2 a_N\leq \delta_N $.
\end{enumerate}
\end{assumption}

\begin{assumption}
\label{a:eigenvalue_logistic}
All eigenvalues of the matrix $\Gamma$ defined in (\ref{true_Gamma_logistic}) are bounded from below by $c_0$.
\end{assumption}

\noindent
Assumption \ref{a:RE} is a sparse eigenvalue condition similar to the one in \cite{bickel2009simultaneous} under i.i.d. setting. Sufficient conditions for the sparse eigenvalue conditions under the related multiway clustering is provided by Proposition 3 in \cite{chiang2020inference}. Assumption \ref{a:sparsity} imposes a restriction on the number of non-zero components in the nuisance parameter vectors by $s_N$, where a constraint on the growth rate of $s_N$ in turn will be introduced in Assumption \ref{a:covariates} (iv). 
Assumption \ref{a:parameter} is standard and requires the parameter space $\Theta$ to be bounded as well as the true nuisance parameter vectors $\beta_0$ and $\gamma_0$ to have bounded $\ell_2$-norms. 
It permits $\|\beta_0\|_1$ and $\|\gamma_0\|_1$ to be growing as $N$ increases. 
Assumption \ref{a:covariates} (i) and (ii) impose lower bounds on eigenvalues and variances of some population conditional moments.
Note that the $U_j$'s here come from the Aldous-Hoover-Kallenberg representation introduced in Equation (\ref{eq:AKH}). 
Assumption \ref{a:covariates} (iii) requires the fourth moments of the covariates to be bounded in a uniform manner. Finally, Assumption \ref{a:covariates} (iv) constraints the rates at which the sparsity index $s_N$, the dimensionality $p$, as well as the bound for the $2q$-th moments of the variable of interest and the maximal components of the covariates vector can grow.
It also implicitly requires these $2q$-th moments to exist (albeit the bounds can be diverging with $N$). 
As an illustration, suppose that $\|(D_{ij},X_{ij}')\|_\infty=1$. 
Then  $q$ can be taken arbitrarily large and $M_N=1$. 
In this case, it essentially requires $N^{-1}s_N^2\log^2 a_N =o(1)$. 
These rate requirements are comparable to those assumed in state-of-the-art results in the literature under i.i.d. sampling (e.g. \cite{BCCW18}).

These assumptions are sufficient for the high-level conditions invoked in Theorem \ref{theorem_DDML_non_linear}, as formally stated in the following lemma.

\begin{lemma}
\label{lemma_logistic_lasso}
Assumptions \ref{a:RE}, \ref{a:sparsity}, \ref{a:parameter}, \ref{a:covariates} and \ref{a:eigenvalue_logistic} imply Assumptions \ref{a:nonlinear_moment_condition} and \ref{a:nonlinear_score_regularity_nuisance_parameters}.
\end{lemma}

As a consequence of this lemma and the general result in Theorem \ref{theorem_DDML_non_linear}, we obtain the following result of the root-$N$ asymptotic normality of the dyadic machine learning estimator $\check\theta$ in the high-dimesnional logit dyadic link formation models.

\begin{theorem}(Inference about Regression Coefficients in the High-Dimensional Logit Dyadic Model)
\label{theorem_logit}
Suppose that Assumptions \ref{a:sampling}, \ref{a:RE}, \ref{a:sparsity}, \ref{a:parameter}, \ref{a:covariates}, and \ref{a:eigenvalue_logistic} hold.
Then, the dyadic machine learning estimator $\check{\theta}$ constructed using score (\ref{score_logistic}) follows
\begin{align*}
\sqrt{N}\sigma^{-1}(\check{\theta}-\theta_0)\leadsto N(0,1),
\end{align*}
where $\sigma^2=J_0^{-1}\Gamma (J_0^{-1})'$, $J_0$ and $\Gamma$ are defined in Equation(\ref{true_J_logistic}) and (\ref{true_Gamma_logistic}) separately.
Moreover, this result continues to hold when $\sigma^2$ is replaced by $\hat{\sigma}^2$ defined in Algorithm \ref{algorithm_logit}. 
\end{theorem}

\section{Monte Carlo Simulation Studies}

Focusing on the application to the high-dimensional dyadic link formation model introduced in Section \ref{sec:application_dyadic_link_formation_models}, we study and present finite sample performance of the proposed method of dyadic machine learning in this section through Monte Carlo simulations.

We design the data generating process for our Monte Carlo simulation studies as follows.
For each $(i,j) \in \overline{[N]^2}$, generate the random vector $(D_{ij},X_{ij}',\varepsilon_{ij})'$ according to 
\begingroup
\allowdisplaybreaks
\begin{align*}
D_{ij} &= \frac{1}{3} \tilde D_{i} + \frac{1}{3} \tilde D_{j} + \frac{1}{3} \tilde D_{ij},
\\
X_{ij} &= \frac{1}{3} \tilde X_{i} + \frac{1}{3} \tilde X_{j} + \frac{1}{3} \tilde X_{ij},
\\
\varepsilon_{ij} &= F_{\text{Logistic}(0,1)}^{-1} \circ F_{\text{Normal}(0,1)}\left( \sqrt{\frac{1}{3}} \tilde \varepsilon_{i} + \sqrt{\frac{1}{3}} \tilde \varepsilon_{j} + \sqrt{\frac{1}{3}} \tilde \varepsilon_{ij} \right),
\end{align*}
\endgroup
where $\{(\tilde D_{i}, \tilde X_{i}')'\}_{i \in [N]}$, $\{(\tilde D_{ij}, \tilde X_{ij}')'\}_{(i,j) \in \overline{[N]^2}}$, $\{\tilde \varepsilon_{i}\}_{i \in [N]}$ and $\{\tilde \varepsilon_{ij}\}_{(i,j) \in \overline{[N]^2}}$ are independently drawn with the laws
\begin{align*}
(\tilde D_{i}, \tilde X_{i}')' &\sim \text{Normal}(0,\Sigma),&
(\tilde D_{ij}, \tilde X_{ij}')' &\sim \text{Normal}(0,\Sigma),
\\
\tilde \varepsilon_{i} &\sim \text{Normal}(0,1),&
\tilde \varepsilon_{ij} &\sim \text{Normal}(0,1),
\end{align*}
and $\Sigma$ is the variance-covariance matrix whose $(r,c)$-th entry given by $\Sigma_{rc} = 5^{-|r-c|}$.
While $D_{ij}$ is designed as a scalar random variable, we will vary the dimension of the random vector $X_{ij}$ across sets of simulations.

From these primitive variables, construct the dyadic binary decision $Y_{ij}$ in turn by the logistic threshold-crossing model
\begin{align*}
Y_{ij} = \mathbbm{1}\{  D_{ij}\theta_0 + X_{ij}'\beta_0 \ge \varepsilon \}.
\end{align*}
Note that this data generating model implies the logistic regression
\begin{align*}
\text{E}[Y_{ij}|D_{ij},X_{ij}] = \Lambda( D_{ij}\theta_0 + X_{ij}' \beta_0).
\end{align*}
We thus obtain a dyadic sample $\{(Y_{ij},D_{ij},X_{ij}')'\}_{(i,j) \in \overline{[N]^2}}$, which we assume is observed by a researcher.
For the true parameter values, set $\theta_0 = 1$, and let the $c$-th coordinate of $\beta_0$ be $2(-2)^{-c}$ for $c \leq \lfloor \sqrt{N} \rfloor$ and 0 otherwise.
We will experiment with various sample sizes $N$ as well as various dimensions of $X_{ij}$ across sets of simulations.
The number of Monte Carlo iterations is set to 2,500 for each set of simulations.

Table \ref{tab:simulation_results} summarizes Monte Carlo simulation results.
The top panel of the table shows results for the conventional double/debiased machine learning without the dyadic cross fitting or the dyadic robust standard errors.
The bottom panel of the table shows results for our proposed dyadic double/debiased machine learning with the dyadic cross fitting and the dyadic robust standard errors.
Displayed are simulation results that vary with the sample size $N$, the dimension $\text{dim}(X)$ of $X_{ij}$, and the number $K$ of folds for the dyadic cross fitting.
Displayed statistics include 
the simulation mean of $\check\theta$, 
the simulation bias of $\check\theta$,
the simulation standard deviation of $\check\theta$,
the simulation root mean squared error of $\check\theta$,
the simulation inter-quartile range of $\check\theta$,
and the simulation coverage frequencies for nominal sizes of 90\% and 95\%.

\begin{table}
	\centering
	\scalebox{0.82}{
		\begin{tabular}{cccccccccccccc}
		\hline\hline
Method & $N$	& $\text{dim}(X)$	& $K$ & True & Mean	& Bias	& SD	& RMSE	& Q25	& Q50	& Q75	& 90\%	& 95\%\\
		\hline
Conventional DML &
50	& 25	& 5		& 1.000	& 1.144	& 0.144	& 0.316	& 0.347	& 0.928	& 1.139	& 1.354	& 0.363	& 0.422\\
Conventional DML &
100	& 50	& 5		& 1.000	& 1.148	& 0.148	& 0.220	& 0.265	& 0.999	& 1.141	& 1.293	& 0.233	& 0.277\\
		\hline
Conventional DML &
50	& 25	& 10	& 1.000	& 1.145	& 0.145	& 0.316	& 0.348	& 0.929	& 1.139	& 1.356	& 0.360	& 0.421\\
Conventional DML &
100	& 50	& 10	& 1.000	& 1.149	& 0.149	& 0.220	& 0.265	& 0.998	& 1.143	& 1.294	& 0.234	& 0.274\\
		\hline
Conventional DML &
50	& 50	& 5		& 1.000	& 1.253	& 0.253	& 0.316	& 0.405	& 1.044	& 1.253	& 1.464	& 0.332	& 0.392\\
Conventional DML &
100	& 100	& 5		& 1.000	& 1.252	& 0.252	& 0.222	& 0.336	& 1.100	& 1.248	& 1.404	& 0.170	& 0.209\\
		\hline
Conventional DML &
50	& 50	& 10	& 1.000	& 1.256	& 0.256	& 0.316	& 0.407	& 1.044	& 1.257	& 1.465	& 0.324	& 0.388\\
Conventional DML &
100	& 100	& 10	& 1.000	& 1.254	& 0.254	& 0.222	& 0.338	& 1.102	& 1.249	& 1.406	& 0.171	& 0.206\\
		\hline\hline
		\\
		\hline\hline
Method & $N$	& $\text{dim}(X)$	& $K$ & True & Mean	& Bias	& SD	& RMSE	& Q25	& Q50	& Q75	& 90\%	& 95\%\\
		\hline
Dyadic DML &
50	& 25	& 5		& 1.000	& 1.059	& 0.059	& 0.470	& 0.474	& 0.746	& 1.052	& 1.369	& 0.908	& 0.950\\
Dyadic DML &
100	& 50	& 5		& 1.000	& 1.045	& 0.045	& 0.288	& 0.292	& 0.848	& 1.040	& 1.236	& 0.901	& 0.946\\
		\hline
Dyadic DML &
50	& 25	& 10	& 1.000	& 1.047	& 0.047	& 0.518	& 0.520	& 0.688	& 1.034	& 1.394	& 0.922	& 0.965\\
Dyadic DML &
100	& 50	& 10	& 1.000	& 1.039	& 0.039	& 0.303	& 0.305	& 0.825	& 1.042	& 1.239	& 0.916	& 0.957\\
		\hline
Dyadic DML &
50	& 50	& 5		& 1.000	& 1.113	& 0.113	& 0.522	& 0.534	& 0.736	& 1.111	& 1.460	& 0.897	& 0.953\\
Dyadic DML &
100	& 100	& 5		& 1.000	& 1.105	& 0.105	& 0.304	& 0.322	& 0.903	& 1.105	& 1.314	& 0.883	& 0.940\\
		\hline
Dyadic DML &
50	& 50	& 10	& 1.000	& 1.114	& 0.114	& 0.582	& 0.593	& 0.715	& 1.101	& 1.497	& 0.908	& 0.964\\
Dyadic DML &
100	& 100	& 10	& 1.000	& 1.095	& 0.095	& 0.320	& 0.334	& 0.880	& 1.093	& 1.316	& 0.908	& 0.958\\
		\hline\hline
		\end{tabular}
	}
	\caption{Simulation results based on 2,500 Monte Carlo iterations. Displayed are the sample size $N$, dimension $\text{dim}(X)$ of $X_{ij}$, number $K$ of folds for the dyadic cross fitting, simulation mean of $\check{\theta}$, simulation bias of $\check{\theta}$, simulation standard deviation of $\check{\theta}$, simulation root mean square error of $\check{\theta}$, simulation inter quartile range of $\check{\theta}$, and coverage frequency for the nominal sizes of 90\% and 95\%. 
	The top panel of the table displays results for the conventional DML without dyadic cross fitting or dyadic robust standard errors, and the bottom panel of the table displays results for our proposed method of dyadic DML with dyadic cross fitting and dyadic robust standard errors.}
	${}$\label{tab:simulation_results}
\end{table}

First, compare the top panel and bottom panel of the table in terms of simulation statistics of the estimates. 
Notice that the conventional DML (without accounting for dyadic data) produces biased estimates of $\check{\theta}$, while the dyadic DML produces less biased estimates.
Recall that the conventional cross fitting works under independent sampling, while dyadic data are not independent.
As such, for dyadic data, conventional cross fitting may not be able to remove over-fitting biases.
On the other hand, the simulation results support the claim that dyadic cross fitting accounts for dyadic dependence and is successful in removing over-fitting biases.

Second, compare the top panel and bottom panel of the table now in terms of simulation coverage frequencies. 
Notice that the conventional DML (without accounting for dyadic data) suffers from severe under-coverage especially for larger sample sizes, while the dyadic DML maintains correct coverage.
This difference may be imputed to two sources.
One source of this difference is that estimates by the conventional DML are biased while those by the dyadic DML is less biased as discussed in the previous paragraph.
As the sample size increases, it becomes less likely that shorter confidence intervals centered around the biased estimates $\check{\theta}$ contain the true parameter value $\theta_0$.
Another source of the difference is that the conventional standard error without accounting for dyadic sampling is simply too small compared to the dyadic robust standard error.

In summary, the dyadic DML produces less biased estimates and asymptotically accurate coverage.
Failure to account for dyadic dependence in cross fitting results in biased estimates, while
failure to account for dyadic dependence in cross fitting and/or standard errors results in severe under-coverage.

\section{The Determinants of FTAs}

In this section, we reconsider the classic  network setting: the formation of free trade agreements across countries.  We compare our proposed method of dyadic DML to traditional methods for analyzing determinants of free trade agreements (FTA) using both parsimonious and rich empirical models.  In this sense, we contribute towards addressing the ``inherent complexity and ambiguity'' of the ``pure economic theory of trading blocs,'' \citep[][p. 60]{krugman1993regionalism} and add to the literature confirming the assertion ``detailed empirical work is the right direction'' needed to make headway.
\citet{baier2004economic} pioneered empirical studies in this field, and a large number of subsequent studies have followed since that time based on their foundational specification.
We aim to strengthen our understanding of the conclusions produced in this literature by conducting a new empirical analysis with the following econometric details in mind:
1. we control for high-dimensional covariates in order to mitigate unobserved endogeneity or unobserved confoundedness; and
2. we account for the dyadic dependence inherent in trade data and, for the first time, report standard errors robust to this source of misleading inference.

Consider the empirical model
\begin{align*}
\Ep[Y_{ij}|D_{ij},X_{ij}] = \Lambda\left(D_{ij}\theta + X_{ij}'\beta\right)
\text{ for } (i,j) \in \overline{[N]^2},
\end{align*}
where
$Y_{ij}$ is the dummy variable if there is a bilateral FTA between $i$ and $j$ as of year 2000.  The parsimonious specification includes four classic explanatory variables
$(D_{ij},X_{ij}')'$: (A)
the logarithm of the population-weighted bilateral distance between $i$ and $j$ in kilometers,
(B) the sum of the logarithms of the per-capita GDPs of $i$ and $j$ in the baseline year,
(C) the absolute difference of the logarithms of the per-capita GDPs between $i$ and $j$ in the baseline year, and
(D) the absolute difference of the logarithms of the capital-labor ratios between $i$ and $j$ in the baseline year.  Following \citet{baier2004economic} we complete our most parsimonious specification by including in the square of the difference of the logarithms of the capital-labor ratio. 

Bilateral distance is a standard proxy for trade costs in this literature and helps capture the fact smaller bilateral trade costs raise the benefits from a trade agreement, ceteris paribus.  As highlighted by \citet{baier2004economic}, a standard prediction of economic theory is that ``[t]he net gain from an FTA between two countries increases as the distance between them decreases.'' Also see \citet{krugman1991move} and \citet{frankel1993continental,frankel1995trading,frankel1996regional}.
This hypothesis motivates (A). 

Market size is measured by the sum of log country-level GDPs, conditional on trade costs, while similarity in economic development is captured by the absolute difference of log GDPs. The former reflects the common finding that ``[t]he net welfare gain from an FTA between a pair of countries increases the
larger are their economic sizes (i.e. average real GDPs),'' while the latter investigates the notion that ``[t]he net welfare gain from an FTA between a pair of countries increases the more similar are their economic sizes (i.e. real GDPs).'' Also see \citet{krugman1998comment}. These hypotheses motivate (B) and (C). 

Last, we also investigate the degree to which differences in factor endowments increase the likelihood of an FTA through the inclusion of the difference of the log capital-labor ratio across potential trade partners. \citet{baier2004economic} state that ``[t]he net welfare gain from an FTA between a pair of countries increases with wider relative factor endowments, but might eventually decline due to increased specialization.'' Accordingly, we include the absolute difference of the log capital-labor ratio to capture the notion that countries with a greater difference in initial endowments are more likely to have greater differences in comparative advantage across industries and, as such, greater potential gains from trade with each other. This hypothesis motivates (D).  As in \citet{baier2004economic} we also include the square of the difference of the log capital-labor ratio as a control for a diminishing impact of differences in factor endowments. 


Of course, there are wide host of additional bilateral characteristics which affect realized trade flows across countries and, as such, the likelihood of successfully forming an FTA.  For instance, research confirms that time differences, colonial ties, language or ethic overlap, common legal origins, political ties, among many other country-pair characteristics, affect trade flows.  Indeed, to the extent that these co-variates reflect differential market access, trade costs or the potential future benefits from trade agreements, there exclusion from existing work may present a source of omitted variable bias.

To investigate the impact of restricting attention to a parsimonious specification, we add a wide set of additional controls to our benchmark empirical model including the time difference between $i$ and $j$ in hours, a religious proximity index and a wide set of dummy variables capturing whether $i$ and $j$ have ever been in a colonial relationship, whether they have been in a colonial relationship after 1945, whether they had a common colonizer after 1945, whether they share a common official or primary language, whether there is a common language spoken by at least 9\% of the population in both countries, whether one country is current or former hegemon of the other, whether $j$ is current or former hegemon of $i$, whether they have ever been two colonies of the same empire, whether they shared a common legal origin before or after transition to independent statehood.  To saturate our set of potential confounders we also include all of the powers and interactions for each of these variables -- see Appendix \ref{sec:data_appendix} for details of how we construct each co-variate. In total, there are more than 140 explanatory variables in the rich specifications that we consider, and these dimensions are high relative to the effective sample size $N$ of the dyad which consists of $N=229$ economies. Following \citet{baier2004economic}, we set year 1960 as the baseline year for variable measurement, which is roughly the year when FTAs started to be signed across across global trading partners. 


The dependent variable reflects whether any two countries have an existing FTA and comes from the data used by \citet{head2010erosion} and \citet{baier2009estimating}.
Distance and gross domestic product data are retrieved from the CEPII Distance Dataset and World Development Indicators (World Bank), respectively.
The information about colonial relationships, languages, wars, hegemon relationships, and sibling relationships come from the data used by \citet{head2010erosion}, while the nature of bilateral legal relationships and religious proximity are respectively reported in \citet{la2008economic} and \citet{disdier2007je}. See the data descriptions in Appendix \ref{sec:data_appendix} for further details of this data set.


Table \ref{tab:empirical_results} summarizes estimation and inference results based on 50 iterations of resampled cross fitting.
In this table, we report estimates and standard errors for the coefficients of log population-weighted distance, the sum of the log GDPs, the absolute difference of the log GDPs, and the difference of the log capital-labor ratios. These four explanatory variables correspond to hypotheses (A), (B), (C) and (D) suggested in the previous paragraph.
The first column (I) reports results based on the simple, parsimonious logistic regression not including high-dimensional regressors and not accounting for dyadic sampling.
The second column (II) reports results based on the simple logistic regression, while including high-dimensional regressors.  
It does not account for dyadic sampling. Columns (III) and (IV) report results based on conventional machine learning with $K=5$ and 10, respectively.  
The last two columns (V) and (VI) report results based on our proposed dyadic machine learning where we again set $K=5$ and 10, respectively.

\begin{table}
	\centering
		\begin{tabular}{lcccccccccc}
		\hline\hline
			Dependent variable: && Logit && Full Logit && \multicolumn{2}{c}{Conventional ML} && \multicolumn{2}{c}{Dyadic ML}\\
			\cline{3-3}\cline{5-5}\cline{7-8}\cline{10-11}
			free trade agreement && (I) && (II) && (III) & (IV) && (V) & (VI)\\
			\hline
			(A) Distance    
			&&-1.690 &&-1.358 &&-1.662 &-1.660 &&-1.515 &-1.762\\
			
			&&(0.046)&&(0.075)&&(0.081)&(0.079)&&(0.111)&(0.115)\\
			(B) Size (Sum of log GDP)
			&& 0.236 && 0.343 && 0.359 & 0.360 && 0.263 & 0.244 \\
			&&(0.013)&&(0.020)&&(0.008)&(0.007)&&(0.043)&(0.035)\\
			(C) Similarity ($\Delta$ log GDP)
			&&-0.003 &&-0.004 &&-0.004 &-0.004 &&-0.001 & 0.001\\
			
			&&(0.015)&&(0.018)&&(0.014)&(0.014)&&(0.002)&(0.002)\\
			(D) Rel. Factor Endowments
			&& 0.231 && 0.187 &&-0.460 &-0.460 &&-0.432 &-0.396\\
			 \ \ \ \ \ ($\Delta$ log $K/L$)
			&&(0.060)&&(0.072)&&(0.143)&(0.143)&&(0.326)&(0.362)\\
			\\
			Effective sample size
			&& 13,027 && 13,027 && 13,027 & 13,027 && 229 & 229\\
			Dimension $\text{dim}(D',X')'$
			&& 5 && 141 && 141 & 141 && 141 & 141\\
			Number $K$ of folds
			&& N/A && N/A && 5 & 10 && 5 & 10\\
		\hline\hline
		\end{tabular}
	\caption{Estimation and inference results based on 50 iterations of resampled cross fitting. Displayed are estimates and standard errors for the coefficients of (A) the logarithm of population-weighted distance, (B) the sum of the logarithms of GDPs in the baseline year, (C) the absolute difference of the logarithms GDPs in the baseline year and (D) the difference of the logarithms of capital-labor ratios. The first column (I) reports results based on the simple, parsimonious logistic regression without the high-dimensional controls, the second column (II) reports results based on the simple logistic regression with all the covariates. The next two columns (III) and (IV) report results based on conventional machine learning. The last two columns (V) and (VI) report results based on our proposed dyadic machine learning.}
	${}$\label{tab:empirical_results}
\end{table}

We observe the following sharp results. First, the simple logit replicates the sign on each of the included co-variates in \citet{baier2004economic}. Likewise, each co-variate is also statistically significant at conventional levels, save one.\footnote{This extends to the square of the log difference in capital labor ratios.  The estimated coefficient for this co-variate is -0.169 with a standard error of 0.021 in column (I), and -0.117 with as standard error of 0.036 in column (II).  We do not report the estimated coefficient for the square of the log difference in capital labor ratios in columns (III)-(VI) since it is treated as a control variable in $X_{ij}$ rather than part of $D_{ij}$ itself.}  While \citet{baier2004economic} find that greater economic similarity, measured as the difference in log GDPs, is a significant determinant of FTAs, our simple logit suggests that this coefficient is small and insignificantly different from zero. This likely reflects differences in the underlying data.  Our model is based on FTAs signed by the year 2000, four years after the analysis conducted in \citet{baier2004economic} and a period when a significant number of FTAs were signed.  Further, improved data availability allows us to extend their analysis of FTA formation among 54 countries to a setting with 167 countries, including many additional developing countries.\footnote{\citet{baier2004economic} also use a probit model rather than a logit model. We investigated whether difference in methodology resulted in significantly different results in column (I).  Rather, probit results were very similar to the logit specification.  We focus on the logit since it is most directly comparable to the conventional ML and dyadic ML approaches in columns (III)-(VI).}

Second, rows (A), (B) and (D) each highlight a separate, but equally important, feature of dyadic machine learning in this classic setting.  For instance, in row (A) the dyadic ML indicates that the coefficient of (A) the logarithm of the population-weighted distance is negative and significantly different from zero. In this sense, the estimated coefficient is similar to that estimated by the parsimonious logit, the full logit, and conventional ML. However, the estimated standard errors in columns (V) and (VI) are 48-53 percent larger than those in column (II), the rich logit specification, and 37-46 percent larger than those in columns (III)-(IV), where conventional machine learning is employed. Nonetheless, our estimates support the hypothesis that as distance, and trade costs, between two countries decrease the net gain from an FTA between them rises \citep{krugman1991move,frankel1993continental,frankel1995trading,frankel1996regional,baier2004economic}. The dyadic ML estimates again confirm that this common empirical conclusion is robust to the inclusion of high-dimensional covariates and controlling for dyadic dependence among FTA formation data.

In row (B) we observe that the dyadic ML also indicates that the coefficient of (B) economic size, measured as the sum of log GDP, is positive and significantly different from zero. Again, this finding is similar to that indicated in the parsimonious logit, the full logit, and the conventional ML and supports the hypothesis that greater economic size between a pair of countries increases the net welfare gain from an FTA \citep{krugman1998comment,baier2004economic}.  However, our results also suggest that ignoring dyadic dependence is likely to overstate the importance of market size in the formation in FTAs.

Specifically, the magnitude of the estimated coefficient on the sum of the log GDPs changes significantly across columns and is consistent with claim that conventional methods may yield biased estimates in rich specifications.  For instance, in columns (II), (III) and (IV), the estimation routines all return estimates which are substantially larger than those in columns (V) and (VI) where we control dyadic dependence.  These empirical results are consistent with the findings in our Monte Carlo simulations: conventional ML was likely to produce biased estimates in rich specifications.  In Table 2, we find that FTA market size is likely to receive outsized importance; the estimated coefficients on the size variable in columns (V) and (VI) are 27-32 percent smaller than those columns (III) and (IV).

Row (D) illustrates a different source of bias.  In particular, the simple logit indicates that the coefficient on (D) relative factor endowments, measured as the difference of the log capital-labor ratio, is positive and significant in columns (I) and (II).  This result is consistent with existing empirical results that confirm that differences in factor endowments are important determinants of FTAs.  However, once we implement  machine learning variable selection in the high-dimensional controls (columns (III)-(IV)), just the opposite pattern emerges: greater similarity in capital-labor ratios are found to \textit{encourage} FTA formation when we employ conventional ML.  In this sense, benchmark analyses which do not directly account for the inclusion and confoundedness associated with high dimensional co-variates are likely to deliver misleading results.  For instance, \citet{baier2004economic} do control for a number of the above co-variates in the robustness checks of their seminal analysis and, like our results in column (II), find that the inclusion of additional co-variates in the simple logit has little impact on their benchmark findings.

At the same time allowing for high dimensional co-variates and dyadic dependence in columns (V) and (VI) also has a large impact on the recovered standard errors. Indeed, in this case, it changes preceding conclusions: while we maintain the positive estimated coefficient on the relative factor endowments variable after controlling for dyadic dependence, the standard error more than doubles. We can no longer conclude that relative factor endowments are an important determinant of FTA formation, overturning this classic result. More generally, accounting for dyadic dependence significantly increases the recovered standard errors for all of our explanatory variables by at least 37 percent relative to conventional ML, which was already relatively conservative.

Finally, all the methods indicate that the coefficient of (C) economic similarity, measured as the absolute difference of log GDP, is consistently negative, but statistically insignificant.  In this sense, we do not find strong evidence, regardless of methodology, to support the hypothesis that greater economic similarity between a pair or countries increases the likelihood of an FTA being formed between them.

In sum, we confirm that trade costs and market size are key determinants of FTA formation, even after controlling for high-dimensional covariates and robustly accounting for the dyadic dependence in data. Allowing for high-dimensional covariates and robustly accounting for the dyadic dependence are both generally found to substantially increase standard errors and reduce $t$-statistics. Our estimates also suggest that the importance of market size in FTA formation may be biased and overstated when employing conventional methods, including conventional ML, and a rich specification. Indeed, we document that standard estimation approaches, such as logit or conventional ML, may lead researchers to conclude that larger differences in relative factor endowment may encourage or discourage FTA formation. Accounting for dyadic dependence, in contrast, increases standard errors sufficiently that we can no longer conclude that relative factor endowments are an important determinant of FTA formation. In this sense, we further confirm that the use of econometric methods which do not allow for high-dimensional controls and/or dyadic dependence can lead to misleading conclusions, either in terms of point estimates or standard errors in the context of FTA formation models.  Last, in contrast to the existing literature, we do not find evidence that differences in economic similarity or factor endowments are particularly important determinants of FTA formation regardless of which estimation approach we employ.

\section{Summary and Discussions}

This paper presents new methods and theories for the use of machine learning when data are dyadic.
A novel dyadic cross fitting algorithm is proposed to remove over-fitting biases under arbitrary dyadic dependence.
Together with the use of Neyman orthogonal scores, this dyadic cross fitting method enables double/debiased machine learning for dyadic data.
Applying the general results (Theorem \ref{theorem_DDML_non_linear}), we demonstrate how to estimate and conduct inference for high-dimensional logit dyadic link formation models.

With this method applied to empirical data of international economic networks, we reexamine FTA determinants as links formed in the dyad composed of world economies.
We reconfirm two important theoretical implications suggested by the international trade literature, even after robustly accounting for both dyadic dependence and high-dimensional controls.
Namely,
(A) a greater distance between economies makes an FTA less beneficial and thus makes an FTA less likely to be formed; and
(B) larger sizes of economies make an FTA more beneficial and thus make an FTA more likely to be formed.
In the context of a rich specification, we further document that conventional methods return biased FTA market size coefficients. Similarly, while standard approaches suggest larger differences in relative factor endowment are an important determinant of FTA formation, once we account for dyadic dependence we can no longer support this standard conclusion.

Motivated by the application to the dyadic link formation model for the analysis of FTA determinants, we focus our econometric methodology on unconditional moment restrictions in this paper.
In other applications, such as the instrumental variable regression in dyadic data, conditional moment restrictions are more relevant \citep[cf.][]{AiChen2003,AiChen2007,ChenLintonKeilegom2003,ChenPouzo2015}.  
We conjecture that our proposed method with dyadic cross fitting will extend to estimation and inference for models based on conditional moment restrictions.
Formal theoretical investigation in this important direction is left for future research, given the length of the current paper.

\vspace{2cm}
\appendix
\section*{Appendix}

\section{The Case of Linear Scores}\label{sec:linear_score_case}

The main text presents the method of dyadic DML with a general class of nonlinear and nonseparable form of the score function $\psi$.
On the other hand, some of the most commonly used econometric models entail linear and additively separable score functions.
In this appendix section, we focus on the cases with linear Neyman orthogonal scores $\psi$ of the form
\begin{align}
\psi(w;\theta,\eta)=\psi^a(w;\eta)\theta +\psi^b(w;\eta), \text{ for all $w\in \supp(W)$, $\theta\in\Theta$, $\eta\in T$,} 
\label{eq:linear_score}
\end{align}
and develop asymptotic theoretical results under an alternative set of assumptions.

\subsection{Two Examples}\label{sec:two_examples}

Two of the most frequently used instances of linear Neyman orthogonal scores $\psi$ satisfying \eqref{eq:existance_condition}, \eqref{eq:Neyman_orthogonal_condition}, and \eqref{eq:linear_score} are presented as examples below.

\begin{example}[High-Dimensional Linear Models]
Let $W_{ij}=(Y_{ij},D_{ij},X_{ij}')'$, where $Y_{ij}$ and $D_{ij}$ are two random variables, and $X_{ij}$ is a high-dimensional random vector.
Consider the linear model
\begin{align*}
\Ep[Y_{ij}|D_{ij},X_{ij}]= D_{ij}\theta_0 + X_{ij}'\beta_0 \text{ for }(i,j)\in \overline{[N]^2},
\end{align*}
where $\theta$ is a parameter of interest and $\beta$ is a nuisance parameter.
A linear Neyman orthogonal score $\psi$ in this model is given by
$$
\psi(W_{ij};\theta,\eta)=\{Y_{ij}-D_{ij}\theta-X_{ij}'\gamma\}(D_{ij}-X_{ij}'\delta),
$$
where $\eta = (\beta',\gamma',\delta')'$.
See \citet[][Section 4.1]{CCDDHNR18} for more details.
\qed
\end{example}

\begin{example}[High-Dimensional Linear IV Models]
Let $W_{ij}=(Y_{ij},D_{ij},X_{ij}',Z_{ij})'$, where $Y_{ij}$, $D_{ij}$ and $Z_{ij}$ are three random variables, and $X_{ij}$ is a high-dimensional random vector.
Consider the linear IV model
\begin{align*}
\Ep[Y_{ij}|D_{ij},X_{ij}]&= D_{ij}\theta_0 + X_{ij}'\beta_0 &&\text{ for }(i,j)\in \overline{[N]^2} \text{ and }
\\
\Ep[Z_{ij}|X_{ij}]&= X_{ij}'\delta &&\text{ for }(i,j)\in \overline{[N]^2},
\end{align*}
where $\theta$ is a parameter of interest and $(\beta',\delta')'$ is a nuisance parameter.
A linear Neyman orthogonal score $\psi$ in this model is given by
$$
\psi(W_{ij};\theta,\eta)=\{Y_{ij}-D_{ij}\theta-X_{ij}'\gamma\}(Z_{ij}-X_{ij}'\delta),
$$
where $\eta = (\beta',\gamma',\delta')'$.
See \citet[][Section 4.2]{CCDDHNR18} for more details.
\qed
\end{example}

\subsection{Asymptotic Normality and Variance Estimation}\label{sec:asymptotic_normality_and_variance_estimation}

Let $c_0>0$, $c_1>0$, $s>0$, and $q\ge 4$ be some finite constants with $c_0\le c_1$. 
Let $\{\delta_N\}_{N\ge 1}$, $\{\Delta_N\}_{N\ge 1}$, and $\{\tau_N\}_{N\geq 1}$ be sequences of positive constants that converge to zero such that $\delta_N \ge N^{-1/2}$. 
Let $K\ge 2$ be a fixed integer. 
With these notations, consider the following two assumptions.

\begin{assumption}[Linear Neyman Orthogonal Score]\label{a:linear_orthogonal_score}
The following conditions hold for all $P\in \mathcal P_N$ for all $N \ge 4$.
\begin{enumerate}[(i)]
\item The true parameter value $\theta_0$ satisfies (\ref{eq:existance_condition}).
\item $\psi$ is linear in the sense that it satisfies (\ref{eq:linear_score}).
\item The map $\eta \mapsto \Ep[\psi(W_{12};\theta,\eta)]$ is twice continuously Gateaux differentiable on $T$.
\item $\psi$ satisfies either the Neyman orthogonality condition (\ref{eq:Neyman_orthogonal_condition}) or more generally 
the Neyman $\lambda_N$ near orthogonality condition at $(\theta_0,\eta_0)$ with respect to a nuisance realization set $\mathcal T_N\subset T$, i.e.,
\begin{align*}
\lambda_N:=\sup_{\eta \in \mathcal T_N}\Big\| \partial_\eta \Ep\psi(W_{12};\theta_0,\eta_0)[\eta-\eta_0] \Big\|\le \delta_N N^{-1/2}.
\end{align*}
\item The identification condition holds in the sense that the singular values of the matrix $J_0:=\Ep[\psi^a(W_{12};\eta_0)]$ are between $c_0$ and $c_1$.
\end{enumerate}
\end{assumption}

\begin{assumption}[Score Regularity and Nuisance Parameter Estimators]\label{a:regularity_nuisance_parameters}
Let $K$ be a fixed integer.
The following conditions hold for all $P\in \mathcal P_N$ for all $N \ge 4$.
\begin{enumerate}[(i)]
\item Given random subsets $I\subset [N]$  such that $|I|=\lfloor N/K\rfloor$, the nuisance parameter estimator $\hat \eta=\hat\eta((W_{ij})_{(i,j)\in \overline{([N]\setminus I)^2}}) $ belongs to the realization set $\mathcal T_N$ with probability at least $1-\Delta_N$, where $\mathcal T_N$ contains $\eta_0 $.
\item The following moment conditions hold :
\begingroup
\allowdisplaybreaks
\begin{align*}
\sup_{\eta\in \T_N}(\Ep[\|\psi(W_{12};\theta_0,\eta)\|^q])^{1/q} \vee \sup_{\eta\in \T_N}(\Ep[\|\psi^a(W_{12};\eta)\|^q])^{1/q}\le c_1.
\end{align*}
\endgroup
\item The following conditions on the rates $r_N$, $r_N'$ and $\lambda_N'$ hold:
\begingroup
\allowdisplaybreaks
\begin{align*}
r_N:=& \sup_{\eta\in \T_N}
\|\Ep[\psi^a(W_{12};\eta)]-\Ep[\psi^a(W_{12};\eta_0)]\|\le \delta_N,\\
r_N':=& \sup_{\eta\in \T_N}
(\Ep[\|\psi(W_{12};\theta_0,\eta)-\psi(W_{12};\theta_0,\eta_0)\|^2])^{1/2}\le \delta_N,\\
\lambda_N'= & \sup_{r\in (0,1),\eta\in \T_N}\|\partial^2_r \Ep[\psi (W_{12};\theta_0,\eta_0+r(\eta-\eta_0)) ] \|\le \delta_N/\sqrt{N}.
\end{align*}
\endgroup
\item All eigenvalues of the matrix
\begingroup
\allowdisplaybreaks
\begin{align}
\label{true_gamma_linear}
\Gamma:=\Ep[\psi(W_{12};\theta_0,\eta_0)\psi(W_{13};\theta_0,\eta_0)']+\Ep[\psi(W_{12};\theta_0,\eta_0)\psi(W_{31};\theta_0,\eta_0)']
\\\nonumber
+\Ep[\psi(W_{21};\theta_0,\eta_0)\psi(W_{13};\theta_0,\eta_0)']+\Ep[\psi(W_{21};\theta_0,\eta_0)\psi(W_{31};\theta_0,\eta_0)']
\end{align}
\endgroup
are bounded from below by $c_0$. 
\end{enumerate}
\end{assumption}

\noindent
The last two assumptions are counterparts of Assumptions 3.1 and 3.2 in \citet{CCDDHNR18}, as well as Assumptions 2 and 3 in \citet{CKMS2019}, except that the data in this paper are allowed to exhibit dyadic dependence rather than i.i.d. or multiway cluster dependence.

Under these conditions, we guarantee that the dyadic machine learning estimator is asymptotically root-$N$ Gaussian and that its asymptotic variance can be consistently estimated, as formally stated in the following two lemmas. Denote $W_{i0}=\tau_n(U_i,U_0,U_{\{i,0\}})$ and $W_{0i}=\tau_n(U_0,U_i,U_{\{i,0\}})$ where $U_i$ and $U_{\{i,0\}}$ are $\text{Unif}[0,1]$ random variables independent from each other and from the data.

\begin{lemma}[Dyadic DML for Linear Scores]\label{theorem:main_result_linear}
Suppose that Assumptions \ref{a:sampling}, \ref{a:linear_orthogonal_score} and \ref{a:regularity_nuisance_parameters} are satisfied.
If $\delta_N\ge N^{-1/2}$ for all $N\ge 4$, then
\begin{align}
\sqrt{N}\sigma^{-1}(\tilde \theta - \theta_0)
&=
\frac{\sqrt{N}}{K} \sumk\Enk \bar\psi(W) +\Opn(\rho_N)
\nonumber\\
&=\frac{1}{\sqrt{N}}\sum_{i=1}^N \left\{\Ep [\bar\psi(W_{i0})\mid U_i] +\Ep [\bar\psi(W_{0i})\mid U_i]\right\}+\Opn(\rho_N) \leadsto N(0,I_{d_\theta})
\label{eq:main_result_linear}
\end{align}
holds uniformly over $P\in\mathcal P_N$, where the order of remainder follows
$
\rho_N 
\lesssim \delta_N,
$
the influence function takes the form $\bar \psi(\cdot):=-\sigma^{-1}J_0^{-1} \psi(\cdot;\theta_0,\eta_0)$, 
and the variance is given by
\begin{align}
\sigma^2:=J_0^{-1}\Gamma (J_0^{-1})'. \label{eq:population_variance}
\end{align}
\end{lemma}
In view of the second line in \eqref{eq:main_result_linear}, observe that the asymptotic linear representation is independent of the choice of $K$ and also independent of realizations of the folds in the dyadic cross fitting. 
In addition, although the estimator itself is formed by the scores consisting only of a part of all the $i$-$j$-pairs, this linear representation implies that we in fact utilize all the underlying independent latent random variables $(U_i)_{i=1}^N$.
To make use of this lemma for statistical inference in practice, it remains to consistently estimate $\sigma^2$.
The sample-counterpart variance estimator can be given by
\begin{align}
\hat \sigma^2=\hat J^{-1}\hat\Gamma(\hat J^{-1})',
\end{align}
where
\begin{align}
\label{estimator_j}
\hat J:=\frac{1}{K}\sumk \Enk[\psi^a(W;\hat \eta_{k})]
\end{align}
and  $\hat\Gamma$ is defined in \eqref{estimator_gamma}.
This estimator $\hat\sigma^2$ is consistent for $\sigma^2$ under the same set of assumptions as that invoked in Lemma \ref{theorem:main_result_linear}, as formally stated in the following lemma.

\begin{lemma}[Variance Estimator]\label{theorem:variance_estimator_linear}
Under the assumptions required by Lemma \ref{theorem:main_result_linear}, we have
\begin{align*}
\hat \sigma^2=\sigma^2 +\Op(\rho_N).
\end{align*}
Furthermore, the statement of Lemma \ref{theorem:main_result_linear} holds true with $\hat \sigma^2$ in place of $\sigma^2$.
\end{lemma}

\noindent
Lemmas \ref{theorem:main_result_linear} and \ref{theorem:variance_estimator_linear} can be used together for constructing confidence intervals as formally stated below.

\begin{corollary}\label{corollary:inference_t-test}
Suppose that all the assumptions required by Lemma \ref{theorem:main_result_linear} are satisfied.
Let $r$ be a $d_\theta$-dimensional vector of constants. 
The $(1-a)$ confidence interval of $r'\theta_0$ given by
\begin{align*}
\text{CI}_a:=[r'\tilde \theta\pm \Phi^{-1}(1-a/2)\sqrt{r'\hat \sigma^2 r/N}]
\end{align*}
satisfies
\begin{align*}
\sup_{P\in\mathcal P_N}|P_P(r'\theta_0 \in \text{CI}_a)-(1-a)|\to 0.
\end{align*}

\end{corollary}

\section{Lemmas for Dyadic Asymptotic Theory}\label{sec:lemmas_dyadic}
In this section, we present some generic results for sample means of dyadic random arrays.
We use these generic results as auxiliary lemmas to prove our main theoretical results.

We first present a useful maximal inequality for empirical processes under dyadic sampling. It is built upon the Hoeffding type decomposition and the maximal inequalities in \cite{CCK14}, \cite{ChenKato2019}, and \cite{cattaneo2022}. Define 
\begin{align*}
\Gn f:=\frac{\sqrt{n}}{n(n-1)}\sumijn \{f(X_{ij})-\Ep[f(X_{12})]\}.
\end{align*}
\begin{lemma}[A local maximal inequality for dyadic empirical processes]
\label{lemma_maximal_inequality_dyadic_data}
Let $(X_{ij})_{(i,j)\in\overline{[n]^2}}$ be a sample from $\mathcal S$-valued jointly exchangeable random variables $(X_{ij})_{(i,j)\in\overline{\mathbb{N}^{+2}}}$. Assume the function class $\calF\ni f:\calS\to \Real$ with envelope $F\in L^q(P)$, $q\ge 4$, satisfies
$
\sup_{Q }N(\calF,\|\cdot\|_{Q,2},\varepsilon\|F\|_{Q,2})\le (A/\varepsilon)^v,\:\forall \:0<\varepsilon\le 1,
$
for some constants $A\ge e$ and $v\ge 1$, where the supremum is taken over the set of all finite discrete distributions on $\calS$. In addition, suppose that the function classes
\begin{align*}
\calG=\{\Ep[f(X_{12})|U_k=u]:k=1,2, f\in \calF\},\quad\calH=\{\Ep[f(X_{12})|U_1=u_1,U_2=u_2]: f\in \calF\}
\end{align*}
have envelopes $G(u):=\Ep[F(X_{12})|U_1=u]\vee\Ep[F(X_{12})|U_2=u]$ and $H(u_1,u_2):=\Ep[F(X_{12})|U_1=u_1,U_2=u_2]$, respectively. Assume that there exist finite constants 
$
 b_n \ge \overline \sigma_n>0
$
such that 
\begin{align*} 
&\|G\|_{P,q}\vee \|H\|_{P,q}\le b_n,\qquad \sup_{g\in\calG}\|g\|_{P,2}\vee \sup_{h\in\calH}\|h\|_{P,2}\le \overline \sigma_n.
\end{align*}
In addition, suppose that $\Ep[f(X_{12})]=0$ for all $f\in \calF$.
Then, we have
\begin{align*}
\Ep[\Gn f]\lesssim&\:\overline \sigma_n\sqrt{v\log(A\vee n)}
+ \frac{b_n v\log(A\vee n)}{n^{1/2-1/q}}.
\end{align*}

\end{lemma}

\noindent
A proof can be found in Appendix \ref{sec:lemma_maximal_inequality_dyadic_data}.

We now derive the asymptotic linear representation for a sum of jointly exchangeable arrays. The result is well-known but we include the derivation for completeness.
Recall that Assumption \ref{a:sampling}(i)(ii) implies the Aldous-Hoover-Kallenberg representation (e.g. \citealp[Theorem 7.22 and Corollary 7.35]{Kallenberg2006}), which yields that there exists a Borel measurable function $\tau_n$ such that
\begin{align*}
Z_{ij}=\tau_n (U_i,U_j,U_{\{i,j\}}).
\end{align*}
For any real-valued function $f$,
define $\Gn f := \frac{\sqrt{n}}{n(n-1)}\sum\limits_{(i,j)\in \overline{[n]^2}} f(Z_{ij})$. Consider $\Hn f:= \frac{\sqrt{n}}{n(n-1)}\sum\limits_{l \in [n]} \{\sum_{j\ne l }\\ \Ep[f(Z_{lj})|U_l] +
 \sum_{i\ne l}\Ep[f(Z_{il})|U_l]\}$, its H\'ajek projection on functions of each single $U_l$. We will show $\Hn f$ is asymptotically normal with mean zero and variance 
\begin{align*}
\V(\Hn f)= \Cov(f(Z_{12}),f(Z_{13}))+\Cov(f(Z_{12}),f(Z_{31}))+ \Cov(f(Z_{21}),f(Z_{13})) + \Cov(f(Z_{21}),f(Z_{31})) 
\\
+ o(1).
\end{align*}
In addition, a direct calculation shows
\begin{align*}
\V(\Gn f)=\Cov(f(Z_{12}),f(Z_{13}))+\Cov(f(Z_{12}),f(Z_{31}))+ \Cov(f(Z_{21}),f(Z_{13})) + \Cov(f(Z_{21}),f(Z_{31})) 
\\
+ o(1),
\end{align*}
and therefore $\frac{\V(\Hn f)}{\V(\Gn f)}=1 + o(1)$ follows from the assumption that $\V(\Gn f)$ is uniformly bounded away from zero over $n$.
\begin{lemma}[Asymptotic Linear Representation for a Dyadic Sum]
\label{lemma:hajek}
Suppose $(Z_{ij})_{(i,j)\in \overline{[n]^2}}$ are zero mean random vectors and such that
\begin{align*}
Z_{ij}=\tau_n (U_i,U_j,U_{\{i,j\}})
\end{align*}
holds for some Borel measurable $\tau_n$ and  $\{U_i, U_{\{i,j\}}: 1\le i\ne j\le n\}$ are some i.i.d. random variables. In addition, suppose that $\Cov(f(Z_{12}),f(Z_{13}))+\Cov(f(Z_{12}),f(Z_{31}))+ \Cov(f(Z_{21}),f(Z_{13})) + \Cov(f(Z_{21}),f(Z_{31}))$ is bounded away from zero uniformly in $n$.
Then
\begin{align*}
\V(\Hn f)=
\Cov(f(Z_{12}),f(Z_{13}))+\Cov(f(Z_{12}),f(Z_{31}))+ \Cov(f(Z_{21}),f(Z_{13})) + \Cov(f(Z_{21}),f(Z_{31}))
\\
+o(1)
\end{align*}
and thus $\Hn f = \Gn f +o_P(1)$.
\end{lemma}

\noindent
A proof can be found in Appendix \ref{sec:lemma:hajek}

\section{Proofs of the Main Results}

This section presents proofs of the main results.
We prove the results for the linear case first, followed by the results for the general nonlinear case presented in the main text, because the latter result uses the former result along with a linearization of the nonlinear score.

\subsection{Proof of Lemma \ref{theorem:main_result_linear}}\label{sec:theorem:main_result_linear}

\begin{proof}
Fix a sequence $\{P_N\}_{N\geq 1}$ such that $P_N\in\mathcal{P}_N$ for all $N$.
We use the short-hand notation $\Ikc$ to denote $\overline{([N]\setminus I_k)^2}$.
$[N]$ is partitioned equally into $K$ parts, which implies $|I_k|=|I|$ for all $k \in [K]$. Assume without loss of generality that $N$ is divisible by $K$. 
Let $\E_N$ denote the event $\hat \eta_{k} \in \T_N$ for all $k\in [K]$. 
Assumption \ref{a:regularity_nuisance_parameters} (i) implies $P(\E_N)\ge 1- K\Delta_N$.
Although the proof strategies for asymptotic theory in the existing DML literature do not apply here due to the dyadic clustering, we shall follow the proof structure of Theorem 3.1 in \citet{CCDDHNR18} and Theorem 1 in \cite{CKMS2019} for the sake of comparability.
\bigskip\\
\noindent \textbf{Step 1.}
This step is the main step, showing the linear representation and asymptotic normality for the proposed estimator. 
Denote 
\begingroup
\allowdisplaybreaks
\begin{align*}
&\hat J:=\frac{1}{K}\sumk \Enk[\psi^a(W;\hat \eta_{k})],
\qquad R_{N,1}:=\hat J - J_0,\\
&R_{N,2}:=\frac{1}{K} \sumk\Enk[\psi(W;\theta_0,\hat\eta_{k})]
- \frac{1}{K} \sumk\Enk [\psi(W;\theta_0,\eta_0)].  
\end{align*}
\endgroup
We will later show in Steps 2, 3, 4 and 5, respectively, that
\begingroup
\allowdisplaybreaks
\begin{align}
&\|R_{N,1}\|=\Opn(N^{-1/2} + r_N), \label{eq:step2}\\
&\|R_{N,2}\|=\Opn(N^{-1/2}r_N'+\lambda_N+\lambda_N'), \label{eq:step3}\\
&\left\|\frac{\sqrt{N}}{K}\sumk\Enk [\psi(W;\theta_0,\eta_0)]\right\|=\Opn(1), \qquad\text{and}\label{eq:step4}\\
&\|\sigma^{-1}\|=\Opn(1). \label{eq:step5}
\end{align}
\endgroup
Under Assumptions \ref{a:linear_orthogonal_score} and \ref{a:regularity_nuisance_parameters}, $N^{-1/2}+r_N\le \rho_N=o(1)$ and all singular values of $J_0$ are bounded away from zero, and so all singular values of $\hat J$ are bounded away from zero with $P_N$-probability at least $1-o(1)$. 
Therefore, with the same $P_N$ probability, the dyadic machine learning solution is uniquely written as
\begin{align*}
\tilde \theta =- \hat J^{-1} \frac{1}{K}\sumk \Enk[\psi^b(W;\hat \eta_{k})],
\end{align*}
and 
\begingroup
\allowdisplaybreaks
\begin{align}
\sqrt{N}(\tilde \theta - \theta_0)
=&
-\sqrt{N} \hat J^{-1}\frac{1}{K} \sumk\Big(\Enk[\psi^b(W;\hat \eta_{k})] +\hat J \theta_0\Big) \nonumber\\
=&-\sqrt{N} \hat J^{-1}\frac{1}{K} \sumk \Enk[\psi(W;\theta_0,\hat \eta_{k})] \nonumber\\
=&
-\Big(J_0 + R_{N,1}\Big)^{-1} \times \Big(
\frac{\sqrt{N}}{K}\sumk\Enk[\psi(W;\theta_0,\eta_0)]+\sqrt{N}R_{N,2}
\Big).
\label{eq:step1_1}
\end{align}
\endgroup
Using the fact that 
\begin{align*}
\Big(J_0 + R_{N,1}\Big)^{-1}  - J_0^{-1}=-(J_0 +R_{N,1})^{-1}R_{N,1} J_0^{-1},
\end{align*}
we have 
\begingroup
\allowdisplaybreaks
\begin{align*}
\|(J_0+R_{N,1})^{-1} - J_0^{-1}\| = &
\|(J_0 +R_{N,1})^{-1}R_{N,1} J_0^{-1}\|\le  \|(J_0 +R_{N,1})^{-1}\|\,\|R_{N,1}\|\, \|J_0^{-1}\|\\
=&\Opn(1)\Opn(N^{-1/2}+ r_N)\Opn(1)=\Opn(N^{-1/2}+ r_N).
\end{align*}
\endgroup
Furthermore, $r_N'+N^{1/2}(\lambda_N+\lambda_N')\le \rho_N=o(1)$, and it holds that 
\begingroup
\allowdisplaybreaks
\begin{align*}
\Big\|\frac{\sqrt{N}}{K}\sumk\Enk [\psi(W_{ij};\theta_0,\eta_0)]+\sqrt{N}R_{N,2}\Big\|
\le&
\Big\|\frac{\sqrt{N}}{K}\sumk\Enk [\psi(W_{ij};\theta_0,\eta_0)]\Big\| +\Big\|\sqrt{N} R_{N,2}\Big\|  \\
=&
\Opn(1) +\opn(1)=\Opn(1),
\end{align*}
\endgroup
where the first equality is due to (\ref{eq:step3}) and (\ref{eq:step4}).
Combining above the two bounds gives 
\begin{align}
\Big\|\Big(J_0 + R_{N,1}\Big)^{-1}  - J_0^{-1}\Big\|\times\Big\|\frac{\sqrt{N}}{K}\sumk\Enk [\psi(W_{ij};\theta_0,\eta_0)]+\sqrt{N}R_{N,2}\Big\| =&
\Opn(N^{-1/2}+ r_N)\Opn(1) \nonumber\\
=&
\Opn(N^{-1/2}+ r_N).
\label{eq:step1_2}
\end{align}
From (\ref{eq:step5}), (\ref{eq:step1_1}) and (\ref{eq:step1_2}), it therefore follows that
\begin{align*}
\sqrt{N}\sigma^{-1}(\tilde \theta - \theta_0)
&=
\underbrace{\frac{\sqrt{N}}{K} \sumk\Enk  \bar\psi(W)}_{=:\Gnk \bar\psi} +\Opn(\rho_N).
\end{align*} 
Applying Lemma \ref{lemma:hajek} to each $I_k$, we obtain the following independent linear representation for $\Gnk \bar\psi$
\begin{align*}
\Hnk \bar{\psi}
&= \frac{\sqrt{N}}{\sumk |I_k|(|I_k|-1)}\sumk\sum_{l \in I_k}\left \{\sum_{j \in I_k,j\ne l } \Epn[\bar{\psi}(W_{lj})|U_l] + \sum_{i \in I_k,i\ne l}\Epn[\bar{\psi}(W_{il})|U_l]\right\}
\\
&=\frac{1}{\sqrt{N}}\sum_{i=1}^N \left\{\Ep [\bar\psi(W_{i0})\mid U_i] +\Ep [\bar\psi(W_{0i})\mid U_i]\right\},
\end{align*}
which is independent of $K$,
and it holds $P_N$-a.s. that
\begingroup
\allowdisplaybreaks
\begin{align*}
V(\Gnk \bar\psi)=& V(\Hnk \bar\psi) + O(N^{-1})=J_0^{-1}\Gamma (J_0^{-1})' + O(N^{-1})
\qquad\text{and}\\
\Gnk \bar\psi =& \Hnk \bar\psi + \Opn(N^{-1/2})
\end{align*}
\endgroup
under Assumption \ref{a:regularity_nuisance_parameters} (iv).
Recall that $q\ge 4$.
The third moments of both summands of $\Hnk \bar\psi$ are bounded under Assumptions \ref{a:linear_orthogonal_score} (v) and \ref{a:regularity_nuisance_parameters} (ii) (iv).
We have verified all the conditions for Lyapunov's CLT, and an application of Lyapunov's CLT and Cramer-Wold device thus gives
\begin{align*}
&\Hnk \bar\psi \leadsto N(0,I_{d_\theta}).
\end{align*}
An application of Theorem 2.7 of \cite{vdV98} concludes the proof.



\bigskip
\noindent \textbf{Step 2.}
For the fixed $K$, it suffices to show that, for any $k \in [K]$ and $(i,j)\in \overline{[N]^2}$,
\begin{align*}
\Big\|\Enk[\psi^a(W;\hat \eta_{k})]-\Epn[\psi^a(W_{12};\eta_0)]\Big\|
=
\Opn(N^{-1/2} + r_N).
\end{align*}
For any fixed $k \in [K]$,
\begin{align*}
&\Big\|\Enk[\psi^a(W;\hat \eta_{k})]-\Epn[\psi^a(W_{12};\eta_0)]\Big\|
\le \mathcal I_{1,k} + \mathcal I_{2,k},
\end{align*}
where
\begingroup
\allowdisplaybreaks
\begin{align*}
\mathcal I_{1,k} &= \Big\|\Enk[\psi^a(W;\hat \eta_{k})]-\Epn[\psi^a(W_{ij};\hat \eta_{k})|\Ikc]\Big\|,
\\
\mathcal I_{2,k} &= \Big\|\Epn[\psi^a(W_{ij};\hat \eta_{k})|\Ikc]-\Epn[\psi^a(W_{12};\eta_0)]\Big\|.
\end{align*}
\endgroup
Then $\mathcal I_{2,k}\le r_N$ with $P_N$-probability $1-o(1)$ follows directly from Assumptions \ref{a:sampling} (ii) and \ref{a:regularity_nuisance_parameters} (iii). 
Denote $\tilde\psi^a_{ij,m}=\psi^a_m(W_{ij};\hat \eta_{k})-\Epn[\psi^a_m(W_{ij};\hat \eta_{k})|\Ikc]$ and $\tilde \psi^a_{ij}=(\tilde \psi^a_{ij,m})_{m\in [d_\theta]}$. 
To bound $\mathcal I_{1,k}$, note that conditional on $\Ikc$, it holds that
\begingroup
\allowdisplaybreaks
\begin{align*}
\Epn[\mathcal I_{1,k}^2|\Ikc]
=&
\Epn\Big[\Big\|\Enk[\psi^a(W;\hat \eta_{k})]-\Epn[\psi^a(W_{ij};\hat \eta_{k})|\Ikc]\Big\|^2\Big|\Ikc\Big]\\
=&
\frac{1}{\big(|I_k|(|I_k|-1)\big)^2}\Epn\Big[\sum_{m=1}^{d_\theta}\Big(\sumIJ\tilde \psi^a_{ij,m}\Big)^2\Big|\Ikc\Big]\\
=&
\frac{1}{\big(|I_k|(|I_k|-1)\big)^2}\sumIJ\sum_{(i,j')\in \overline{I_k^2},j'\neq j}\Epn\Big[\sum_{m=1}^{d_\theta}\tilde \psi^a_{ij,m}\tilde \psi^a_{ij',m}\Big|\Ikc\Big]\\
&+
\frac{1}{\big(|I_k|(|I_k|-1)\big)^2}\sumIJ\sum_{(i',j)\in \overline{I_k^2},i'\neq i}\Epn\Big[\sum_{m=1}^{d_\theta}\tilde \psi^a_{ij,m}\tilde \psi^a_{i'j,m}\Big|\Ikc\Big]
\\
&+
\frac{1}{\big(|I_k|(|I_k|-1)\big)^2}\sumIJ\Epn\Big[\sum_{m=1}^{d_\theta}(\tilde \psi^a_{ij,m})^2\Big|\Ikc\Big]
+ 0
\\
= &
\frac{1}{\big(|I_k|(|I_k|-1)\big)^2} \sumIJ\sum_{(i,j')\in \overline{I_k^2},j'\neq j}\Epn[\langle\tilde\psi^a_{ij},\tilde\psi^a_{ij'}\rangle|\Ikc]\\
&+
\frac{1}{\big(|I_k|(|I_k|-1)\big)^2} \sumIJ\sum_{(i',j)\in \overline{I_k^2},i'\neq i}\Epn[\langle\tilde\psi^a_{ij},\tilde\psi^a_{i'j}\rangle|\Ikc]\\
&+
\frac{1}{\big(|I_k|(|I_k|-1)\big)^2} \sumIJ\Epn[\|\tilde\psi^a_{ij}\|^2|\Ikc]
\\
\lesssim&
\frac{1}{|I_k|-1} \Epn\Big[\Big\|\psi^a(W_{ij};\hat\eta_{k})-\Epn[\psi^a(W_{ij};\hat\eta_{k})|\Ikc]\Big\|^2\Big|\Ikc\Big] \\
\le&
\frac{1}{|I_k|-1} \Epn[\|\psi^a(W_{ij};\hat\eta_{k})\|^2|\Ikc] \\
\lesssim& \frac{c_1^2}{|I_k|-1}
\end{align*}
\endgroup
by an application of Cauchy-Schwartz's inequality and Assumptions \ref{a:sampling} and \ref{a:regularity_nuisance_parameters} (ii). 
Note that $N\lesssim |I_k|-1\lesssim N$. 
Hence an application of Lemma \ref{lemma:conditional_convergence} (i) implies 
$
\mathcal I_{1,k}=\Opn(N^{-1/2}).
$
This completes a proof of (\ref{eq:step2}).

\bigskip
\noindent \textbf{Step 3.}
It again suffices to show that, for any $k \in [K]$, we have
\begin{align*}
\Big\|\Enk[\psi(W;\theta_0,\hat\eta_{k})]
- \Enk[\psi(W;\theta_0,\eta_0)]\Big\|=\Opn(N^{-1/2}r_N'+ \lambda_N +  \lambda_N').
\end{align*}
Denote
\begin{align*}
\mathbb{G}_{I_k}[\phi(W)]=\frac{\sqrt{N}}{|I_k|(|I_k|-1)}\sumIJ \Big( \phi(W_{ij}) - \int \phi(w) dP_N \Big), 
\end{align*}
where $\phi$ is $P_N$-integrable function on $\supp(W)$.
Then,
\begin{align*}
&\Big\|\Enk[\psi(W;\theta_0,\hat\eta_{k})]
- \frac{1}{|I_k|(|I_k|-1)} \sumIJ\psi(W_{ij};\theta_0,\eta_0)\Big\|
\le
\frac{\mathcal I_{3,k}+\mathcal I_{4,k}}{\sqrt{N}},
\end{align*}
where
\begingroup
\allowdisplaybreaks
\begin{align*}
\mathcal I_{3,k}
:=&
\big\|\mathbb{G}_{I_k}[\psi(W;\theta_0,\hat \eta_{k})]-\mathbb{G}_{I_k}[\psi(W;\theta_0,\eta_0)]\big\|,
\\
\mathcal I_{4,k}
:=&\sqrt{N}
\Big\|
\Epn[\psi(W_{ij};\theta_0,\hat \eta_{k})|\overline{I_k^2}]- \Epn[\psi(W_{12};\theta_0,\eta_0)]
\Big\|.
\end{align*}
\endgroup
Denote $\tilde \psi_{ij,m}:=\psi_m(W_{ij};\theta_0,\hat \eta_{k})-\psi_m(W_{ij};\theta_0,\eta_0)$ and $\tilde \psi_{ij}=(\tilde \psi_{ij,m})_{m\in[d_\theta]}$. 
To bound $\mathcal I_{3,k}$, notice that using a similar argument to that for the bound of $\mathcal I_{1,k}$, one has
\begingroup
\allowdisplaybreaks
\begin{align*}
&\Epn\Big[\|\mathcal I_{3,k}\|^2\big|\Ikc \Big]
\\
=&\Epn\Big[\big\|\mathbb{G}_{I_k}\big[\psi(W;\theta_0,\hat \eta_{k})-\psi(W;\theta_0,\eta_0)\big]\big\|^2\big|\Ikc\Big]\\
=&\Epn\Big[\frac{N}{\big(|I_k|(|I_k|-1)\big)^2}\sum_{m=1}^{d_\theta}\Big\{\sumIJ \Big(\tilde\psi_{ij,m}
-\Epn\tilde\psi_{ij,m}
\Big)\Big\}^2\Big|\Ikc\Big]\\
=&
\frac{N}{\big(|I_k|(|I_k|-1)\big)^2}\sumIJ\sum_{(i,j')\in \overline{I_k^2}, j'\ne j} \Epn\Big[\sum_{m=1}^{d_\theta}\Big(\tilde\psi_{ij,m}-\Epn\tilde\psi_{ij,m}
\Big)\Big(\tilde\psi_{ij',m}-\Epn\tilde\psi_{ij',m}
\Big)\Big|\Ikc\Big]\\
&+
\frac{N}{\big(|I_k|(|I_k|-1)\big)^2}\sumIJ\sum_{(i',j)\in \overline{I_k^2}, i'\ne i} \Epn\Big[\sum_{m=1}^{d_\theta}\Big(\tilde\psi_{ij,m}-\Epn\tilde\psi_{ij,m}
\Big)\Big(\tilde\psi_{i'j,m}-\Epn\tilde\psi_{i'j,m}
\Big)\Big|\Ikc\Big]\\
&+
\frac{N}{\big(|I_k|(|I_k|-1)\big)^2}\sumIJ \Epn\Big[\sum_{m=1}^{d_\theta}\Big(\tilde\psi_{ij,m}-\Epn\tilde\psi_{ij,m}
\Big)^2\Big|\Ikc\Big]
+0\\
=&
\frac{N}{\big(|I_k|(|I_k|-1)\big)^2}\sumIJ\sum_{(i,j')\in \overline{I_k^2}, j'\ne j} \Epn\Big[\langle\tilde\psi_{ij}-\Epn\tilde\psi_{ij}
,\tilde\psi_{ij'}-\Epn\tilde\psi_{ij'}
\rangle\Big|\Ikc\Big]\\
&+
\frac{N}{\big(|I_k|(|I_k|-1)\big)^2}\sumIJ\sum_{(i',j)\in \overline{I_k^2}, i'\ne i} \Epn\Big[\langle\tilde\psi_{ij}-\Epn\tilde\psi_{ij}
,\tilde\psi_{i'j}-\Epn\tilde\psi_{i'j}
\rangle\Big|\Ikc\Big]\\
&+
\frac{N}{\big(|I_k|(|I_k|-1)\big)^2}\sumIJ \Epn\Big[\Big\|\tilde\psi_{ij}-\Epn\tilde\psi_{ij}
\Big\|^2\Big|\Ikc\Big]\\
\lesssim&
 \Epn\Big[\Big\|\psi(W_{ij};\theta_0,\hat\eta)-\psi(W_{ij};\theta_0,\eta_0)-\Epn[\psi(W_{ij};\theta_0,\hat\eta)-\psi(W_{ij};\theta_0,\eta_0)]\Big\|^2\Big|\Ikc\Big] \\
\leq&
 \Epn[\|\psi(W_{ij};\theta_0,\hat\eta)-\psi(W_{ij};\theta_0,\eta_0)\|^2|\Ikc] \\
 \le&
 \sup_{\eta \in \T_N}\Epn[\|\psi(W_{ij};\theta_0,\eta)-\psi(W_{ij};\theta_0,\eta_0)\|^2|\Ikc] \\
 =&
  \sup_{\eta \in \T_N}\Epn[\|\psi(W_{12};\theta_0,\eta)-\psi(W_{12};\theta_0,\eta_0)\|^2] = (r_N')^2,
\end{align*}
\endgroup
where the first inequality follows from Cauchy-Schwartz's inequality, the second-to-last equality is due to Assumption \ref{a:sampling}, and the last equality is due to Assumption \ref{a:regularity_nuisance_parameters} (iii).
Hence, $\mathcal I_{3,k}=\Opn(r_N')$.
To bound $\mathcal I_{4,k}$, let
\begin{align*}
f_{k}(r):=
\Epn\Big[\psi(W_{ij};\theta_0,\eta_0 + r(\hat \eta_{k}-\eta_0))\big|\Ikc\Big]
-
\Epn[\psi(W_{12};\theta_0,\eta_0)],
\qquad r\in[0,1].
\end{align*}
An application of the mean value expansion coordinate-wise gives
\begin{align*}
f_{k}(1)=f_{k}(0)+f_{k}'(0)+f_{k}''(\tilde r)/2,
\end{align*}
where $\tilde r \in (0,1)$.
Note that $f_{k}(0)=0$ under Assumption \ref{a:linear_orthogonal_score} (i), and
\begin{align*}
\|f_{k}'(0)\|=\Big\| \partial_\eta \Epn \psi(W_{12};\theta_0,\eta_0)[\hat \eta_{k}-\eta_0] \Big\|\le \lambda_N
\end{align*}
under Assumption \ref{a:linear_orthogonal_score} (iv).
Moreover, under Assumption \ref{a:regularity_nuisance_parameters} (iii), on the event $\E_N$, we have
\begin{align*}
\|f_{k}''(\tilde r)\|\le \sup_{r\in(0,1)}\|f_{k}''(r)\|\le \lambda_N'.
\end{align*}
This completes a proof of (\ref{eq:step3}).

\bigskip
\noindent \textbf{Step 4.}
Note that
\begingroup
\allowdisplaybreaks
\begin{align*}
&
\Epn\Big[\Big\|\frac{\sqrt{N}}{K}\sumk\Enk[\psi(W;\theta_0,\eta_0)]\Big\|^2\Big]
\\
=&
\frac{N}{\big(\sumk |I_k|(|I_k|-1)\big)^2}\Epn\Big[\sum_{m=1}^{d_\theta}\Big(\sumk\sum_{(i,j)\in \overline{I_k^2}}\psi_m(W_{ij};\theta_0,\eta_0)\Big)^2\Big]\\
=&
\frac{N}{\big(\sumk |I_k|^2-N\big)^2}\sumk\sumIJ\sum_{(i,j')\in \overline{I_k^2}, j'\neq j}\Epn\Big[\sum_{m=1}^{d_\theta}\psi_m(W_{ij};\theta_0,\eta_0)\psi_m(W_{ij'};\theta_0,\eta_0)\Big]\\
&+
\frac{N}{\big(\sumk |I_k|^2-N\big)^2}\sumk\sumIJ\sum_{(i',j)\in \overline{I_k^2}, i'\neq i}\Epn\Big[\sum_{m=1}^{d_\theta}\psi_m(W_{ij};\theta_0,\eta_0)\psi_m(W_{i'j};\theta_0,\eta_0)\Big]\\
&+
\frac{N}{\big(\sumk |I_k|^2-N\big)^2}\sumk\sumIJ\Epn\Big[\sum_{m=1}^{d_\theta}\psi^2_m(W_{ij};\theta_0,\eta_0)\Big] + 0\\
\lesssim& 
\Epn[\|\psi(W_{12};\theta_0,\eta_0)\|^2]\le c_1^2,
\end{align*}
\endgroup
where the first equality is due to $|I_k|=|I|$ for all $k \in [K]$.
Therefore, an application of Markov's inequality implies
\begin{align*}
\Big\|\frac{\sqrt{N}}{K}\sumk \Enk[\psi(W_{ij};\theta_0,\eta_0)]\Big\|=\Opn(1).
\end{align*}
This completes a proof of (\ref{eq:step4}).

\bigskip
\noindent \textbf{Step 5.}
Note that all the singular values of $J_0$ are bounded from above by $c_1$ under Assumption \ref{a:linear_orthogonal_score} (v) and all the eigenvalues of $\Gamma$ are bounded from below by $c_0$ under Assumption \ref{a:regularity_nuisance_parameters} (iv). 
Therefore, we have $\|\sigma^{-1}\|\le c_1/\sqrt{ c_0}$ and thus
$
\|\sigma^{-1}\|
=\Opn(1).
$
This completes a proof of (\ref{eq:step5}).
\end{proof}

\subsection{Proof of Lemma \ref{theorem:variance_estimator_linear}}\label{sec:theorem:variance_estimator_linear}
\begin{proof}
Step 2 of the proof of Lemma \ref{theorem:main_result_linear} proves $\|\hat J - J_0\| = O_p(N^{-1/2} + r_N)$, and Assumption \ref{a:linear_orthogonal_score} (v) implies $\|J_0^{-1}\| \leq c_0^{-1}$.
Therefore, to prove the claim of the lemma, it suffices to show
\begingroup
\allowdisplaybreaks
\begin{align*}
 \ \Big\|\frac{1}{K}\sumk
 &\Big\{
\frac{|I_k|-1}{(|I_k|(|I_k|-1))^2}
\Big[
\sum_{i \in I_k}\sum_{\substack{{j,j' \in I_k}\\{j,j'\neq i}}}\psi(W_{ij};\tilde \theta,\hat\eta_{k}) \psi(W_{ij'};\tilde \theta,\hat\eta_{k})' 
+\sum_{j \in I_k}\sum_{\substack{{i,i' \in I_k}\\{i,i'\neq j}}}\psi(W_{ij};\tilde \theta,\hat\eta_{k}) \psi(W_{i'j};\tilde \theta,\hat\eta_{k})' \\
&+\sum_{i \in I_k}\sum_{\substack{{j,j' \in I_k}\\{j,j'\neq i}}}\psi(W_{ij};\tilde \theta,\hat\eta_{k}) \psi(W_{j'i};\tilde \theta,\hat\eta_{k})' 
+\sum_{j \in I_k}\sum_{\substack{{i,i' \in I_k}\\{i,i'\neq j}}}\psi(W_{ij};\tilde \theta,\hat\eta_{k}) \psi(W_{ji'};\tilde \theta,\hat\eta_{k})' 
 \Big]\Big\}
\\
&-\{
\Ep[\psi(W_{12};\theta_0,\eta_0)\psi(W_{13};\theta_0,\eta_0)]+\Ep[\psi(W_{21};\theta_0,\eta_0)\psi(W_{31};\theta_0,\eta_0)]\\
&+\Ep[\psi(W_{12};\theta_0,\eta_0)\psi(W_{31};\theta_0,\eta_0)]
+\Ep[\psi(W_{21};\theta_0,\eta_0)\psi(W_{13};\theta_0,\eta_0)]\}
\Big\|=\Op(\rho_N).
\end{align*}
\endgroup
Moreover, since $K$ and $d_{\theta}$ are constants, it suffices to show that, each of the following four objects is $\Op(\rho_N)$ for each $k\in [K]$ and $l,m\in[d_\theta]$:
\begingroup
\allowdisplaybreaks
\begin{align*}
&\Big|
\frac{|I_k|-1}{(|I_k|(|I_k|-1))^2}\sum_{i \in I_k}\sum_{\substack{{j,j' \in I_k}\\{j,j'\neq i}}} \psi_l(W_{ij};\tilde \theta,\hat\eta_{k}) \psi_m(W_{ij'};\tilde \theta,\hat\eta_{k})
- 
\Ep[\psi_l(W_{12};\theta_0,\eta_0)\psi_m(W_{13};\theta_0,\eta_0)]\Big|,
\\
&\Big|
\frac{|I_k|-1}{(|I_k|(|I_k|-1))^2}\sum_{j \in I_k}\sum_{\substack{{i,i' \in I_k}\\{i,i'\neq j}}} \psi_l(W_{ij};\tilde \theta,\hat\eta_{k}) \psi_m(W_{i'j};\tilde \theta,\hat\eta_{k})
- 
\Ep[\psi_l(W_{21};\theta_0,\eta_0)\psi_m(W_{31};\theta_0,\eta_0)]\Big|,
\\
&\Big|
\frac{|I_k|-1}{(|I_k|(|I_k|-1))^2}\sum_{i \in I_k}\sum_{\substack{{j,j' \in I_k}\\{j,j'\neq i}}} \psi_l(W_{ij};\tilde \theta,\hat\eta_{k}) \psi_m(W_{j'i};\tilde \theta,\hat\eta_{k})
- 
\Ep[\psi_l(W_{12};\theta_0,\eta_0)\psi_m(W_{31};\theta_0,\eta_0)]\Big|,
\\
&\Big|
\frac{|I_k|-1}{(|I_k|(|I_k|-1))^2}\sum_{\substack{{j \in I_k}\\{i,i' \in I_k}}}\sum_{\substack{{i\neq i'}\\{i,i'\neq j}}} \psi_l(W_{ij};\tilde \theta,\hat\eta_{k}) \psi_m(W_{ji'};\tilde \theta,\hat\eta_{k})
- 
\Ep[\psi_l(W_{21};\theta_0,\eta_0)\psi_m(W_{13};\theta_0,\eta_0)]\Big|
.
\end{align*}
\endgroup

We will show the first statement, and the three others will follow analogously. 
Let the left-hand side of the first equation be denoted by $\mathcal I_{k,lm}$. 
Apply the triangle inequality to get
\begingroup
\allowdisplaybreaks
\begin{align*}
\mathcal I_{k,lm}\le \mathcal I_{k,lm,1}+ \mathcal I_{k,lm,2},
\end{align*}
where
\begin{align*}
\mathcal I_{k,lm,1}&:=
\Big|\frac{1}{|I_k|^2(|I_k|-1)}\sum_{i \in I_k}\sum_{\substack{{j,j' \in I_k}\\{j,j'\neq i}}} \Big\{\psi_l(W_{ij};\tilde \theta,\hat\eta_{k}) \psi_m(W_{ij'};\tilde \theta,\hat\eta_{k})
-
 \psi_l(W_{ij};\theta_0,\eta_0) \psi_m(W_{ij'};\theta_0,\eta_0) \Big\}
\Big|\\
\mathcal I_{k,lm,2}:&=
\Big|\frac{1}{|I_k|^2(|I_k|-1)}\sum_{i \in I_k}\sum_{\substack{{j,j' \in I_k}\\{j,j'\neq i}}} \psi_l(W_{ij};\theta_0,\eta_0) \psi_m(W_{ij'};\theta_0,\eta_0)
-
\Ep[\psi_l(W_{12};\theta_0,\eta_0)\psi_m(W_{13};\theta_0,\eta_0)]
\Big|.
\end{align*}
We first find a bound for $\mathcal I_{k,lm,2}$. Since $q>4$, it holds that
\begin{align*}
&\Ep\big[\mathcal I_{k,lm,2}^2\big]
\\
=&
\frac{1}{|I_k|^4(|I_k|-1)^2}\Ep\Big[
\Big|\sum_{i \in I_k}\sum_{\substack{{j,j' \in I_k}\\{j,j'\neq i}}} \psi_l(W_{ij};\theta_0,\eta_0) \psi_m(W_{ij'};\theta_0,\eta_0)
-
\Ep[\psi(W_{12};\theta_0,\eta_0)\psi(W_{13};\theta_0,\eta_0)]
\Big|^2
\Big]\\
\le &
\frac{1}{|I_k|^4(|I_k|-1)^2}\Ep\Big[
\sum_{i,i'\in I_k }\sum_{\substack{j,\jmath \in I_k\setminus\{i\} \\ j,\jmath' \in I_k\setminus \{i'\}}} \psi_l(W_{ij};\theta_0,\eta_0) \psi_m(W_{i\jmath};\theta_0,\eta_0)
\psi_l(W_{i'j};\theta_0,\eta_0) \psi_m(W_{i'\jmath'};\theta_0,\eta_0)
\Big]\\
&+\frac{1}{|I_k|^4(|I_k|-1)^2}\Ep\Big[
\sum_{i\in I_k}\sum_{\substack{{j,j',\jmath,\jmath'\in I_k}\\{ j,j',\jmath,\jmath' \neq i}}} \psi_l(W_{ij};\theta_0,\eta_0) \psi_m(W_{ij'};\theta_0,\eta_0)\psi_l(W_{i\jmath};\theta_0,\eta_0) \psi_m(W_{i\jmath'};\theta_0,\eta_0)
\Big]\\
&+o((|I_k|-1)^{-1}) + 0\\
\lesssim&\frac{1}{|I_k|-1}\Ep[\|\psi(W;\theta_0,\eta_0)\|^4]\lesssim c_1^4/N=O(N^{-1}).
\end{align*}
\endgroup

Now, to bound $\mathcal I_{k,lm,1}$, we make use of the following identity borrowed from the proof of Theorem 3.2 in \citet{CCDDHNR18}: for any numbers $a$, $b$, $\delta a$, $\delta b$ such that
$|a|\vee |b|\le c$ and $|\delta a |\vee |\delta b| \le r$, it holds that 
$
|(a+\delta a)(b+\delta b)-ab|\le 2r( c+r).
$ 
Denote $\psi_{ij,h}:=\psi_h(W_{ij};\theta_0,\eta_0)$ and $\hat\psi_{ij,h}:=\psi_h(W_{ij};\tilde \theta,\hat\eta_{k})$ for $h\in\{l,m\}$, and apply the above identity with $a=\psi_{ij,l}$, $b=\psi_{ij',m}$, $a+\delta a =\hat\psi_{ij,l}$, $b+ \delta b=\hat\psi_{ij',m}$, $r=|\hat\psi_{ij,l}-\psi_{ij,l}|\vee |\hat\psi_{ij',m}-\psi_{ij',m}|$ and $c=|\psi_{ij,l}|\vee|\psi_{ij',m}|$.
Then,
\begingroup
\allowdisplaybreaks
\begin{align*}
\mathcal I_{k,lm,1}=&
\Big|\frac{1}{|I_k|^2(|I_k|-1)}\sum_{i \in I_k}\sum_{\substack{{j,j'\in I_k}\\{j,j'\neq i}}} \Big\{
\hat\psi_{ij,l} \hat\psi_{ij',m}
-
\psi_{ij,l} \psi_{ij',m}  \Big\}
\Big|\\
\le &
\frac{1}{|I_k|^2(|I_k|-1)}\sum_{i \in I_k}\sum_{\substack{{j,j'\in I_k}\\{j,j'\neq i}}}|
\hat\psi_{ij,l} \hat\psi_{ij',m}
-
\psi_{ij,l} \psi_{ij',m} |\\
\le& 
\frac{2}{|I_k|^2(|I_k|-1)}\sum_{i \in I_k}\sum_{\substack{{j,j'\in I_k}\\{j,j'\neq i}}} (|\hat\psi_{ij,l}-\psi_{ij,l}|\vee |\hat\psi_{ij',m}-\psi_{ij',m}|)
\\
&\qquad \times \Big(|\psi_{ij,l}|\vee|\psi_{ij',m}|
+|\hat\psi_{ij,l}-\psi_{ij,l}|\vee |\hat\psi_{ij',m}-\psi_{ij',m}|\Big)\\
\le& 
\Big(\frac{2}{|I_k|^2(|I_k|-1)}\sum_{i \in I_k}\sum_{\substack{{j,j'\in I_k}\\{j,j'\neq i}}} |\hat\psi_{ij,l}-\psi_{ij,l}|^2\vee |\hat\psi_{ij',m}-\psi_{ij',m}|^2\Big)^{1/2}\\
&\qquad\times
 \Big(\frac{2}{|I_k|^2(|I_k|-1)}\sum_{i \in I_k}\sum_{\substack{{j,j'\in I_k}\\{j,j'\neq i}}}\Big\{|\psi_{ij,l}|\vee|\psi_{ij',m}|
+|\hat\psi_{ij,l}-\psi_{ij,l}|\vee |\hat\psi_{ij',m}-\psi_{ij',m}|\Big\}^2\Big)^{1/2}\\
\le&
\Big(\frac{2}{|I_k|^2(|I_k|-1)}\sum_{i \in I_k}\sum_{\substack{{j,j'\in I_k}\\{j,j'\neq i}}} |\hat\psi_{ij,l}-\psi_{ij,l}|^2\vee |\hat\psi_{ij',m}-\psi_{ij',m}|^2\Big)^{1/2}\\
&\times\Big\{
 \Big(\frac{2}{|I_k|^2(|I_k|-1)}\sum_{i \in I_k}\sum_{\substack{{j,j'\in I_k}\\{j,j'\neq i}}}|\psi_{ij,l}|^2\vee|\psi_{ij',m}|^2\Big)^{1/2} 
\\
&
+
  \Big(\frac{2}{|I_k|^2(|I_k|-1)}\sum_{i \in I_k}\sum_{\substack{{j,j'\in I_k}\\{j,j'\neq i}}} |\hat\psi_{ij,l}-\psi_{ij,l}|^2\vee |\hat\psi_{ij',m}-\psi_{ij',m}|^2\Big)^{1/2}\Big\},
\end{align*}
\endgroup
where the second to the last inequality follows the Cauchy-Schwartz's inequality and Minkowski's inequality.
Notice that
\begingroup
\allowdisplaybreaks
\begin{align*}
&\sum_{i \in I_k}\sum_{\substack{{j,j'\in I_k}\\{j,j'\neq i}}}|\psi_{ij,l}|^2\vee|\psi_{ij',m}|^2\le   |I_k|\sumIJ \|\psi(W_{ij};\theta_0,\eta_0)\|^2,\\
&\sum_{i \in I_k}\sum_{\substack{{j,j'\in I_k}\\{j,j'\neq i}}}|\hat \psi_{ij,l}-\psi_{ij,l}|^2\vee|\hat \psi_{ij',m}-\psi_{ij',m}|^2\le |I_k|\sumIJ \|\psi(W_{ij};\tilde\theta,\hat \eta_{k})-\psi(W_{ij};\theta_0,\eta_0)\|^2.
\end{align*}
\endgroup
Thus, the above bound for $\mathcal I_{k,lm,1}$ implies that
\begin{align*}
\mathcal I_{k,lm,1}^2
\lesssim&
R_N\times \Big(\frac{1}{|I_k|(|I_k|-1)}\sumIJ
\| \psi(W_{ij};\theta_0, \eta_0) 
\|^2
+
R_N
\Big),
\end{align*}
where
\begin{align*}
R_N:=\frac{1}{|I_k|(|I_k|-1)}\sumIJ
\| \psi(W_{ij};\tilde\theta,\hat \eta_{k}) - \psi(W_{ij};\theta_0,\eta_0) 
\|^2.
\end{align*}
Note that
\begin{align*}
\frac{1}{|I_k|(|I_k|-1)}\sumIJ
\| \psi(W_{ij};\theta_0,\eta_0) 
\|^2=\Op(1),
\end{align*}
which is implied by Markov's inequality and the calculations
\begin{align*}
\Ep\Big[\frac{1}{|I_k|(|I_k|-1)}\sumIJ
\| \psi(W_{ij};\theta_0,\eta_0) 
\|^2\Big]=&\Ep[ \|\psi(W_{12};\theta_0,\eta_0) 
\|^2]\le c_1^2
\end{align*}
under Assumptions \ref{a:sampling} and \ref{a:regularity_nuisance_parameters} (ii). 
Finally, to bound $R_N$, use Assumption \ref{a:linear_orthogonal_score} (ii) to get
\begin{align*}
R_N\lesssim&
 \frac{1}{|I_k|(|I_k|-1)}\sumIJ
\| \psi^a(W_{ij};\hat \eta_{k})(\tilde \theta -\theta_0)
\|^2 +
\frac{1}{|I_k|(|I_k|-1)}\sumIJ
\| \psi(W_{ij};\theta_0,\hat\eta_{k}) -\psi(W_{ij};\theta_0,\eta_0) 
\|^2.
\end{align*}
The first term on the right-hand side is bounded by
\begin{align*}
 \Big(\frac{1}{|I_k|(|I_k|-1)}\sumIJ
\| \psi^a(W_{ij};\hat \eta_{k})
\|^2 \Big)\times\|\tilde \theta -\theta_0\|^2=\Op(1)\times \Op(N^{-1})=\Op(N^{-1})
\end{align*}
due to Assumption \ref{a:regularity_nuisance_parameters} (ii), Markov's inequality, and Lemma \ref{theorem:main_result_linear}.
Furthermore, given that $(W_{ij})_{(i,j)\in \Ikc}$ satisfies $\hat \eta_{k}\in\mathcal T_N$,
\begin{align*}
\Ep\Big[ \|\psi(W_{ij};\theta_0,\hat\eta_{k}) -\psi(W_{ij};\theta_0,\eta_0)\|^2\Big|\Ikc\Big]\le&
\sup_{\eta\in \mathcal T_N}\Ep\Big[ \|\psi(W_{ij};\theta_0,\eta) -\psi(W_{ij};\theta_0,\eta_0)\|^2\Big|\Ikc\Big] \le (r_N')^2
\end{align*}
due to Assumptions \ref{a:sampling} and \ref{a:regularity_nuisance_parameters} (iii).
Also, the event $\hat \eta_{k}\in\mathcal T_N$ happens with probability $1-o(1)$, and we have $R_N=\Op(N^{-1}+(r'_N)^2)$. 
Thus, we conclude that 
\begin{align*}
\mathcal I_{k,lm,1}=\Op(N^{-1/2}+r'_N).
\end{align*}
This completes the proof of Lemma \ref{theorem:variance_estimator_linear}.
\end{proof}

\subsection{Proof of Theorem \ref{theorem_DDML_non_linear}}\label{sec:theorem:main_result_linear}
\begin{proof}
With the aid of Lemma \ref{lemma_linearization_subsample_DML} stated and proved below, which shows the approximate linearity of the subsample dyadic machine learning estimators $\tilde{\theta}$, the proof will be same as that in the linear case presented in the proof of Lemma \ref{theorem:main_result_linear}.
\end{proof}
\begin{lemma}[Linearization for Subsample DML for Nonlinear Scores]
\label{lemma_linearization_subsample_DML}
Suppose that the conditions of Theorem \ref{theorem_DDML_non_linear} hold. 
For any $k=1,...,K$, the estimator $\tilde{\theta}$ follows
\begin{align}
\label{lemma_subsample_DML}
\sqrt{N}\sigma^{-1}(\tilde{\theta}-\theta_0)=\sqrt{N}\Enk \bar{\psi}(W)+\Op(\delta_N)
\end{align}
uniformly over $P\in\mathcal P_N$, where the influence function takes the form $\bar \psi(\cdot):=-\sigma^{-1}J_0^{-1} \psi(\cdot;\theta_0,\eta_0)$.
\end{lemma}

\begin{proof}
Fix any $k=1,...,K$, and any sequence $\{P_N\}_{N\geq 1}$ such that $P_N\in \mathcal P_N$ for all $N\geq 1$.  
We shall follow roughly the proof structure of Theorem 6.3 in  \citet{CCDDHNR18}. Note that, however, the modification of asymptotic theory to accommodate the current dyadic setting is far from trivial due to the new cross-fitting strategy and the need of controlling empirical processes with dyadic clustered data. 

To prove the asserted claim, it suffices to show that $\tilde{\theta}$ satisfies (\ref{lemma_subsample_DML}) with $P$ replaced by $P_N$. 

\bigskip
\noindent \textbf{Step 1.} (Preliminary result)
We claim that, with $P_N$-probability $1-o(1)$,
\begin{align}
\label{inequality_theta_bound}
\big\|\tilde{\theta}-\theta_{0}\big\| \lesssim B_{1N}\tau_{N}.
\end{align}
To show this claim, note that the definition of $\tilde{\theta}$ implies that
\begin{align*}
\left\| \frac{1}{K} \sumk  \Enk\left[\psi (W ; \tilde{\theta}, \hat{\eta}_{k})\right]\right\|
\leq
\inf_{\theta\in\Theta}
\left\| \frac{1}{K} \sumk  \Enk\left[\psi\left(W ; \theta, \hat{\eta}_{ k}\right)\right] \right\|+\epsilon_{N},
\end{align*}
which in turn implies via the triangle inequality that, with $P_N$-probability $1-o(1)$,
\begin{align}
\label{triangle_I}
\left\|\left.\Epn\left[\psi\left(W_{12} ; \theta, \eta_{0}\right)\right]\right|_{\theta=\tilde{\theta}}\right\| \leq \epsilon_{N}+2 \mathcal{I}_{1}+2 \mathcal{I}_{2},
\end{align}
where
\begingroup
\allowdisplaybreaks
\begin{align*}
&\mathcal{I}_{1}:=\sup _{\theta \in \Theta}\left\|\Enk[\psi(W ; \theta, \hat\eta)]-\Enk\left[\psi\left(W; \theta, \eta_{0}\right)\right]\right\|\lesssim B_{1N}\tau_N ,
\\
& \mathcal{I}_{2}:=\max _{\eta \in\left\{\eta_{0}, \hat{\eta}_{k}\right\}} \sup _{\theta \in \Theta}\left\| \frac{1}{K}\sumk \Enk[\psi(W ; \theta, \eta)]-\Epn[\psi(W_{12} ; \theta, \eta)]\right\|\lesssim \tau_N,
\end{align*}
and the bounds of $\mathcal{I}_1$  and $\mathcal{I}_2$ will be derived in step 2. Recall that $\epsilon_N=o(\delta_NN^{-1/2})$.
Since $\delta_NN^{-1/2}\lesssim \tau_N$, we have $\epsilon_N=o(\tau_N)$.
\endgroup

Hence, it follows from (\ref{triangle_I}) and Assumption \ref{a:nonlinear_moment_condition} that, with $P_N$-probability $1-o(1)$,
\begin{align}
\left\|J_{0}\left(\tilde{\theta}-\theta_{0}\right)\right\| \wedge c_{0} \leq\left\|\left.\Epn\left[\psi\left(W_{12} ; \theta, \eta_{0}\right)\right]\right|_{\theta=\tilde{\theta}}\right\|\lesssim B_{1N}\tau_N.
\end{align}
Combining this bound with the fact that the singular values of $J_0$ are bounded away from zero, which holds by Assumption \ref{a:nonlinear_moment_condition}, proves the claim of this step.

\bigskip
\noindent \textbf{Step 2.} 
In this step, we shall establish that
$\mathcal{I}_1\lesssim B_{1N}\tau_N \text{ and  }  \mathcal{I}_2\lesssim\tau_N$ with probability $1-o(1)$.
Notice that with probability $1-o(1)$, we have $\mathcal{I}_1\leq 2\mathcal{I}_{1,1}+\mathcal{I}_{1,2}$ and $\mathcal{I}_2\leq \mathcal{I}_{1,1}$, where 
\begin{align*}
&\mathcal{I}_{1,1}=\sup_{\theta\in\Theta,\eta\in\mathcal{T}_N}\Big\|\Enk[\psi(W;\theta,\eta)]-\Epn[\psi(W_{12};\theta,\eta)]\Big\|,
\\
&\mathcal{I}_{1,2}=\sup_{\theta\in\Theta,\eta\in\mathcal{T}_N}\Big\|\Epn[\psi(W_{12};\theta,\eta)]-\Epn[\psi(W_{12};\theta,\eta_0)]\Big\|.
\end{align*}
By Taylor expansion,
\begin{align*}
\mathcal{I}_{1,2}\leq \sup_{\theta\in\Theta,\eta\in\mathcal{T}_N,r\in[0,1)}\big\|\partial_r\Epn\big[\psi\big(W_{12};\theta,\eta_0+r(\eta-\eta_0)\big)\big]\big\|
\leq B_{1N}\sup_{\eta\in\mathcal{T}_N}\|\eta-\eta_0\|\leq B_{1N}\tau_N,
\end{align*}
where the second and last inequalities follow from Assumption \ref{a:nonlinear_score_regularity_nuisance_parameters} (iv) and (v).
To bound $\mathcal{I}_{1,1}$, we apply Lemma \ref{lemma_maximal_inequality_dyadic_data} to the class $\mathcal{F}_1$ defined in Assumption \ref{a:nonlinear_score_regularity_nuisance_parameters} (i) with $n=N$, $v=v_N$, $A=a_N$, $b_n=D_N$ and $\bar{\sigma}_n=C_0$,  and the function classes defined in Lemma \ref{lemma_maximal_inequality_dyadic_data} are given by
$\mathcal{G}=\{\Ep[f(W_{12})|U_{k}=\cdot]:k=1,2,f\in\mathcal{F}_1\}$ and $\mathcal{H}=\{\Ep[f(W_{12})|U_1=\cdot,U_2=\cdot]:f\in\mathcal{F}_1\}$ with envelopes $G(u)=\Ep[F_1(W_{12})|U_1=u]\vee \Ep[F_1(W_{12})|U_2=u]$ and $H(u_1,u_2)=\Ep[F_1(W_{12})|U_1=u_1,U_2=u_2]$, respectively. Observe that $a_N\geq N\vee D_N$, $v_N\geq 1$. Furthermore, by Jensen's inequality, $\|G\|_{P,q}\vee \|H\|_{P,q}\le\|F_1\|_{P,q}\leq D_N$ under Assumption \ref{a:nonlinear_score_regularity_nuisance_parameters} (i), and $\sup_{g\in\mathcal{G}}\|g\|_{P,2}\vee \sup_{h\in\mathcal{H}}\|h\|_{P,2}\leq C_0$ following Assumption \ref{a:nonlinear_score_regularity_nuisance_parameters} (ii).
Then, an application of Lemma \ref{lemma_maximal_inequality_dyadic_data} yields that with probability $1-o(1)$,
\begin{align}
\mathcal{I}_{1,1}\lesssim N^{-1/2}(\sqrt{v_N\log a_N}+N^{-1/2+1/q}v_ND_N\log a_N).
\end{align}
By Assumption \ref{a:nonlinear_score_regularity_nuisance_parameters} (vii), we have $(v_N\log a_N)^{1/2}\lesssim N^{1/2}\tau_N$ and $N^{-1/2}N^{-1/2+1/q}v_ND_N\log a_N\lesssim N^{-1/2}\delta_N\lesssim N^{-1/2}\lesssim\tau_N$. Hence, $\mathcal{I}_{1,1}\lesssim \tau_N$, which completes the proof of this step.

\bigskip
\noindent \textbf{Step 3.} (Linearization)
Here in this step, we prove the claim of the lemma. 
First, by the definition of $\tilde{\theta}$, we have 
\begin{align}
\sqrt{N}\left\| \Enk[\psi (W ; \tilde{\theta}, \hat{\eta}_{k})]\right\| \leq
 \inf _{\theta \in \Theta} \sqrt{N}\left\| \Enk[\psi (W ; \theta, \hat{\eta}_{k} )]\right\|+\epsilon_{N} \sqrt{N}.
\end{align}
Then, for any $\theta\in \Theta$ and $\eta \in \mathcal T_N$, we have
\begingroup
\allowdisplaybreaks
\begin{align}
\label{empirical_expansion}
\nonumber
 \sqrt{N} \Enk[\psi(W ; \theta, \eta)]
 =
 & \sqrt{N} \Enk\left[\psi\left(W ; \theta_{0}, \eta_{0}\right)\right]+\mathbb{G}_{I_k}\left[\psi(W ; \theta, \eta)-\psi\left(W ; \theta_{0}, \eta_{0}\right)\right] 
 \\
  &+\sqrt{N}\Epn[\psi(W_{12} ; \theta, \eta)],
\end{align}
\endgroup
where $\mathbb G_{I_k}[\psi(W;\theta,\eta)]:=\frac{\sqrt{N}}{|I_k|(|I_k|-1)}\sumIJ\big(\psi(W_{ij})-\int\psi(w)dP_N(w)\big)$
and we are using the fact that $\Epn[\psi(W_{12};\theta_0,\eta_0)]=0$. 
\\Hence, by a Taylor expansion of the function 
$
r \mapsto \Epn\left[\psi\big(W_{12} ; \theta_{0}+r\left(\theta-\theta_{0}\right), 
\eta_{0}+r\left(\eta-\eta_{0}\right)\big)\right]
$,
\begingroup
\allowdisplaybreaks
It holds that
\begin{align}
\label{taylor_expansion}
\nonumber
\Epn[\psi(W_{12} ; \theta, \eta)]=& J_{0}\left(\theta-\theta_{0}\right)+D_0\left[\eta-\eta_{0}\right] 
\\
 &+2^{-1} \partial_{r}^{2} \Epn\left[\psi(W_{12} ; \theta_{0}+r\left(\theta-\theta_{0}\right), \eta_{0}+r\left(\eta-\eta_{0}\right))\right]\big|_{r=\bar{r}},
\end{align}
\endgroup
where $\bar{r}\in (0,1)$.
Applying (\ref{empirical_expansion}) with $\theta=\tilde{\theta}$ and $\eta=\hat{\eta}_k$, we have, with $P_N$-probability $1-o(1)$,
\begin{align}
\sqrt{N}\left\|\Enk\left[\psi\left(W ; \theta_{0}, \eta_{0}\right)\right]+J_{0}\left(\tilde{\theta}-\theta_{0}\right)\right\| \leq \lambda_{N} \sqrt{N}+\epsilon_{N} \sqrt{N}+\mathcal{I}_{3}+\mathcal{I}_{4}+\mathcal{I}_{5},
\end{align}
where 
\begin{align}
& \mathcal{I}_{3} :=\inf _{\theta \in \Theta} \sqrt{N}\left\| \Enk[\psi (W ; \theta, \hat{\eta}_{k} )]\right\|
,
\label{I_3}\\
&\mathcal{I}_{4}:=\sqrt{N} \sup _{ r\in[0,1),\theta\in\Theta,\eta \in \mathcal{T}_{N}}\left\|  \partial_{r}^{2} \Epn\left[\psi(W_{12} ; \theta_{0}+r\left(\theta-\theta_{0}\right), \eta_{0}+r\left(\eta-\eta_{0}\right))\right] \right\| ,
\label{I_4}\\
& \mathcal{I}_{5}: =\sup _{\theta\in\Theta}\left\|\mathbb{G}_{I_k}\big[\psi(W ; \theta, \hat{\eta}_{k})-\psi(W ; \theta_{0}, \eta_{0})\big]\right\|.\label{I_5}
\end{align}
Step 5 shows that $\mathcal{I}_3=O_{P_N}(\delta_N)$, and Step 4 shows that $\mathcal{I}_4=O_{P_N}(\delta_N)$ and $\mathcal{I}_5=O_{P_N}(\delta_N)$.
Because all singular values of $J_0$ are bounded below from zero by Assumption \ref{a:nonlinear_moment_condition} (iii), we have
\begingroup
\allowdisplaybreaks
\begin{align*}
\left\|J_{0}^{-1} \sqrt{N} \Enk\left[\psi\left(W ; \theta_{0}, \eta_{0}\right)\right]+\sqrt{N}(\tilde{\theta}-\theta_{0})\right\| =O_{P_{N}}\left(\delta_N\right) .
\end{align*}
\endgroup
The asserted claim now follows by multiplying both parts of the display by $\sigma^{-1}$ (inside the norm on the left-hand side) and noting that singular values of $\sigma^2$ are bounded below from zero by Assumption \ref{a:regularity_nuisance_parameters}.

\bigskip
\noindent \textbf{Step 4.}
Here in this step, we derive  bounds on $\mathcal{I}_4$ and $\mathcal I_5$. 
First, observe that with probability $1-o(1)$,
\begin{align*}
\mathcal{I}_4\leq \sqrt{N}B_{2,N}\sup_{\theta\in\Theta,\eta\in\mathcal{T}_N}\Big(\|\theta-\theta_{0}\|^2\vee\|\eta-\eta_0\|^2\Big)\lesssim \sqrt{N}B_{1N}^2B_{2N}\tau_N^2\lesssim \delta_N,
\end{align*}
where the first inequality follows from Assumption \ref{a:nonlinear_score_regularity_nuisance_parameters} (iv) and the second from Assumption \ref{a:nonlinear_score_regularity_nuisance_parameters} (v) and Equation (\ref{inequality_theta_bound}), and the last from Assumption \ref{a:nonlinear_score_regularity_nuisance_parameters} (vii) c. Next, we turn to bounding $\mathcal{I}_5$. 
Without loss of generality, let us fix $k$ such that $1,2\notin I_k$ following identical distribution implied by Assumption \ref{a:sampling} (i).
We have
\begin{align*}
\mathcal{I}_{5} \lesssim \sup _{f \in \mathcal{F}_{2}}\left|\mathbb{G}_{I_k}(f)\right|, \quad \mathcal{F}_{2}=\left\{\psi_{j}\left(\cdot; \theta, \hat{\eta}_{k}\right)-\psi_{j}\left(\cdot; \theta_{0}, \eta_{0}\right): j\in[d_\theta],\left\|\theta-\theta_{0}\right\| \lesssim CB_{1N}\tau_{N}\right\}.
\end{align*}
To bound $\sup _{f \in \mathcal{F}_{2}}\left|\mathbb{G}_{I_k}(f)\right|$, we apply Lemma \ref{lemma_maximal_inequality_dyadic_data} conditional on $(W_{ij})_{(i,j) \in \Ikc}$  so that $\hat{\eta}_k$ can be treated as fixed. 
Observe that, with $P_N$-probability $1-o(1)$, 
\begin{align*}
\sup_{f\in\mathcal{F}_2}\|f\|_{P_N,2}
&\leq \sup_{\|\theta-\theta_{0}\| \leq CB_{1N}\tau_{N}, \eta\in\mathcal{T}_N}\Ep\Big[\big\|\psi(W_{12};\theta,\eta)-\psi(W_{12};\theta_0,\eta_0)\big\|^2\Big]^{1/2}
\\
&\leq \sup_{\|\theta-\theta_{0}\| \leq CB_{1N}\tau_{N}, \eta\in\mathcal{T}_N} C_0(\|\theta-\theta_0\|\vee\|\eta-\eta_0\|)
\lesssim B_{1N}\tau_N,
\end{align*}
where the second inequality follows from  Assumption \ref{a:nonlinear_score_regularity_nuisance_parameters} (iv). 
Thus, we apply Lemma \ref{lemma_maximal_inequality_dyadic_data} to the empirical process $\{\mathbb{G}_{I_k}(f),f\in \mathcal F_2\}$ with $n=N$, $\overline{\sigma}_n=CB_{1N}\tau_N$ (where $n=N$), $v=v_N$, $b_n=D_N$, $A=D_N$, 
 an envelope $F_2=2F_1$ for sufficiently large constant $C$, and 
the function classes defined in Lemma \ref{lemma_maximal_inequality_dyadic_data} given by $\mathcal{G}=\{\Ep[f(W_{12})|U_k=\cdot]:k=1,2,f\in\mathcal{F}_2\}$ and $\mathcal{H}=\{\Ep[f(W_{12})|U_1=\cdot,U_2=\cdot]:f\in\mathcal{F}_2\}$ with envelopes $G(u)=\Ep[F_2(W_{12})|U_1=u]\vee \Ep[F_2(W_{12})|U_2=u]$ and $H(u_1,u_2)=\Ep[F_2(W_{12})|U_1=u_1,U_2=u_2]$, respectively. Applying Lemma \ref{lemma_maximal_inequality_dyadic_data}  conditional on $(W_{ij})_{(i,j) \in \Ikc}$ yields that, with $P_N$-probability $1-o(1)$,
\begin{align}
\sup _{f \in \mathcal{F}_{2}}\left|\mathbb{G}_{N}(f)\right| \lesssim B_{1N}\tau_N\sqrt{v_N\log a_N}+N^{-1/2+1/q}v_ND_N\log a_N,
\end{align}
because of $\left\|G\right\|_{P_N, q}\vee \left\|H\right\|_{P_N, q}\le  2 D_N$ following from Jensen's inequality and Assumption \ref{a:nonlinear_score_regularity_nuisance_parameters} (i) and $a_N\geq N\vee D_N$, $v_N\geq 1$, and $\sup_{g\in\mathcal{G}}\|g\|_{P_N,2}\vee \sup_{h\in\mathcal{H}}\|h\|_{P_N,2}\leq CB_{1N}\tau_N$ by Jensen's inequality and Assumption \ref{a:nonlinear_score_regularity_nuisance_parameters} (iv) a and Equation (\ref{inequality_theta_bound}), $\sup_{f\in\mathcal{F}_2}\|f\|_{P_N,2}\leq CB_{1N}\tau_N$,
 and an application of Lemma \ref{lemma_covering_entrpoy} leads to
\begin{align*}
\log \sup _{Q} N\left( \mathcal{F}_{2},\|\cdot\|_{Q, 2},\epsilon\left\|F_{2}\right\|_{Q, 2}\right) \leq 2 v_N \log (2 a_N / \epsilon), \quad \text { for all } 0<\epsilon \leq 1 ,
\end{align*}
where  $\mathcal{F}_{2} \subset \mathcal{F}_{1, \hat{\eta}_{k}}-\mathcal{F}_{1, \eta_{0}}$ with  $\calF_{1,\eta}=\{\psi(\cdot;\theta,\eta):\theta \in \Theta\}$ for fix $\eta$.

\bigskip
\noindent \textbf{Step 5.}
Here in this step, we derive a bound on $\mathcal I_3$ in (\ref{I_3}). Let $\bar{\theta}_0=\theta_0-J_0^{-1}\Enk[\psi(W;\theta_0,\eta_0)]$. Then, $||\bar{\theta}_0-\theta_0||=O_{P_N}(1/\sqrt{N})$ as $\Epn\big[\big\|\sqrt{N} \Enk[\psi\left(W ; \theta_{0}, \eta_{0}\right)]\big\|\big]$ is bounded and the singular values of $J_0$ are bounded below from zero by Assumption \ref{a:nonlinear_moment_condition} (iii). Therefore, $\bar{\theta}_0\in \Theta$ with $P_N$-probability $1-o(1)$ by Assumption \ref{a:nonlinear_moment_condition} (i). 
Hence, with the same probability,
\begin{align*}
\inf _{\theta \in \Theta} \sqrt{N}\left\|\Enk\left[\psi\left(W ; \theta, \hat{\eta}_{k}\right)\right]\right\| 
\leq \sqrt{N}\left\|\Enk\left[\psi\left(W ; \bar{\theta}_{0}, \hat{\eta}_{k}\right)\right]\right\|,
\end{align*}
and so it suffices to show that, with $P_N$-probability $1-o(1)$,
\begin{align*}
\sqrt{N}\left\|\Enk\left[\psi\left(W ; \bar{\theta}_{0}, \hat{\eta}_{k}\right)\right]\right\|=O\left(\delta_N\right).
\end{align*}
To prove this, substitute $\theta=\bar{\theta}_0$ and $\eta=\hat{\eta}_k$ into (\ref{empirical_expansion}), and use the Taylor expansion in (\ref{taylor_expansion}). 
Then, with $P_N$-probability $1-o(1)$, it holds that
\begin{align*}
 \sqrt{N}\left\|\Enk\left[\psi\left(W ; \bar{\theta}_{0}, \hat{\eta}_{k}\right)\right]\right\| 
 &
 \leq \widetilde{\mathcal{I}}_4+\widetilde{\mathcal{I}}_5+\sqrt{N}\big|\Enk[\psi(W_{ij};\theta_0,\eta_0)]+J_0(\bar{\theta}_0-\theta_0)+D_0[\hat{\eta}-\eta_0]\big|
,
\end{align*}
where $\widetilde{\mathcal{I}}_4$ and $\widetilde{\mathcal{I}}_5$ follow the definitions of $\mathcal{I}_4$ and $\mathcal{I}_5$ in (\ref{I_4}) and (\ref{I_5}) with $\theta=\bar{\theta}_0$. As $\Enk[\psi(W_{ij};\theta_0,\eta_0)]+J_0(\bar{\theta}_0-\theta_0)=0$ by the definition of $\bar{\theta}_0$,
combining this with the bounds on $\mathcal I_4$ and $\mathcal I_5$ derived above gives the claim of this step and completes the proof for the lemma.
\end{proof}

\subsection{Proof of Lemma \ref{lemma_logistic_lasso}}\label{sec:lemma_logistic_lasso}
\begin{proof}
Recall that $\theta$ is a scalar.
Observe that the score $\psi$ is nonlinear in $\theta$:
\begin{align}
\psi(W;\theta,\eta)=\{Y_{ij}-\Lambda(D_{ij}\theta+X_{ij}'\beta)\}(D_{ij}-X_{ij}' \gamma),
\end{align}
where the nuisance parameter is $\eta=(\beta',\gamma')'$.

We will make use of the short-hand notations of $\Lambda_{ij,0}$ for $\Lambda(D_{ij}\theta_{0}+X_{ij}'\beta_0)$, $\Lambda_{ij,0}^{(1)}$ for $\Lambda^{(1)}(D_{ij}\theta_0+X_{ij}'\beta_0)$, 
 and
 $\Xi_{ij}(\theta,\beta,r)$ for $D_{ij}\theta+X_{ij}'\big(\beta_0+r(\beta-\beta_0)\big)$, where $r\in (0,1)$.
We split the proof into eleven steps.

\bigskip
\noindent \textbf{Step 1.} We first verify the Assumption \ref{a:nonlinear_moment_condition} (iv). 
We have that $\Ep[\psi(W;\theta_{0},\eta_{0})]=0$ by Equation (\ref{IV_estimation_1}). 
Also, Equation (\ref{orthogonality_equation_logit}) indicates that the moment conditions are locally insensitive to the nuisance parameter $\beta$. Hence, it suffices to show that the moment conditions are locally insensitive to $\gamma$.
For any $\eta=(\beta',\gamma')'\in \mathcal{T}_N$, the Gateaux derivative in the direction $\eta-\eta_{0}=(\beta'-\beta'_0, \gamma'-\gamma'_0)'$ is given by 
\begingroup
\allowdisplaybreaks
\begin{align*}
& |\partial_{\eta} \Ep\psi(W_{ij};\theta_{0},\eta_{0})[\eta-\eta_{0}]| \\
\leq &\Big|\Ep\big[X_{ij}'(Y_{ij}-\Lambda_{ij,0})\big](\gamma-\gamma_0)\Big|+\Big|\Ep\big[X_{ij}'Z_{ij}\Lambda^{(1)}_{ij,0} \big](\beta-\beta_0)\Big|
=0,
\end{align*}
\endgroup
where the first inequality follows from the triangle inequality, and the equality follows from Equation (\ref{orthogonality_equation_logit}) and the law of iterated expectations, as $\Ep[X_{ij}'Z_{ij}\Lambda^{(1)}_{ij,0} ]=\Ep[f_{ij} X_{ij}V_{ij}]=0$ by Equation (\ref{decomposition}). 
This yields Assumption \ref{a:nonlinear_moment_condition} (iv) with $\lambda_N=0$.

\bigskip
\noindent \textbf{Step 2.}
Let us verify Assumption \ref{a:nonlinear_moment_condition} (iii).
Note that  for any $\theta\in\Theta$
\begin{align}
\label{equation_first_order_score}
\Ep[\partial_{\theta}\psi(W_{ij};\theta,\eta_0)]&=\Ep\big[-\Lambda^{(1)}(D_{ij}\theta+X_{ij}'\beta)D_{ij}Z_{ij}\big]\leq \Big(\Ep\big[D_{ij}^2\big]\Ep\big[Z_{ij}^2\big]\Big)^{1/2}
\\ \nonumber
&\leq \Big(\Ep\big[D_{ij}^4\big]\Big)^{1/4}\Big(\Ep\big[Z_{ij}^4\big]\Big)^{1/4}\leq C_1,
\end{align}
where the first inequality follows from Cauchy-Schwarz inequality and $|\Lambda^{(1)}(t)|\leq 1$ for all $t\in\mathbb{R}$, the second inequality follows from Jensen's inequality, and the last inequality follows from  Assumption \ref{a:covariates}(iii), as $\|\gamma_0\|\leq C_1$ by Assumption \ref{a:parameter}. 
Note that 
\begingroup
\allowdisplaybreaks
\begin{align*}
|J_0|=\big|\partial_{\theta}\{\Ep[\psi(W_{ij};\theta,\eta_{0})]\}|_{\theta=\theta_{0}}\big|
=\big|\Ep\big[ \Lambda_{ij,0}^{(1)}D_{ij}Z_{ij}\big]\big|
=\big|\Ep\big[f_{ij}^2D_{ij}Z_{ij}\big]\big|
=\big|\Ep\big[f_{ij}^2Z_{ij}^2\big]\big|
\gtrsim 1,
\end{align*}
\endgroup
where the second equality follows by applying Corollary 5.9 in \citet*{bartle2014elements}, as the derivative $\{\partial_{\theta}\psi(w;\theta,\eta_0):\theta\in\Theta\}$ exists and is bounded by an integrable function $w\mapsto d(d-x'\gamma_0)$ by Equation (\ref{equation_first_order_score}), the last equality follows from $\Ep[f_{ij}X_{ij}V_{ij}]=0$,
and the last inequality follows form Assumption \ref{a:covariates} (i) and Jensen's inequality, as $\|\gamma_0\|\leq C_1$ by Assumption \ref{a:parameter}.
Also 
$|J_0|=\big|\Ep\big[f_{ij}^2Z_{ij}^2\big]\big|\leq \Ep[Z_{ij}^2]\leq \big(\Ep\big[Z_{ij}^4\big]\big)^{1/2}\lesssim 1$ where the second inequality follows from Jensen's inequality and the last follows from Assumption \ref{a:covariates} (iii), as $\|\gamma_0\|\leq C_1$ by Assumption \ref{a:parameter}.
In addition, 
\begin{align*}
\Ep[\psi(W_{ij};\theta,\eta_0)]=J_0(\theta-\theta_0)+\frac{1}{2}\partial_\theta^2\{\Ep[\psi(W_{ij};\theta,\eta_0)]\}|_{\theta=\bar{\theta}}(\theta-\theta_0)^2,
\end{align*}
where $\bar{\theta}$ is between $\theta$ and $\theta_0$. Notice that  for any $\theta\in\Theta$
\begin{align}
\label{equation_second_derivative}
\Ep\big[\partial_{\theta}^2\psi(W_{ij};\theta,\eta_0)\big]
&=\Ep\big[-\Lambda^{(2)}(D_{ij}\theta+X_{ij}'\gamma_0)D_{ij}^2Z_{ij}\big]\leq \Ep\big[\big|D_{ij}^2Z_{ij}\big|\big]
\\
&\leq \Big(\Ep\big[D_{ij}^4\big]\Ep\big[Z_{ij}^2\big]\Big)^{1/2}\leq \big(\Ep\big[D_{ij}^4\big]\big)^{1/2}\big(\Ep\big[Z_{ij}^4\big]\big)^{1/4}\leq C_1,
\end{align}
where the first inequality follows from the fact that $|\Lambda^{(2)}(t)|\leq 1$ for all $t\in\mathbb{R}$, the second follows from Cauchy-Schwarz inequality, the third inequality follows from Jensen's inequality, the last follows from Assumption \ref{a:covariates} (iii), as $\|\gamma_0\|\leq C_1$ by Assumption \ref{a:parameter}. The existence of $\Ep\big[\partial_{\theta}^2\psi(W_{ij};\theta,\eta_0)\big]$ implies that $\{\partial_{\theta}^2\psi(w;\theta,\eta_0):\theta\in\Theta\}$ is bounded by an integrable function $w\mapsto d^2(d-x'\gamma_0)$ by Equation (\ref{equation_second_derivative}). Notice that $\psi(w;\theta,\eta)$ is twice continuously differentiable with respect to $\theta$ on $\Theta$ and both the first and second derivative are bounded by  integrable functions. By applying Corollary 5.9 in \citet*{bartle2014elements}
 we have
\begin{align*}
\partial_\theta^2\{\Ep[\psi(W_{ij};\theta,\eta_0)]\}|_{\theta=\bar{\theta}}=\Ep\big[\partial_\theta^2\psi(W_{ij};\theta,\eta_0)\big]|_{\theta=\bar{\theta}}\lesssim 1.
\end{align*}
This yields Assumption \ref{a:nonlinear_moment_condition} (iii). 
Given that Assumption \ref{a:nonlinear_moment_condition} (i)(ii) holds trivially, Step 1 and Step 2 show that all conditions of Assumption \ref{a:nonlinear_moment_condition} hold.

\bigskip
\noindent \textbf{Step 3.}
Next, we verify Assumption \ref{a:nonlinear_score_regularity_nuisance_parameters} (i) with $v_N=Cs_N$ and $D_N=M_N$ for a sufficient large $C$. Recall that $w=(y,d,x')'$. Define $\tau_N=C(s_N\log a_N/N)^{1/2}$ and 
\begin{align*}
\mathcal{G}_1=&\big\{w\mapsto d-x'\gamma_0\big\},
\\
\mathcal{G}_2=&\big\{w\mapsto (d,x')\xi: \xi\in\mathbb{R}^{p+1}, \|\xi\|_0\leq Cs_N,\|\xi\|\leq C\big\},
\\
\mathcal{G}_3=&\big\{w\mapsto\xi d+x'\beta_0: |\xi|\leq C\big\},
\\
T_N=&\big\{\eta=\big(\eta^{(1)},\eta^{(2)}\big):\eta^{(1)}\in\mathbb{R}^{p},\eta^{(2)}\in\mathbb{R}^{p}\big\},
\\
\mathcal{T}_N=&\{\eta_0\}\cup\big\{\eta=\big(\eta^{(1)},\eta^{(2)}\big)\in T_N:\|\eta^{(1)}\|_0\vee\|\eta^{(2)}\|_0\leq Cs_N,
\\
&\|\eta^{(1)}-\beta_0\|\vee\|\eta^{(2)}-\gamma_0\|\leq \tau_N,\|\eta^{(2)}-\gamma_0\|_1\leq C\sqrt{s_N}\tau_N\big\}
\end{align*}
for sufficiently large $C$. 
Moreover, define
\begin{align*}
\mathcal{F}_{1,1}=&\big\{w\mapsto \psi(w;\theta,\eta): \theta\in\Theta, \eta\in\mathcal{T}_N\setminus \{\eta_0\}
\big\},
\\
\mathcal{F}_{1,2}=&\big\{w\mapsto \psi(w;\theta,\eta_0): \theta\in\Theta\big\}.
\end{align*}
Then, $\mathcal{F}_1=\mathcal{F}_{1,1}\cup\mathcal{F}_{1,2}$, and 
\begin{align*}
&\mathcal{F}_{1,1} \subset\big(y-\Lambda(\mathcal{G}_{2})\big) \cdot \mathcal{G}_{2}, \\
&\mathcal{F}_{1,2} \subset\big(y-\Lambda( \mathcal{G}_{3})\big) \cdot \mathcal{G}_{1}.
\end{align*}
Pointwise measurability of these classes follows from their continuity.
Observe that $\mathcal{G}_1$ consists of a single function, which implies  that $\mathcal{G}_1$ is trivially a VC-subgraph class, and 
$\mathcal{G}_3$ is a sum of a single function and 
subset of one dimensional vector space spanned by a single function, which means $\mathcal{G}_3$ is also a VC-subgraph class implied by Example 19.17 in \cite{van1996weak}. By Theorem 2.6.7 in \cite{van1996weak}, the uniform entropy numbers obey
\begin{align}
\sup _{Q} \log N\left(\mathcal{G}_{i},\|\cdot\|_{Q, 2},\varepsilon\left\|\widetilde{G}_{i}\right\|_{Q, 2}\right) \leqslant C \log (C / \varepsilon), \quad \text { for all } 0<\varepsilon \leq 1 \text{ and } i=1,3,
\end{align}
where $\widetilde{G}_1=|d-x'\gamma_0|$ and $\widetilde{G}_3=\sup_{|\xi|\leq C}|\xi d+x'\beta_0|$  are their envelopes, and both $\widetilde{G}_1$ and $\widetilde{G}_3$ are integrable by Assumption \ref{a:covariates} (iii).  $\mathcal{G}_2$ is a union over 
${p+1\choose Cs_N}$
VC-subgraph classes with indices $O(s_N)$, since for a large $p$ such that $p\gg s_N$, ${p\choose s_N}=O(p^{s_N})$. 
Therefore, $\Lambda(\mathcal{G}_2)$ is too. 
By Lemma \ref{lemma_covering_entrpoy} (i),
\begin{align*}
\sup _{Q} \log N\left(\mathcal{F}_{1,1},\|\cdot\|_{Q, 2},\varepsilon\left\|\widetilde{F}_{1,1}\right\|_{Q, 2}\right) \leqslant C s_N\log (a_N / \varepsilon), \quad \text { for all } 0<\varepsilon \leq 1,
\end{align*}
where the envelope of $\mathcal{F}_{1,1}$ is given by
\begin{align*}
\widetilde{F}_{1,1}(W)=\sup_{\gamma\in\mathbb{R}^{p}:\|\gamma-\gamma_0\|_1\leq C\sqrt{s_N}\tau_N}2|D_{ij}-X_{ij}'\gamma|.
\end{align*}
Observe that $|D_{ij}-X_{ij}'\gamma|\leq |D_{ij}-X_{ij}'\gamma_0|+|X_{ij}'(\gamma-\gamma_0)|\leq |D_{ij}-X_{ij}'\gamma_0|+\|X_{ij}\|_\infty\|\gamma-\gamma_0\|_1\leq |D_{ij}-X_{ij}'\gamma_0|+\|X_{ij}\|_\infty C\sqrt{s_N}\tau_N$.  Using this and Jensen's inequality, we have
\begin{align*}
\|\widetilde{F}_{1,1}\|_{P,q}&\leq M_N+M_N C\sqrt{s_N}\tau_N\lesssim M_N.
\end{align*}
Also, we have
\begin{align*}
\sup _{Q} \log N\left(\mathcal{F}_{1,2},\|\cdot\|_{Q, 2},\varepsilon\left\|\widetilde{F}_{1,2}\right\|_{Q, 2}\right) \leqslant C s_N\log (a_N / \varepsilon), \quad \text { for all } 0<\varepsilon \leq 1,
\end{align*}
where 
$
\widetilde{F}_{1,2}(W)=2|Z_{ij}|
$. Then we have $\|\widetilde{F}_{1,2}\|_{P,q} \lesssim M_N$.
Therefore, by applying Lemma \ref{lemma_covering_entrpoy}, Assumption \ref{a:nonlinear_score_regularity_nuisance_parameters} (i) holds with envelope $F_1=\widetilde{F}_{1,1}\vee\widetilde{F}_{1,2}$ and $D_N=M_N$.

\bigskip
\noindent \textbf{Step 4.}  Next, we verify Assumption \ref{a:nonlinear_score_regularity_nuisance_parameters} (ii). Let us define $\mathcal{T}_N^{(2)}$ as $p+1$ to $2p$ dimensions of $\mathcal{T}_N$.
Notice that we have
\begin{align*}
\Ep\big[\psi(W_{ij};\theta,\eta)^2\big]&=\Ep\Big[\big\{Y_{ij}-\Lambda(D_{ij}\theta+X_{ij}'\beta)\big\}^2(D_{ij}-X_{ij}' \gamma)^2\Big]
\\
&\leq \big(\Ep\big[(D_{ij}-X_{ij}' \gamma)^4\big]\big)^{1/2}\lesssim 1,
\end{align*}
where the first inequality follows from Jensen's inequality and the fact that $|\Lambda(t)|\leq 1$ for all $t\in\mathbb{R}$, and the last follows from Assumption \ref{a:covariates} (iii), as $\|\gamma\|\leq \|\gamma_0\|+\sup_{\gamma\in \mathcal{T}_N^{(2)}}\|\gamma-\gamma_0\|\leq C_1+\tau_N\lesssim 1$ by Assumption \ref{a:sparsity}. Then, we have
\begin{align*}
\Ep\big[\psi(W_{ij};\theta,\eta)^2\big]&=\Ep\Big[\Ep[\big\{Y_{ij}-\Lambda(D_{ij}\theta+X_{ij}'\beta)\big\}^2|D,X](D_{ij}-X_{ij}' \gamma)^2\Big]
\\
&\geq \Ep\Big[f_{ij}^2(D_{ij}-X_{ij}' \gamma)^2\Big]\geq c_1,
\end{align*}
where the first equality follows from the law of iterated expectations and the first inequality from the fact that $f_{ij}^2=\var(Y_{ij}|D_{ij},X_{ij})$ is obtained when $\Lambda(D_{ij}\theta+X_{ij}'\beta)=E[Y_{ij}]$, and the last from Assumption \ref{a:covariates} (i) and Jensen's inequality.

\bigskip
\noindent \textbf{Step 5.} 
Next, we verify Assumption \ref{a:nonlinear_score_regularity_nuisance_parameters} (iii). Notice that Assumptions \ref{a:RE},
\ref{a:sparsity}, \ref{a:parameter}, and \ref{a:covariates} imply Assumptions \ref{a2:RE}, \ref{a2:sparsity}, and \ref{a2:covariates}. 
Then by Theorems \ref{theorem_rate_lasso} and \ref{theorem_linear_weight} invoked implied by Assumptions \ref{a:RE}
\ref{a:sparsity}, \ref{a:parameter}, and \ref{a:covariates}, with probability $1-o(1)$, 
\begin{align*}
&\|\hat{\beta}-\beta_0\|\vee \|\tilde{\beta}-\beta_0\|\lesssim \sqrt{s_N\log a_N/N}, \quad \|\hat{\beta}-\beta_0\|_1\vee \|\tilde{\beta}-\beta_0\|_1\lesssim s_N\sqrt{\log a_N/N},
\\
&  \|\hat{\gamma}-\gamma_0\|_1\vee \|\tilde{\gamma}-\gamma_0\|_1\lesssim s_N\sqrt{\log a_N/N}.
\end{align*}
$\eta_0\in\mathcal{T}_N$ holds by construction of $\mathcal{T}_N$.
Therefore, Assumption \ref{a:nonlinear_score_regularity_nuisance_parameters}(iii) holds.

\bigskip
\noindent \textbf{Step 6.} 
Now, we verify  Assumption \ref{a:nonlinear_score_regularity_nuisance_parameters} (iv) a. Define 
\begingroup
\allowdisplaybreaks
\begin{align*}
I_{1,1}&=2|X_{ij}'(\gamma-\gamma_0)|
\qquad\text{and}\qquad
I_{1,2}=|D_{ij}(\theta-\theta_0)+X_{ij}'(\beta-\beta_0)||Z_{ij}|.
\end{align*}
\endgroup
Then 
$
|\psi(W;\theta,\eta)-\psi(W;\theta_0,\eta_0)|\leq I_{1,1}+I_{1,2},
$
as $|\Lambda^{(1)}(t)|\leq 1$ for all $t\in\mathbb{R}$. 
Also, $\Ep\big[\|\psi(W_{ij};\theta,\eta)-\psi(W_{ij};\theta_0,\eta_0)\|^2\big]\lesssim \Ep\big[I_{1,1}^2\big]+\Ep\big[I_{1,2}^2\big]$.
In addition, it holds that
\begingroup
\allowdisplaybreaks
\begin{align}
\label{equation_matrix_norm}
\Ep\big[I_{1,1}^2\big]&=\Ep\Big[\big(X_{ij}'(\gamma-\gamma_0)\big)^2\Big]
=\Ep\Big[(\gamma-\gamma_0)'X_{ij}X_{ij}'(\gamma-\gamma_0)\Big]
\\&\nonumber
\leq \|\gamma-\gamma_0\|\cdot\|\Ep[X_{ij}X_{ij}']\|\cdot\|\gamma-\gamma_0\|
\leq \sup_{\|\xi\|=1}\Ep\big[(X_{ij}'\xi)^2\big]\cdot\|\gamma-\gamma_0\|^2
\\&\nonumber
\leq \Big(\sup_{\|\xi\|=1}\Ep\big[(X_{ij}'\xi)^4\big]\Big)^{1/2}\|\gamma-\gamma_0\|^2\lesssim \|\gamma-\gamma_0\|^2
\qquad\text{and}\\
\Ep\big[I_{1,2}^2\big]
&
\nonumber\leq 
\Ep\big[Z_{ij}^2\big]\times \Ep\Big[\big(D_{ij}(\theta-\theta_0)+X_{ij}'(\beta-\beta_0)\big)^2\Big]
\\
\nonumber&\lesssim  \big(\Ep\big[Z_{ij}^4\big]\big)^{1/2}\times \left(\Ep\big[D_{ij}^2\big]|\theta-\theta_0|^2+ \big\|\Ep[X_{ij}X_{ij}']\big\|\|\beta-\beta_0\|\right)
\lesssim |\theta-\theta_0|^2+\|\beta-\beta_0\|^2,
\end{align}
\endgroup
where the first inequality follows from Cauchy-Schwarz inequality, the third inequality follows from the Jensen's inequality, the fourth inequality follows from Assumption \ref{a:covariates} (iii),
the fourth line follows from Cauchy-Schwarz inequality, the first inequality in the last line follows from Jensen's inequality and Cauchy-Schwarz inequality, the last inequality follows from  Assumption \ref{a:covariates}  (iii) as $\|\gamma_0\|\lesssim 1$, and  the same argument as in (\ref{equation_matrix_norm}) to obtain the bound for $\big\|\Ep[X_{ij}X_{ij}']\big\|$. 
Therefore, Assumption \ref{a:nonlinear_score_regularity_nuisance_parameters}
(iv) a holds.

\bigskip
\noindent \textbf{Step 7.}
We verify Assumption \ref{a:nonlinear_score_regularity_nuisance_parameters} (iv) b  in this step. Let us treat $B_{1N}$ as a constant. By applying Corollary 5.9 in \citet*{bartle2014elements}, we have
\begin{align*}
\partial_r\Ep\big[\psi\big(W_{12};\theta,\eta_0+r(\eta-\eta_0)\big)\big]=\Ep\big[\partial_r \psi\big(W_{12};\theta,\eta_0+r(\eta-\eta_0)\big)\big].
\end{align*}
For any $r\in [0,1)$, let us define
\begingroup
\allowdisplaybreaks
\begin{align*}
&I_{2,1}=-X_{ij}'(Y_{ij}-\Lambda(\Xi_{ij}(\theta,\beta,r)))(\gamma-\gamma_0)
\qquad\text{and}\\
&I_{2,2}=-\Lambda^{(1)}(\Xi_{ij}(\theta,\beta,r))\big(Z_{ij}-rX_{ij}'(\gamma-\gamma_0)\big)X_{ij}'(\beta-\beta_0).
\end{align*}
\endgroup
Then, $\partial_r \psi(W;\theta,\eta_0+r(\eta-\eta_0))=I_{2,1}+I_{2,2}$ and $\Ep[\partial_r \psi(W;\theta,\eta_0+r(\eta-\eta_0))]=\Ep[I_{2,1}]+\Ep[I_{2,2}]$.
Observe that,
\begin{align*}
\Ep[|I_{2,1}|]&\leq \Ep\big[|X_{ij}'(\gamma-\gamma_0)|\big]
\leq \Big(\Ep\big[\big(X_{ij}'(\gamma-\gamma_0)\big)^2\big]\Big)^{1/2} =\Big(\Ep\big[(\gamma-\gamma_0)'X_{ij}X_{ij}'(\gamma-\gamma_0)\big]\Big)^{1/2} 
\\
&\leq \Big(\|\gamma-\gamma_0\|\cdot\|\Ep[X_{ij}X_{ij}']\|\cdot\|\gamma-\gamma_0\|\Big)^{1/2}\leq \|\gamma-\gamma_0\|\cdot \Big(\sup_{\|\xi\|=1}\Ep\big[(X_{ij}'\xi)^2\big]\Big)^{1/2}
\\
&\leq \|\gamma-\gamma_0\|\cdot \Big(\sup_{\|\xi\|=1}\Ep\big[(X_{ij}'\xi)^4\big]\Big)^{1/4}\lesssim \|\gamma-\gamma_0\|,
\end{align*}
where the first inequality follows from $|Y_{ij}-\Lambda(\Xi_{ij}(\theta,\beta,r))|\leq 1$,
the second follows from Jensen's inequality,
the third inequality follows from Cauchy-schwarz inequality, the last line follows from Jensen's inequality, and the last inequality follows from Assumption \ref{a:covariates} (iii).

Since $\|\gamma_0\|\leq C_1$ by Assumption \ref{a:parameter}, we have
 $\Ep\big[Z_{ij}^4\big]=\Ep\big[(D_{ij}-X_{ij}'\gamma_0)^4\big]\lesssim 1$ by Assumption \ref{a:covariates}(iii). Based on this, we obtain
\begin{align*}
\Ep[|I_{2,2}|]
&\leq
\Ep\Big[\big|Z_{ij}X_{ij}'(\beta-\beta_0)\big|\Big]+\Ep\Big[\big|X_{ij}'(\gamma-\gamma_0)X_{ij}'(\beta-\beta_0)\big|\Big]
\\
&\leq \left(\Ep\Big[Z_{ij}^2\Big]\cdot \Ep\Big[\big(X_{ij}'(\beta-\beta_0)\big)^2\Big]\right)^{1/2}+\left(\Ep\Big[\big(X_{ij}'(\gamma-\gamma_0)\big)^2\Big]\cdot\Ep\Big[\big(X_{ij}'(\beta-\beta_0)\big)^2\Big] \right)^{1/2}
\\
&\lesssim \left(\Ep\big[Z_{ij}^4\big]\right)^{1/4}\cdot \left(\Ep\Big[\big(X_{ij}'(\beta-\beta_0)\big)^2\Big]\right)^{1/2}+\left(\Ep\Big[\big(X_{ij}'(\gamma-\gamma_0)\big)^2\Big]\right)^{1/2}\cdot\left(\Ep\Big[\big(X_{ij}'(\beta-\beta_0)\big)^2\Big] \right)^{1/2}
\\
&\lesssim \|\beta-\beta_0\|+\|\gamma-\gamma_0\|\|\beta-\beta_0\|\lesssim \|\beta-\beta_0\|,
\end{align*}
where the first inequality follows from linearity of the expectation operator and  $|\Lambda^{(1)}(t)|\leq 1$ for all $t\in \mathbb{R}$, the second follows from Cauchy-schwarz inequality, the third follows from Jensen's inequality, 
the fourth follows because $\left(\Ep\big[Z_{ij}^4\big]\right)^{1/4}=O(1)$ by the above statement,
 and the last inequality follows because $\sup_{\gamma\in \mathcal{T}_N^{(2)}}\|\gamma-\gamma_0\|\leq C_1(s_N\log a_N/N)^{1/2}=o(1)$ by Assumptions \ref{a:sparsity} and \ref{a:covariates} (iv). 
Therefore, Assumption \ref{a:nonlinear_score_regularity_nuisance_parameters} (iv) b holds.

\bigskip
\noindent \textbf{Step 8.} We verify  Assumption \ref{a:nonlinear_score_regularity_nuisance_parameters} (iv) c in this step. Let us treat $B_{2N}$ as a constant.
 For any $\eta=(\beta',\gamma')'\in\mathcal{T}_N$ and $r\in [0,1)$, let us first consider the terms in $\Ep[\partial_r^2 \psi(W;\theta_0+r(\theta-\theta_0),\eta_0+r(\eta-\eta_0))]$. 
Define 
\begingroup
\allowdisplaybreaks
\begin{align*}
I_{3,1}=&X_{ij}'(\gamma-\gamma_0)\Lambda^{(1)}(\Xi_{ij}(\theta,\beta,r))X_{ij}'(\beta-\beta_0),
\\
I_{3,2}=&-\Lambda^{(2)}(\Xi_{ij}(\theta,\beta,r))(Z_{ij}-rX_{ij}'(\gamma-\gamma_0))(X_{ij}'(\beta-\beta_0))^2+\Lambda^{(1)}(\Xi_{ij}(\theta,\beta,r))X_{ij}'(\gamma-\gamma_0)X_{ij}'(\beta-\beta_0).
\end{align*}
\endgroup
Then, $\partial_r^2\psi(W;\theta,\eta_0+r(\eta-\eta_0))=I_{3,1}+I_{3,2}$ and $\partial_r^2\Ep[\psi(W;\theta,\eta_0+r(\eta-\eta_0))]=\Ep[I_{3,1}]+\Ep[I_{3,2}]$.
Observe that
\begingroup
\allowdisplaybreaks
\begin{align*}
\Ep[|I_{3,1}|]&\leq \Big(\Ep\big[\big(X_{ij}'(\gamma-\gamma_0)\big)^2\big]\times\Ep\big[\big(X_{ij}'(\beta-\beta_0)\big)^2\big]\Big)^{1/2}
\\
&\leq\big\|\Ep[X_{ij}X_{ij}']\big\|\|\gamma-\gamma_0\|\|\beta-\beta_0\|
\lesssim \|\eta-\eta_0\|^2,
\\
\Ep[|I_{3,2}|]&\leq \Ep\Big[Z_{ij}\big(X_{ij}'(\beta-\beta_0)\big)^2\Big]+\Ep\Big[X_{ij}'(\gamma-\gamma_0)\big(X_{ij}'(\beta-\beta_0)\big)^2\Big]
\\
&+\Big(\Ep\big[\big(X_{ij}'(\gamma-\gamma_0)\big)^2\big]\times\Ep\big[\big(X_{ij}'(\beta-\beta_0)\big)^2\big]\Big)^{1/2}
\\
&\leq 
\Big(\Ep[Z_{ij}^2]\cdot \Ep\big[\big(X_{ij}'(\beta-\beta_0)\big)^4\big]\Big)^{1/2}+\Big(\Ep\big[\big(X_{ij}'(\gamma-\gamma_0)\big)^2\big]\cdot \Ep\big[\big(X_{ij}'(\beta-\beta_0)\big)^4\big]\Big)^{1/2}+\|\eta-\eta_0\|^2
\\
& 
\lesssim \|\eta-\eta_0\|^2,
\end{align*}
\endgroup
where the first and second inequalities follow from Cauchy-Schwarz inequality and $|\Lambda^{(1)}(t)|\leq 1$ for all $t\in \mathbb{R}$, 
the last inequality in the second line follows from the same argument as in (\ref{equation_matrix_norm}) to obtain the bound for $\big\|\Ep[X_{ij}X_{ij}']\big\|$, 
 the third line follows from Cauchy-Schwarz inequality and $|\Lambda^{(2)}(t)|\leq 1$ for all $t\in\mathbb{R}$, and the bound of the first two terms for $\Ep[|I_{3,2}|]$ follows by using the same technique as the one used for the bound for $\Ep[|I_{2,2}|]$ in step 7, and the bound of the last term for $\Ep[|I_{3,2}|]$ is same as the bound for $\Ep[|I_{3,1}|]$. Also, the terms in $\Ep[\partial_r^2 \psi(W;\theta_0+r(\theta-\theta_0),\eta_0+r(\eta-\eta_0))]$ can be bounded similarly. 
Note that $B_{2N}$ is set to some appropriately large constant.
This verifies Assumption \ref{a:nonlinear_score_regularity_nuisance_parameters}
(iv).

\bigskip
\noindent \textbf{Step 9.} 
Assumption \ref{a:nonlinear_score_regularity_nuisance_parameters} (v) follows directly from Step 5.

\bigskip
\noindent \textbf{Step 10.} 
Assumption \ref{a:nonlinear_score_regularity_nuisance_parameters} (vi) is assumed in Assumption \ref{a:eigenvalue_logistic}.

\bigskip
\noindent \textbf{Step 11.} 
Finally, we verify Assumption \ref{a:nonlinear_score_regularity_nuisance_parameters} (vii).  
Condition (a) follows from construction of $\tau_N$ and $v_N$. Condition (b) holds since $N^{-1/2+1/q}v_ND_N\log a_N\lesssim N^{-1/2+1/q}s_NM_N\log a_N\lesssim \delta_N$ by Assumption \ref{a:covariates} (iv).  Condition (c) holds since $N^{1/2}B_{1N}^2B_{2N}\tau_N^2\lesssim N^{-1/2}s_N\log a_N\lesssim \delta_N$ by Assumption \ref{a:covariates} (iv). This concludes the proof.

\end{proof}

\subsection{Proof of Theorem \ref{theorem_logit}}\label{sec:theorem_logit}
\begin{proof}
Observe that the statement of Theorem \ref{theorem_logit} follows from Theorem \ref{theorem_DDML_non_linear} as soon as we can verify Assumption \ref{a:nonlinear_moment_condition} and \ref{a:nonlinear_score_regularity_nuisance_parameters} hold for the logistic case. 
This follows directly from Lemma \ref{lemma_logistic_lasso}.

\end{proof}

\section{Auxiliary Results for the Application to Logistic Regressions}
\label{lasso_converge_rate}
\subsection{Theory for Lasso Logistic Regression}

In this section, we establish rates of convergence of the Post-Lasso Logistic estimator in Section \ref{model}. 
Let $W_{ij}=(Y_{ij},D_{ij},X_{ij}')'\in [0,1]\times\mathbb{R}\times\mathbb{R}^p$, $\tilde X_{ij}=(D_{ij},X_{ij}')'$ and without loss of generality that $\|\tilde X_{\cdot,k}\|_{2,N}=1$ for all $k\in[p+1]$,  $\mathbb{E}_N[\cdot]=\frac{1}{N(N-1)}\sumij[\cdot]$, and $\Ep[Y_{ij}|D_{ij},X_{ij}]=\Lambda(\tilde X_{ij}'\xi_0)$ for all $(i,j)\in\overline{[N]^2}$, where $\xi_0=(\theta_0,\beta_0')'$.  Define the nuisance parameter $\gamma_0$ by
\begin{align*}
f_{ij}D_{ij}=f_{ij}X_{ij}' \gamma_{0} + V_{ij}, \quad \Ep[f_{ij}X_{ij}'V_{ij}]=0,
\end{align*}
where $f_{ij}^2=\Lambda(\tilde X_{ij}'\xi_0)\big(1-\Lambda(\tilde X_{ij}'\xi_0)\big)$.
The Lasso-Logistic estimator $(\hat{\theta}, \hat{\beta})$ is defined by 
\begin{align}
\label{lasso_estimator}
(\hat{\theta},\hat{\beta}') '\in \arg \min _{\theta,\beta} \mathbb{E}_N[ L(W_{ij};\theta, \beta)]+\frac{\lambda_1}{N}\| (\theta,\beta')\|_{1}.
\end{align}
The weighted-Lasso estimator $\hat\gamma$ is defined by 
\begin{align}
\hat{\gamma} \in \arg \min _{\gamma}\left(\frac{1}{2} \mathbb{E}_{N}\left[\widehat{f}_{ij}^{2}\left(D_{ij}-X_{ij}' \gamma\right)^{2}\right]+\frac{\lambda_2}{N}\left\| \gamma\right\|_{1}\right),
\end{align}
where $\hat{f}_{ij}^{2}=\Lambda(\tilde X_{ij}'\hat\xi)\big(1-\Lambda(\tilde X_{ij}'\hat\xi)\big)$. Denote $Z_{ij}=D_{ij}-X_{ij}\gamma_0'$.

For any $T\subset [p]$, $\delta=(\delta_1,...,\delta_{p+1})'\in \Real^{p+1}$, denote $\delta_T=(\delta_{T,1},...,\delta_{T,p+1})'$ with $\delta_{T,j}=\delta_j$ if $j\in T$ and $\delta_{T,j}=0$ if $j\not \in T$. for the vector that Define the minimum and maximum sparse eigenvalues by 
\begin{align*}
\phi_{\min}(m)=
\inf _{\|\delta\|_0 \le m} \frac{\|f_{ij}\tilde X_{ij}'\delta\|_{2,N}}{\|\delta_{T}\|_1}, \quad \phi_{\max}(m)=
\sup _{\|\delta\|_0 \le m} \frac{\|f_{ij}\tilde X_{ij}'\delta\|_{2,N}}{\|\delta_{T}\|_1}.
\end{align*}
We state the following assumptions on Lasso Logit model.
\begin{assumption}(Sparse eigenvalue conditions).
\label{a2:RE}
The sparse eigenvalue conditions hold with probability $1-o(1)$, namely, for some $\ell_N\to \infty$ slow enough, we have
\begin{align*}
1\lesssim  \phi_{\min}(\ell_N s_N)\le
 \phi_{\max}(\ell_N s_N)\lesssim 1.
\end{align*}
\end{assumption}

\begin{assumption}(Sparsity).
\label{a2:sparsity}
There exists $s_N$, such that  $\left\|\theta_0\right\|_{0}+\left\|\beta_0\right\|_{0}+\max_{k\in[p]}\|\gamma_0^k\|_{0} \leqslant s_{N}$.
\end{assumption}
Let $M_{N}$ be a sequence of positive constants such that 
$M_{N}\geq \big(\mathrm{E}_{P}\big[(D_{12} \vee\|X_{12}\|_{\infty})^{2 q}\big]\big)^{1 /2 q}$. Also, denote $\Delta_{N}$ for a sequence of positive constants that converges to zero. 
\begin{assumption}(Covariates).
\label{a2:covariates}
 For $q>4$, the following inequalities hold:
\begin{enumerate}[(i)]
\item $\max_{j=1,2}\inf _{\|\xi\|=1} \mathrm{E}_{P}\big[\big(\Ep [f_{12}\left(D_{12}, X_{12}'\right) \mid U_{j}]\xi\big)^{2}\big] \geqslant c_{1}$.
\item $\max_{j=1,2}\min_{l\in[p]} \big(\mathrm{E}_{P}\big[\big(\Ep [f_{12}^{2} Z_{12} X_{12,l}\mid U_j]\big)^{2}\big] \wedge \mathrm{E}_{P}\big[\big(\Ep [f_{12}^{2} D_{12} X_{12,l}\mid U_j]\big)^{2}\big] \big) \geqslant c_{1}$.
\item $\sup _{\|\xi\|=1} \mathrm{E}_{P}\big[\big((D_{12}, X_{12}) \xi\big)^{4}\big] \leqslant C_{1}$.

\item $N^{-1/2+2/q}M_{N}^2s_N\log^2 a_N\leq \Delta_N $.

\end{enumerate}
\end{assumption}

\begin{theorem}(Rates of Convergence)
\label{theorem_rate_lasso}
Suppose that Assumptions \ref{a2:RE}, \ref{a2:sparsity}, and \ref{a2:covariates} hold. In addition, suppose that the penalty choice  $\lambda_1=K_1 \sqrt{N\log (pN)}$ for a $K_1>0$. Then there exist a constant $C$ such that with probability $1-o(1)$,
\begingroup
\allowdisplaybreaks
\begin{align*}
&\|f_{ij}\tilde X_{ij}'(\hat\beta-\beta_0)\|_{2,N}\vee\|f_{ij}\tilde X_{ij}'(\tilde\beta-\beta_0)\|_{2,N}\vee\|\hat\beta-\beta_0\|\vee \|\tilde\beta-\beta_0\|\lesssim\sqrt{\frac{s_N\log(pN)}{N}},
\\
&\|\hat\beta-\beta_0\|_1\vee \|\tilde\beta-\beta_0\|_1\lesssim\sqrt{\frac{s_N^2\log (pN)}{N}}.
\end{align*}
\endgroup
\end{theorem}
\begin{proof}
We will apply Lemmas 1 and 2 in \cite{BelloniChernozhukovWei2016} and then obtain a high-probability bound for $\lambda_{1}$.
First note that by Lemma 2.7 in \cite{LecueMendelson2017}, Assumption \ref{a2:RE} implies that the following restricted eigenvalue condition holds with probability $1-o(1)$:  for $T=\supp(\theta_0,\beta_0)$, $|T|\geq 1$, and $c\geq 1$,
\begin{align*}
\kappa_{c_0}=\inf_{\delta\in\mathfrak D_{c_0}}\frac{\|f_{ij}\tilde X_{ij}'\delta\|_{2,N}}{\|\delta_{T}\|_1}>0
\end{align*}
where $\mathfrak D_{c_0}=\{\delta:\|\delta_{T^c}\|_1\leq c_0 \|\delta_T\|_1\}$, $c_0=(c+1)/(c-1)$.

\noindent \textbf{Step 1.}
For a subset $A\subset \mathbb{R}^{p+1}$, define what is known in the literature as the nonlinear impact coefficient by
\begin{align*}
\bar{q}_A=\inf _{\delta \in A} \frac{\mathbb{E}_{N}\left[f_{ij}^2\left|\tilde{X}_{ij}^{\prime} \delta\right|^{2}\right]^{3 / 2} }{ \mathbb{E}_{N}\left[f_{ij}^2\left|\tilde{X}_{ij}^{\prime} \delta\right|^{3}\right]}.
\end{align*}
To apply their Lemma 1,  we verify the condition $\bar{q}_{\mathfrak D_{c_0}}>3\left(1+\frac{1}{c}\right) \lambda_1 \sqrt{s_N} /\left(n \kappa_{c_0}\right)$ with probability $1=o(1)$. Observe that
\begin{align*}
\bar{q}_{\mathfrak D_{c_0}}=&\inf _{\delta \in \mathfrak D_{c_0}} \frac{\mathbb{E}_{N}\left[f_{ij}^2\left|\tilde{X}_{ij}^{\prime} \delta\right|^{2}\right]^{3 / 2} }{ \mathbb{E}_{N}\left[f_{ij}^2\left|\tilde{X}_{ij}^{\prime} \delta\right|^{3}\right]}
\ge
\inf _{\delta \in \mathfrak D_{c_0}} \frac{\mathbb{E}_{N}\left[w_{ij}\left|\tilde{X}_{ij}^{\prime} \delta\right|^{2}\right]^{1 / 2} }{ \max_{(i,j)\in \overline{[N]^2}}\|\tilde X_{ij}\|_\infty\|\delta\|_1}
\gtrsim_P 
\inf _{\delta \in \mathfrak D_{c_0}} \frac{\mathbb{E}_{N}\left[w_{ij}\left|\tilde{X}_{ij}^{\prime} \delta\right|^{2}\right]^{1 / 2} }{N^{1/q} M_N\|\delta\|_1}\\
\ge& 
\inf _{\delta \in \mathfrak D_{c_0}} \frac{\mathbb{E}_{N}\left[w_{ij}\left|\tilde{X}_{ij}^{\prime} \delta\right|^{2}\right]^{1 / 2} }{N^{1/q} M_N(1+c_0)\sqrt{s_N}\|\delta_T\|}\ge \frac{\kappa_{c_0}}{N^{1/q} M_N(1+c_0)\sqrt{s_N}}\ge \frac{1}{\Delta_N^{1/2} N^{1/4}}\gtrsim \sqrt{\frac{s_N \log a_N}{\Delta_N  N}},
\end{align*}
where we have used restricted eigenvalue condition, $M_{N}s_N/ N^{1/2-2/q}\leq \Delta_N$ and $s_N\log a_N/ N^{1/2}\leq \Delta_N$, both implied by Assumption \ref{a2:covariates} (iv), $\Delta_N=o(1)$, as well as the choice of $\lambda_{1}=(N \log a_N)^{1/2}$.
Now, we can invoke the first part of Lemma 1 in \cite{BelloniChernozhukovWei2016} and obtain
\begin{align*}
\|f_{ij}\tilde X_{ij}'(\hat \xi-\xi_0)\|_{2,N}	=O\left(\frac{\lambda_1 \sqrt{s_N}}{N}\right),\quad \|\hat \xi-\xi_0\|_{1}	=O\left(\frac{\lambda_1 s_N}{N}\right).
\end{align*}
Furthermore, note that $\max_{(i,j)\in\NN}\|\tilde X_{ij}\|_\infty\|\hat \xi-\xi_0\|_{1}=(N^{1-2/q}M_N^2 s_N\log a_N)^{1/2}=o(1)$. By the second part of Lemma 1 in \cite{BelloniChernozhukovWei2016}, $\|\hat \xi\|_0\lesssim s_N$ and $\EN[ L(W_{ij};\hat\theta, \hat\beta)-L(W_{ij};\theta_0, \beta_0)]=O\left(N^{-1}\lambda_1 \sqrt{s_N}\right)$.

\noindent \textbf{Step 2.}
We invoke Lemma 2 in \cite{BelloniChernozhukovWei2016}. First let us verify the condition that for some $C>0$ such that $s_N+\hat s_N\le Cs_N$, where $\hat s_N=\|\hat \xi\|_0$, and $A_{Cs_N}=\{\delta\in \Real^{p+1}: \|\delta\|_0\le Cs_N  \}$, then
\begin{align*}
\bar q_{A_{Cs_N}}>\max\left\{ 
\frac{\sqrt{Cs_N}\|\nabla \EN[ L(W_{ij};\xi_0)]\|_\infty}{\sqrt{\phi_{\min}(Cs_N)} },\sqrt{0\vee\left\{\EN[ L(W_{ij};\tilde\xi)- L(W_{ij};\xi_0)]\right\} }
 \right\}=O_P\left(\sqrt{\frac{s_N\log a_N}{N}}\right).
\end{align*}
where the equality follows from regularized event and the previous step.
Observe that
\begin{align*}
\bar{q}_{A_{Cs_N}}=&\inf _{\delta \in A_{Cs_N}} \frac{\mathbb{E}_{N}\left[f_{ij}^2\left|\tilde{X}_{ij}^{\prime} \delta\right|^{2}\right]^{3 / 2} }{ \mathbb{E}_{N}\left[f_{ij}^2\left|\tilde{X}_{ij}^{\prime} \delta\right|^{3}\right]}
\ge
\inf _{\delta \in A_{Cs_N}} \frac{\mathbb{E}_{N}\left[w_{ij}\left|\tilde{X}_{ij}^{\prime} \delta\right|^{2}\right]^{1 / 2} }{ \max_{(i,j)\in \overline{[N]^2}}\|\tilde X_{ij}\|_\infty\|\delta\|_1}\\
\ge&
\inf _{\|\delta\|_0 \le Cs_N} \frac{\mathbb{E}_{N}\left[w_{ij}\left|\tilde{X}_{ij}^{\prime} \delta\right|^{2}\right]^{1 / 2} }{ \max_{(i,j)\in \overline{[N]^2}}\|\tilde X_{ij}\|_\infty\sqrt{Cs_N}\|\delta\|}
\gtrsim_P 
\frac{\sqrt{\phi_{\min}(Cs_N)}}{\sqrt{Cs_N}N^{1/q} M_N}\gtrsim \frac{\log^{1/4} a_N}{\Delta_N^{1/2} N^{1/4}}\gtrsim\sqrt{\frac{s_N\log a_N}{N}},
\end{align*}
where we have used sparse eigenvalue condition and Assumption \ref{a2:covariates} (iv). 
Now, invoke Lemma 2 in \cite{BelloniChernozhukovWei2016} and obtain
\begin{align*}
\|f_{ij}\tilde X_{ij}'(\tilde \xi-\xi_0)\|_{2,N}	=O\left(\frac{\lambda_1 \sqrt{s_N}}{N}\right),\quad \|\tilde \xi-\xi_0\|_{1}	=O\left(\frac{\lambda_1 s_N}{N}\right).
\end{align*}

\noindent \textbf{Step 3.}
We now claim that, for some $K>0$ large enough, let $\zeta\in(0,1)$ and 
\begin{align*}
\lambda_1=K\sqrt{N\log (p/\zeta)},
\end{align*}
then with probability $1-\zeta-o(1)$, for a fixed $c>1$, it holds that
\begin{align*}
P(\lambda_1/N\geq c\|\nabla L(\theta_0,\beta_0)\|_{\infty})\geq 1-\zeta-o(1).
\end{align*}
The proof relies on Theorem 3 in  \cite{chiang2020inference}. First, we verify its required conditions. Condition (3.3) in  \cite{chiang2020inference} is directly implied by definition of $M_N$ and Assumption \ref{a2:covariates} with their $q$ set to be $2q$ here. For $N$ large enough, their Conditions (3.4) and (3.5) are implied by Assumption \ref{a2:covariates} (iii) and (i), respectively. 
	Now, by Theorem 3 in  \cite{chiang2020inference}, we have
	\begin{align*}
	\sup_{t\in \Real} |P(\|\sqrt{N}\nabla L(\theta_0,\beta_0)\|_\infty\le t)-P(\|\mathbf G\|_\infty\le t)|=o(1), 
	\end{align*}
	where $\mathbf G\sim N(0,\Sigma)$, $\Sigma$ is the asymptotic variance of $\sqrt{N}\nabla L(\theta_0,\beta_0)$.
	Then the Gaussian concentration inequality implies that with probability $1-\zeta-o(1)$,
	\begin{align*}
	P(\lambda_1/N\geq c\|\nabla L(\theta_0,\beta_0)\|_{\infty})\geq 1-\zeta-o(1).
	\end{align*}
Now, combining the result with the bound from Step 2 concludes the proof.
\end{proof}

\subsection{Theory for Linear Regression with Estimated Weight}
In this section, we establish rates of convergence  of Post-Lasso-OLS estimator in Section \ref{model}.

\begin{theorem}(Rates and Sparsity for Lasso with Estimated Weights)
\label{theorem_linear_weight}
Suppose that Assumptions \ref{a2:RE}, \ref{a2:sparsity} and  \ref{a2:covariates} hold for all $P\in \mathcal{P}_N$. In addition, suppose that the penalty level $\lambda_2=K_2 \sqrt{N\log (pN)}$,  $K_2>0$ a positive constant. Then uniformly over $P\in\mathcal{P}_N$, with probability $1-o(1)$,
\begin{align*}
&\|\hat f_{ij}X_{ij}'(\hat{\gamma}-\gamma_0)\|_{2,N}\vee \|\hat f_{ij}X_{ij}'(\tilde{\gamma}-\gamma_0)\|_{2,N}\vee\|X_{ij}'(\hat{\gamma}-\gamma_0)\|\vee\|X_{ij}'(\tilde{\gamma}-\gamma_0)\|
\leq \sqrt{\frac{s_N\log a_N}{N}},\\
 &\|\hat{\gamma}-\gamma_0\|_{1}\vee\|\tilde{\gamma}-\gamma_0\|_{1}\leq \sqrt{\frac{s_N^2\log a_N}{N}}.
\end{align*}
\end{theorem}
\begin{proof}
We first apply Lemmas L.1, L.2, and L.3 in \cite{BCCW18}
with $(Y,X,\theta,a,w)=(D,X,\gamma,f,\hat f^2)$, $M(D,X,\gamma,f)=2^{-1}f^2(D-X'\gamma)^2$. Following the verification of Assumption L.1 in the proof of Theorem 4.2 in \cite{BCCW18}, it suffices to show that with probability $1-o(1)$
\begin{align*}
\|(\hat f_{ij }^2 -f_{ij}^2)Z_{ij}/\hat f_{ij}\|_{2,N} \lesssim \sqrt{\frac{s_N\log a_N}{N}}.
\end{align*}
Following the argument of Equation (J.1) in \cite{BCCW18}, we have with probability $1-o(1)$
\begin{align*}
|\hat f_{ij}^2 - f_{ij}^2|\le f_{ij}^2/2,
\end{align*}
which, by reverse triangle inequality, implies that with probability $1-o(1)$, we have
\begin{align*}
\hat f_{ij}^2=|f_{ij}^2-(f_{ij}^2-\hat f_{ij}^2)|\ge \left| |f_{ij}^2|-|f_{ij}^2-\hat f_{ij}^2|\right|\ge f_{ij}^2/2
\end{align*}
and thus following Theorem \ref{theorem_rate_lasso}, we have the desired bound that  with probability $1-o(1)$,
\begin{align*}
\|(\hat f_{ij }^2 -f_{ij}^2)Z_{ij}/\hat f_{ij}\|_{2,N} \lesssim \|(\hat f_{ij }^2 -f_{ij}^2)Z_{ij}/f_{ij}\|_{2,N}\lesssim\sqrt{\frac{s_N\log a_N}{N}}.
\end{align*}
Now, Assumption L.1(b) is trivial and Assumption L.1(a) and (c) can be verified using exactly the same argument as in the proof of Theorem 4.2 in \cite{BCCW18}. By invoking Lemma L.1 in \cite{BCCW18}, we have
\begin{align*}
\|\hat f_{ij}X_{ij}'(\hat{\gamma}-\gamma_0)\|_{2,N}\leq \sqrt{\frac{s_N\log a_N}{N}},\quad \|\hat{\gamma}-\gamma_0\|_{1}\leq \sqrt{\frac{s_N^2\log a_N}{N}}.
\end{align*} 
Following the same arguments as in the proof of Theorem 4.2 in \cite{BCCW18}, Lemmas L.2 and L.3 in \cite{BCCW18} can be applied to obtain 
\begin{align*}
\|\hat f_{ij}X_{ij}'(\tilde{\gamma}-\gamma_0)\|_{2,N}\leq \sqrt{\frac{s_N\log a_N}{N}},\quad \|\tilde{\gamma}-\gamma_0\|_{1}\leq \sqrt{\frac{s_N^2\log a_N}{N}}
\end{align*} 
conditional on the event of $\lambda_2/N\geq c\|\EN [f_{ij}X_{ij}'V_{ij}]  \|_{\infty}$ and $\hat f_{ij}^2 \ge f_{ij}^2/2$.

Finally, by the same argument as in Step 3, in the proof of Theorem \ref{theorem_rate_lasso}, for each $\zeta\in(0,1)$, we have, with probability $1-\zeta-o(1)$, for a fixed $c>1$, it holds that
\begin{align*}
P(\lambda_2/N\geq c\|\EN [f_{ij}X_{ij}'V_{ij}]  \|_{\infty})\geq 1-\zeta-o(1).
\end{align*}
A combination of the above results concludes the proof.
\end{proof}


\section{Useful Lemmas}
This section contains some useful lemmas used in the proofs of other results in the paper.
The following result is the same as Lemma 6.1 in \cite{CCDDHNR18}.
\begin{lemma}[Conditional Convergence Implies Unconditional]\label{lemma:conditional_convergence}
Let $(X_n)$ and $(Y_n)$ be sequences of random vectors.
\begin{enumerate}[(i)]
\item If for $\epsilon_n\to 0$, $P(\|X_n\|>\epsilon_n|Y_n)=\op(1)$ in probability, then $P(\|X_n\|>\epsilon_n)=o(1)$. In particular, this occurs if $\Ep[\|X_n\|^q/\epsilon_n^q|Y_n]=\op(1)$ for some $q\ge 1$.
\item Let $(A_n)$ be a sequence of positive constants. If $\|X_n\|=\Op(A_n)$ conditional on $Y_n$, then $\|X_n\|=\Op(A_n)$ unconditional, namely, for any $l_n\to \infty$, 
$P(\|X_n\|>l_n A_n)=o(1)$.
\end{enumerate}
\end{lemma}

The following result restates Proposition 5 and Corollary 7 in \cite{Kato2017lecture}.
\begin{lemma}(Algebra for Covering Entropies)
\label{lemma_covering_entrpoy}
\begin{enumerate}[(i)]
\item
For any measurable classes of function $\mathcal{F}$ and $\mathcal{F}'$ mapping $\mathcal{W}$ to $\mathbb{R}$,
\begin{align*}
&\log N(\mathcal{F}+\mathcal{F}',\|\cdot\|_{Q,2},\varepsilon\|F+F'\|_{Q,2})\leq \log N(\mathcal{F},\|\cdot\|_{Q,2},\frac{\varepsilon}{2}\|F\|_{Q,2})+\log N(\mathcal{F}',\|\cdot\|_{Q,2},\frac{\varepsilon}{2}\|F'\|_{Q,2}),
\\
&\log N(\mathcal{F}\cdot\mathcal{F}',\|\cdot\|_{Q,2},\varepsilon\|F\cdot F'\|_{Q,2})\leq \log N(\mathcal{F},\|\cdot\|_{Q,2},\frac{\varepsilon}{2}\|F\|_{Q,2})+\log N(\mathcal{F}',\|\cdot\|_{Q,2},\frac{\varepsilon}{2}\|F'\|_{Q,2}),
\\
& N(\mathcal{F}\cup \mathcal{F}',\|\cdot\|_{Q,2},\varepsilon\|F\vee F'\|_{Q,2})\leq N(\mathcal{F},\|\cdot\|_{Q,2},\varepsilon\|F\|_{Q,2})+N(\mathcal{F}',\|\cdot\|_{Q,2},\varepsilon\|F'\|_{Q,2}).
\end{align*}
\item 
Given measurable classes $\mathcal{F}_j$ and envelopes $F_j$, $j=1,...,k,$mapping $\mathcal{W}$ to $\mathbb{R}$, a function $\phi:\mathbb{R}^k\rightarrow \mathbb{R}$ such that for $f_j, g_j\in\mathcal{F}_j$, $|\phi(f_1,...,f_k)-\phi(g_1,...,g_k)|\leq \sum_{j=1}^k L_j(x)|f_j(x)-g_j(x)|$, $L_j(x)\geq 0$, and fixed functions $\bar{f}_j\in \mathcal{F}_j$, the class of functions $\mathcal{L}=\{\phi(f_1,...,f_k)-\phi(\bar{f}_1,...,\bar{f}_k):f_j\in\mathcal{F}_j,j=1,...,k\}$ satisfies
\begin{align*}
\sup_{Q }\log N(\mathcal{L},\|\cdot\|_{Q,2},\varepsilon\|\sum_{j=1}^k L_jF_j\|_{Q,2})\leq \sum_{j=1}^k\log\sup_Q N(\mathcal{F}_j,\|\cdot\|_{Q,2},\frac{\varepsilon}{k}\|F_j\|_{Q,2}),\text{ for all } \:0<\varepsilon\le 1.
\end{align*}
\end{enumerate}
\end{lemma}

\section{Proofs of the Results in Appendix \ref{sec:lemmas_dyadic}}\label{sec:proofs_for_results_appendix}

\subsection{Proof of Lemma \ref{lemma_maximal_inequality_dyadic_data}}\label{sec:lemma_maximal_inequality_dyadic_data}
\begin{proof}
Consider the decomposition
\begin{align*}
\frac{\sqrt{n}}{n(n-1)}\sumijn f(X_{ij})=\frac{1}{\sqrt{n}}\sum_{k=1}^n W_k(f) + \frac{\sqrt{n}}{n(n-1)}\sumijn R_{ij}(f),
\end{align*}
where $W_k(f)=\Ep\left[\frac{1}{(n-1)}\left(\sum_{j\ne k}f(X_{kj})+\sum_{i\ne k}f(X_{ik})\right)\Big|U_k\right]$ and $R_{ij}(f)= f(X_{ij})-\Ep[f(X_{ij})|U_i]-\Ep[f(X_{ij})|U_j]$. The first term on the right hand side consists of the Hajek projections and the second is a remainder term from first order projection.

First, the Hajek projection terms 
\begin{align*}
\frac{1}{\sqrt{n}}\sum_{k=1}^n W_k(f)= \frac{1}{\sqrt{n}}\sum_{k=1}^n\Ep\left[\frac{1}{(n-1)}\left(\sum_{j\ne k}f(X_{kj})+\sum_{i\ne k}f(X_{ik})\right)\Big|U_k\right]
\end{align*}
consist of i.i.d. terms $(W_k)_{k=1}^n$. Thus we can apply Theorem 5.2 in \cite{CCK14}  and obtain
\begin{align*}
\Ep\left[\sup_{f\in \calF}\left|\frac{1}{\sqrt{n}}\sum_{k=1}^n W_k(f)\right|\right]
\lesssim&\:
\overline \sigma_n\sqrt{v\log(A\vee n)}
+ \frac{b_n v\log(A\vee n)}{n^{1/2-1/q}}.
\end{align*}

Now, for Hajek projection errors, a further decomposition yields
\begingroup
\allowdisplaybreaks
\begin{align*}
&\frac{\sqrt{n}}{n(n-1)}\sumijn R_{ij}(f)\\
=&\frac{\sqrt{n}}{n(n-1)}\sumijn \underbrace{\{\Ep[f(X_{ij})|U_i,U_j]-\Ep[f(X_{ij})|U_i]-\Ep[f(X_{ij})|U_j]\}}_{:=\hat R_{ij}(f)}\\
&+\frac{\sqrt{n}}{n(n-1)}\sumijn \underbrace{\{f(X_{ij})-\Ep[f(X_{ij})|U_i,U_j]\}}_{:=\tilde R_{ij}(f)}.
\end{align*}
\endgroup
Conditional on $(U_i)_{i=1}^n$, the second term on the right hand side is a sum of independent but not necessarily identically distributed random variables. In addition, by Lemma A.2. in \cite{GhosalSenvanderVaart2000}, $\calH$ is a VC type with characteristics $4\sqrt{A}$ and $v$ for envelope $H$.  Thus, Lemma B.2 in \cite{cattaneo2022} and an application of Jensen's inequality yield
\begin{align*}
\Ep\left[\sup_{f\in \calF}\left|\frac{\sqrt{n}}{n(n-1)}\sumijn \tilde R_{ij}(f) \right|\right]
\lesssim&
\frac{\sigma_n \sqrt{v \log (A\vee n)}}{\sqrt{n}}+ \frac{b_n v \log (A\vee n)}{n^{3/2-2/q}}.
\end{align*}
On the other hand, the first term consists of a completely degenerate U-process of order two and thus an upper bound can be obtained by applying Corollary 5.5 of \cite{ChenKato2019},
\begin{align*}
\Ep\left[\sup_{f\in \calF}\left|\frac{\sqrt{n}}{n(n-1)}\sumijn \hat R_{ij}(f) \right|\right]
\lesssim&
\frac{\sigma_n v \log (A\vee n)}{\sqrt{n}}+ \frac{b_n (v \log (A\vee n))^2}{n^{1-2/q}}.
\end{align*}
The desired bound follows from combining the bounds.
\end{proof}

\subsection{Proof of Lemma \ref{lemma:hajek}}\label{sec:lemma:hajek}
\begin{proof}
We will write $\sum_{i\ne l}$ for $\sum_{i\in [n]\setminus \{l\}}$.
Note that the H\'ajek projection on functions of each single $(U_l)_{l=1}^n $ is
\begingroup
\allowdisplaybreaks
\begin{align*}
\sumln \tE\left[\frac{\sqrt{n}}{n(n-1)}\sum\limits_{(i,j) \in \overline{[n]^2}} f(Z_{ij})\Bigg|U_l\right]
=&
\sumln \frac{\sqrt{n}}{n(n-1)}\sum\limits_{(i,j)\in \overline{[n]^2}} \1\{i=l \text{ or } j=l\}\tE\left[f(Z_{ij})|U_l\right]\\
=&
\frac{\sqrt{n}}{n(n-1)}\sumln \left\{\sum_{j\ne l}\tE[f(Z_{lj})|U_l] + \sum_{i\ne l}\tE[f(Z_{il})|U_l]\right\}.
\end{align*}
\endgroup
The right-hand side equals the definition of $\Hn f$ and the summands are independent over $l$. To calculate the variance, for a fixed $l$,
\begingroup
\allowdisplaybreaks 
\begin{align*}
 \V\left(\sum_{j\ne l}\tE[f(Z_{lj})|U_l] + \sum_{i\ne l}\tE[f(Z_{il})|U_l]\right)
 =& 
 \V\left(\sum_{j\ne l}\tE[f(Z_{lj})|U_l]\right) 
 +
  \V\left(\sum_{i\ne l}\tE[f(Z_{il})|U_l]\right)
 \\
 &+2\Cov \left(\sum_{j\ne l}\tE[f(Z_{lj})|U_l], \sum_{i\ne l}\tE[f(Z_{il})|U_l]\right)\\
 =&(1)+(2)+2 \times (3). 
\end{align*}
\endgroup
For $(1)$, one has
\begingroup
\allowdisplaybreaks
\begin{align*}
\V\left(\sum_{j\ne l}\tE[f(Z_{lj})|U_l]\right) 
=&
\sum_{j\ne l} \sum_{\jmath \ne l}\Cov\left(\tE[f(Z_{lj})|U_l],\tE[f(Z_{l\jmath})|U_l]\right)\\
=&
\sum_{j\ne l} \sum_{\jmath \ne l}\left\{\Cov\left(f(Z_{lj}),f(Z_{l\jmath})\right)-\tE[\Cov\left(f(Z_{lj}),f(Z_{l\jmath})|U_l\right)]\right\}\\
=&
\sum_{j\ne l} \left(\sum_{\jmath \ne l}\Cov\left(f(Z_{lj}),f(Z_{l\jmath})\right)-E[\V(f(Z_{lj})|U_l )]\right)\\
=& (n-1)(n-2) \Cov(f(Z_{12}),f(Z_{13}))+O(n) ,
\end{align*}
\endgroup
where the second equality follows from the law of total covariance and the third and the fourth follow from the AHK representation. Similarly, for $(2)$, one has 
\begin{align*}
 \V\left(\sum_{i\ne l}\tE[f(Z_{il})|U_l]\right)
 =
  (n-1)(n-2) \Cov(f(Z_{21}),f(Z_{31}))+O(n).
\end{align*}
In addition, for $(3)$, we have
\begingroup
\allowdisplaybreaks
\begin{align*}
\Cov \left(\sum_{j\ne l}\tE[f(Z_{lj})|U_l], \sum_{i\ne l}\tE[f(Z_{il})|U_l]\right)
=&
\sum_{j\ne l}\sum_{i\ne l}\Cov \left(\tE[f(Z_{lj})|U_l], \tE[f(Z_{il})|U_l]\right)\\
=&
\sum_{j\ne l}\sum_{i\ne l}\left\{\Cov \left(f(Z_{lj}), f(Z_{il})\right)
-
\tE[\Cov (f(Z_{lj}), f(Z_{il})|U_l)]
\right\}\\
=&
 (n-1)(n-2)\Cov(f(Z_{12}),f(Z_{31}))+O(n),
\end{align*}
\endgroup
where the second equality follows from the law of total covariance and the third and the fourth follow from the AHK representation. Also, by symmetry of covariance, we have
\begingroup
\allowdisplaybreaks
\begin{align*}
\Cov \left(\sum_{j\ne l}\tE[f(Z_{lj})|U_l], \sum_{i\ne l}\tE[f(Z_{il})|U_l]\right)
=&
\Cov \left(\sum_{i\ne l}\tE[f(Z_{il})|U_l],\sum_{j\ne l}\tE[f(Z_{lj})|U_l]\right)\\
=&
 (n-1)(n-2)\Cov(f(Z_{21}),f(Z_{13}))+O(n),
\end{align*}
\endgroup
Combining all these, one has  
\begin{align*}
\V(\Hn f)= \Cov(f(Z_{12}),f(Z_{13}))+\Cov(f(Z_{12}),f(Z_{31}))+ \Cov(f(Z_{21}),f(Z_{13})) + \Cov(f(Z_{21}),f(Z_{31}))
\\
+o(1).
\end{align*}
\end{proof} 

\section{Data Appendix}\label{sec:data_appendix}
The data set that we use was originally generated by \citet{head2010erosion}.
Additional sources are mentioned for each variable below. 
The details of the variables that we use in our analysis are as follows.
\begin{enumerate}[(i)]
\item\label{iso3_o} iso3\_o: 
Standard ISO code for exporting country (three letters).
\item\label{iso3_d} iso3\_d: 
Standard ISO code for importing country (three letters).
\item\label{year} year: 
Numeric, from 1948 to 2015.

\item\label{fta_bb} fta\_bb:
Variable coded as 1 for Free Trade Area; 2 for Customs Union; 3 for Common Market; 4 for Economic Union.
Source: \citet{baier2009estimating}.
\item\label{fta_hmr} fta\_hmr:
Dummy for Free Trade Agreement.
Source: \citet{head2010erosion}.

\item\label{distw} distw: 
Weighted bilateral distance between origin and destination in kilometer (population weighted).
Source: CEPII Distance Data set
\item\label{tdiff} tdiff:
Time difference between origin and destination, in number of hours.
For countries which stretch over more than one time zone, the respective time zone is generated via the mean of all its time zones (for instance: Russia, Canada, USA)

\item\label{gdpcap_o} gdpcap\_o:
Gross Domestic Product per capita of origin (current US\$)
From 1960 to 2015, the data comes from the World Development Indicators, World Bank.
For Taiwan, the data comes from the Directorate-General of Budget, Accounting and Statistics (DGBAS).
\item\label{gdpcap_d} gdpcap\_d:
Gross Domestic Product per capita of destination (current US\$)
Source: computed from gdp\_d and pop\_d

\item\label{colony} colony:
Dummy for origin and destination ever in colonial relationship.
Source: \citet{head2010erosion}.
\item\label{col45} col45:
Dummy for origin and destination in colonial relationship post 1945.
Source: \citet{head2010erosion}.
\item\label{col_to} col\_to:
Dummy for origin and destination in colonial relationship post 1945.
Source: \citet{head2010erosion}.
\item\label{col_fr} col\_fr:
Dummy for origin and destination in colonial relationship post 1945.
Source: \citet{head2010erosion}.
\item\label{comcol} comcol:
Dummy for common colonizer of origin and destination post 1945.
Source: \citet{head2010erosion}.
\item\label{comlang_off} comlang\_off:
Dummy for common official or primary language.
Source: \citet{head2010erosion}.
\item\label{comlang_ethno} comlang\_ethno:
Dummy for language spoken by at least 9\% of the population in both countries.
Source: \citet{head2010erosion}.
\item\label{heg_o} heg\_o:
Dummy if origin is current or former hegemon of destination.
Source: \citet{head2010erosion}.
\item\label{heg_d} heg\_d:
Dummy if destination is current or former hegemon of origin.
Source: \citet{head2010erosion}.
\item\label{sibling} sibling:
Dummy for origin and destination ever in sibling relationship, i.e. two colonies of the same empire.
Source: \citet{head2010erosion}.
\item\label{comleg_pretrans} comleg\_pretrans:
Dummy if origin and destination share common legal origins before transition.
Source: \citet{la2008economic}.
\item\label{comleg_posttrans} comleg\_posttrans:
Dummy if origin and destination share common legal origins after transition.
Source: \citet{la2008economic}.
\item\label{comrelig} comrelig:
Religious proximity \citep{disdier2007je} is an index calculated by adding the products of the shares of Catholics, Protestants and Muslims in the exporting and importing countries. It is bounded between 0 and 1, and is maximum if the country pair has a religion which (1) comprises a vast majority of the population, and (2) is the same in both countries.
Source of religion shares: \citet{la1999quality}, completed with the CIA world factbook.

\end{enumerate}

We use the variables (\ref{iso3_o}) and (\ref{iso3_d}) to construct $i$ and $j$ indices, respectively.
The current year is set to be the year 2000 in terms of the variable (\ref{year}), and the baseline year is set to be the year 1960.
The variable (\ref{fta_hmr}) is the main dependent variable throughout the analysis.
The first main explanatory variable is the logarithm of the variable (\ref{distw}).
For the second main explanatory variable, we construct the sum of the logarithms of the per-capita GDPs using the pair of the variables (\ref{gdpcap_o}) and (\ref{gdpcap_d}).
We construct up to the fifth power of the continuous variable (\ref{comrelig}).
Then, we construct interactions of all the pairs of the variables from (\ref{colony}) to (\ref{comrelig}) as well as the powers of (\ref{comrelig}).

In addition, we also construct the capital-labor ratios across countries.
The capital data is sourced from the IMF.  
Documentation for the data can be found at https://www.imf.org/external/np/fad/ publicinvestment/pdf/csupdate\_aug19.pdf. We use population from the above data source as a measure for potential labor supply.

\setlength{\baselineskip}{5.7mm}
\bibliography{biblio}
\end{document}